\documentclass{article}
\usepackage[a4paper, margin=1.2in]{geometry}
\usepackage{authblk}

\usepackage{utopia}
\usepackage{url}
\usepackage[hidelinks]{hyperref}
\usepackage[utf8]{inputenc}
\usepackage[small]{caption}
\usepackage{graphicx}
\usepackage{amsmath}
\usepackage{amsthm}
\usepackage{booktabs}
\usepackage{algorithm}
\usepackage{algorithmic}
\usepackage{amssymb,amsfonts}
\usepackage{dsfont}
\newcommand{\1}{\mathds{1}}
\usepackage[T1]{fontenc}
\newcommand{\citet}[1]{\citeauthor{#1}~[\citeyear{#1}]}
\usepackage{tikz}

\usepackage{natbib}
\setcitestyle{authoryear,open={[},close={]}}

\newtheorem{theorem}{Theorem}
\newtheorem{lemma}[theorem]{Lemma}
\newtheorem{proposition}[theorem]{Proposition}
\newtheorem{definition}{Definition}

\title{An Axiom System for Feedback Centralities}

\author{Tomasz Wąs}

\author{Oskar Skibski%
\thanks{\texttt{\{t.was, o.skibski\}@mimuw.edu.pl} \\
This work is an extended version of \citep{Was:Skibski:2021:feedback} that appeared in the Proceedings of the 30th International Joint Conference on Artificial Intelligence (IJCAI-21).
It includes the complete proofs of all theorems which are omitted in the conference publication.
It also appeared as a chapter in the PhD dissertation of the first author~\citep{Was-2022-PhDThesis}.\\
This work was supported by the Polish National Science Centre grant 2018/31/B/ST6/03201.}
}

\affil{University of Warsaw}

\date{}

\begin{document}

\maketitle

\begin{abstract}
In recent years, the axiomatic approach to centrality measures has attracted attention in the literature.
However, most papers propose a collection of axioms dedicated to one or two considered centrality measures.
In result, it is hard to capture the differences and similarities between various measures.
In this paper, we propose an axiom system for four classic feedback centralities: Eigenvector centrality, Katz centrality, Seeley Index, and PageRank.
We prove that each of these four centrality measures can be uniquely characterized with a subset of our axioms.
Our system is the first one in the literature that considers all four feedback centralities.
\end{abstract}

\section{Introduction}
The question how to assess the importance of a node in a network has puzzled scientists for decades~\citep{Boldi:Vigna:2014}.
While the first methods, called \emph{centrality measures}, have been proposed in social science in 1950s, in the last two decades \emph{centrality analysis} has become an actively developed field in computer science, physics and biology~\citep{Brandes:Erlebach:2005,Newman:2005}.
As a result, with a plethora of measures proposed in the literature, the choice of a centrality measure is harder than ever.

\emph{Feedback centralities} form an especially appealing class of centrality measures.
These measures assess the importance of a node recursively by looking at the importance of its neighbors or, in directed graphs, direct predecessors. 
Such an assumption is desirable in many settings, e.g., in citation networks where a citation from a better journal value more~\citep{Pinski:Narin:1976} or in the World Wide Web where a link from a popular website can significantly increase the popularity of our page~\citep{Kleinberg:1999}.

Chronologically, the first feedback centralities were \emph{Seeley index}, proposed by \citet{See-1949-SeeleyIndex}, and \emph{Katz centrality}, introduced by \citet{Katz:1953}.
Arguably, the most classic feedback centrality is \emph{Eigenvector centrality} proposed by \citet{Bonacich:1972}.
In turn, the most popular feedback centrality is \emph{PageRank} designed for Google search engine~\citep{Page:etal:1999}.
These four classic centralities, while based on the same principle, differ in details which leads to diverse results and often opposite conclusions.

In recent years, the axiomatic approach has attracted considerable attention in the literature~\citep{Boldi:Vigna:2014,Bloch:etal:2016}.
This approach serves as a method to build theoretical foundations of centrality measures and to help in making an informed choice of a measure for an application at hand.
In the axiomatic approach, the measure is characterized by a set of simple properties, called \emph{axioms}.
A number of papers use the axiomatic approach to characterize feedback centralities.
Seeley index was considered by \citet{Palacios-Huerta:Volij:2004} and by \citet{Altman:Tennenholtz:2005}.
\citet{Kitti:2016} proposed algebraic axiomatization of Eigenvector centrality.
\citet{Dequiedt:Zenou:2014} and \citet{Was:Skibski:2018:eigenvector} proposed joint axiomatization of Eigenvector and Katz centralities.
Recently, \citet{Was:Skibski:2023:pagerank} proposed an axiomatization of PageRank.

While these results are a step in the right direction, most papers focus only on one or two feedback centralities.
In result, each paper proposes a collection of axioms dedicated for the considered centrality, but poorly fitted to other measures.
As a consequence, these characterizations based on different axioms do not help much in capturing the differences and similarities between various centrality concepts.

\begin{table}
\begin{center}
\setlength{\tabcolsep}{5pt}
\begin{tabular}{c|c|c||c}
General axioms & Node-modification axiom & Borderline axiom & Centrality \\
\hline
LOC, ED, NC  & EC   & CY    & Eigenvector	  \\
LOC, ED, NC  & EC   & BL & Katz          \\
LOC, ED, NC  & EM & CY    & Seeley index  \\
LOC, ED, NC  & EM & BL & PageRank      
\end{tabular}%
\end{center}%
\vspace{-0.5cm}
\caption{Our axiomatic characterizations based on 7 axioms: Locality (LOC), Edge Deletion (ED), Node Combination (NC), Edge Compensation (EC), Edge Multiplication (EM), Cycle (CY) and Baseline (BL).}
\label{table:summary}
\end{table}

In this paper, we propose an axiom system for four classic feedback centralities. 
Our system consists of seven axioms.
Locality, Edge Deletion, Node Combination are general axioms satisfied by all four centralities.
Edge Compensation and Edge Multiplication concern modification of one node and its incident edges.
Finally, Cycle and Baseline specify centralities in simple borderline graphs.
For this set of axioms, we show that each of four feedback centralities is uniquely characterized by a set of 5 axioms: 3 general ones, one one-node-modification axiom and one borderline axiom.
See Table~\ref{table:summary} for a summary.


\section{Preliminaries}
In this paper, we consider directed weighted graphs with node weights and possible self-loops.

\subsection{Graphs}
A \emph{graph} is a quadruple, $G = (V,E,b,c)$, where $V$ is the set of nodes, $E$ is the set of ordered pairs of nodes called edges and $b$ and $c$ are node and edge weights: $b: V \rightarrow \mathbb{R}_{\ge 0}$ and $c: E \rightarrow \mathbb{R}_{>0}$.
We assume that node weights are non-negative and edge weights are positive.
The set of all possible graphs is denoted by $\mathcal{G}$.

For a graph $G$, the \emph{adjacency matrix} is defined as follows: $A = (a_{u,v})_{u,v \in V}$, where $a_{u,v} = c(v,u)$ if $(v,u) \in E$ and $a_{u,v} = 0$, otherwise.
A real value $r$ is an \emph{eigenvalue} of a matrix $A$ if there exists a non-zero vector $x \in \mathbb{R}^V$ such that $A x = r x$; such vector $x$ is called an \emph{eigenvector}.
The \emph{principal eigenvalue}, denoted by $\lambda$, is the largest eigenvalue.

An edge $(u,v)$ is an \emph{outgoing} edge for node $u$ and an \emph{incoming} edge for node $v$.
For $v$, the set of its incoming edges is denoted by $\Gamma^-_v(G)$ and outgoing edges by $\Gamma^+_v(G)$.
The total weight of outgoing edges, called the \emph{out-degree}, is denoted by $\deg^+_v(G)$: $\deg^+_v(G) = \sum_{e \in \Gamma^+_v(G)} c(e)$.
For any $x$, graph $G$ is \emph{$x$-out-regular} if the out-degree of every node equals $x$: $\deg^+_v(G) = x$ for every $v \in V$.
A graph is \emph{out-regular} if it is $x$-out-regular for some~$x$.

A \emph{walk} is a sequence of nodes $\omega = (\omega(0), \dots, \omega(k))$ such that every two consecutive nodes are connected by an edge: $(\omega(i), \omega(i+1)) \in E$ for every $i \in \{0,\dots,k-1\}$ and $k \ge 1$.
The walk is said to \emph{start} at $\omega(0)$ and \emph{end} at $\omega(k)$ and the \emph{length}, $|\omega|$, of a walk is defined to be $k$.
The set of all walks of length $k$ will be denoted by $\Omega_k(G)$.
If there exists a walk that starts in $u$ and ends in $v$, then $u$ is called a \emph{predecessor} of $v$ and $v$ is called a \emph{successor} of $u$.
If the length of this walk is one, i.e., $(u,v) \in E$, then nodes are \emph{direct} predecessors/successors.
For node $v$, the set of predecessors of $v$ is denoted by $P(v)$ and the set of successors of $v$ by $S(v)$.
The graph is \emph{strongly connected} if there exists a walk between every two nodes, i.e., if $S(v) = V$ for every node $v \in V$.

A strongly connected graph such that every node has exactly one outgoing edge is called a \emph{cycle graph}.

Let us introduce some shorthand notation that we will use throughout the paper.
For an arbitrary function $f: A \rightarrow X$ and a subset $B \subseteq A$, function $f$ with the domain restricted to set $B$ will be denoted by $f_B$ and to set $A \setminus B$: by $f_{-B}$.
If $B$ contains one element, i.e., $B = \{a\}$, we will skip parenthesis and simply write $f_b$ and $f_{-b}$.
Also, for a constant $x$, we define $x \cdot f$ as follows: $(x \cdot f)(a) = x \cdot f(a)$ for every $a \in A$.
Furthermore, for two functions with possibly different domains, $f: A \rightarrow X, f': B \rightarrow X$, we define $f+f': A \cup B \rightarrow X$ as follows: 
$(f+f')(a) = f(a)+f'(a)$ if $a \in A \cap B$, $(f+f')(a) = f(a)$ if $a \in A \setminus B$ and $(f+f')(a) = f'(a)$ if $a \in B \setminus A$ for every $a \in A \cup B$.
In particular, $(b_{-v}+2b_{v})$ are node weights obtained from $b$ by doubling weight of node $v$.

For two graphs, $G = (V,E,b,c), G'=(V',E',b',c')$ with $V \cap V' = \emptyset$, their \emph{sum} $G+G'$ is defined as follows: $G+G' = (V \cup V', E \cup E', b+b', c+c')$.

\subsection{Feedback centralities}
A \emph{centrality measure} is a function $F$ that given a graph $G = (V,E,b,c)$ and a node $v \in V$ returns a real value, denoted by $F_v(G)$.
This value, called a \emph{centrality of a node}, is assumed to be non-negative and represents the importance of node $v$ in graph $G$.

The class of \emph{feedback centralities} aims to assess the importance of a node by looking at the importance of its direct predecessors.
We consider four classic feedback centralities.

\vspace{0.1cm} \noindent 
\textbf{Eigenvector centrality:} 
According to Eigenvector centrality~\citep{Bonacich:1972}, the importance of a node is proportional to the total importance of its direct predecessors:
\begin{equation}\label{eq:rec:ev}
EV_v(G) = \frac{1}{\lambda} \sum_{(u,v) \in \Gamma^-_v(G)} c(u,v) \cdot EV_u(G).
\end{equation}
This system of recursive equations have multiple solutions.
Hence, some additional normalization condition is usually assumed to make a solution unique (e.g., the sum of centralities of all nodes is assumed to be $1$ or $|V|$).
In this paper, we use a normalization more consistent with other feedback centralities---we will discuss it in the next section.
The Eigenvector centrality is usually defined only for strongly connected graphs. 
We relax this assumption by considering also sums of strongly connected graphs with the same principal eigenvalue.%
\footnote{Note that we cannot allow for the sum of arbitrary graphs. It is because if one of the graphs has a smaller principle eigenvalue, then Eigenvector centralities of all its nodes would be zero.}
We denote the class of all such graphs by $\mathcal{G}^{EV}$.

\vspace{0.1cm} \noindent 
\textbf{Katz centrality:} 
\citet{Katz:1953} proposed an alternative to Eigenvector centrality that adds a basic importance to each node.
This shifts the emphasis from the total importance of its direct predecessors to their number.
Formally, for a decay factor $a \in \mathbb{R}_{\ge 0}$, Katz centrality is defined as follows:
\begin{equation}\label{eq:rec:katz}
K^a_v(G) = a \left( \sum_{(u,v) \in \Gamma^-_v(G)} c(u,v) \cdot K^a_u(G) \right) + b(v).
\end{equation}
For a fixed $a$, Katz centrality is uniquely defined for graphs with $\lambda < 1/a$.
We denote the class of such graphs by $\mathcal{G}^{K(a)}$.

\vspace{0.1cm} \noindent 
\textbf{Seeley index:}
In Eigenvector and Katz centralities, the whole importance of a node is ``copied'' to all its direct successors.
In turn, in Seeley index~\citep{See-1949-SeeleyIndex} (which is also known as \emph{Katz prestige}~\citep{Jac-2008-Networks} and \emph{simplified PageRank}~\citep{Page:etal:1999}),
a node splits its importance equally among its successors.
Hence, the importance of predecessors is divided by their out-degree.
Formally, Seeley index is defined as follows:
\begin{equation}\label{eq:rec:si}
SI_v(G) = \sum_{(u,v) \in \Gamma^-_v(G)} \frac{c(u,v)}{\deg^+_u(G)} SI_u(G).
\end{equation}
Similarly to Eigenvector centrality, this system of equations does not imply a unique solution.
We will discuss our normalization condition in the next section.
Seeley index is usually defined only for strongly connected graphs.
In our paper, we relax this assumption and consider sums of strongly connected graphs; we denote the class of all such graphs by $\mathcal{G}^{SI}$.

\vspace{0.1cm} \noindent 
\textbf{PageRank:}
\citet{Page:etal:1999} was proposed to modify Seeley index by adding a basic importance to each node.
In this way, for a decay factor $a \in [0,1)$, PageRank is uniquely defined for all graphs as follows:
\begin{equation}\label{eq:rec:pr}
PR^a_v(G) = a \! \left( \sum_{(u,v) \in \Gamma^-_v\!(G)} \frac{c(u,v)}{\deg^+_u(G)} PR^a_u(G) \right) \!+ b(v).
\end{equation}

\subsection{Walk interpretations of feedback centralities}

In this section we show how PageRank's random-walk interpretation from \citep{Was:Skibski:2023:pagerank} can be extended to work for all four feedback centralities.
These interpretations result in unique definitions of Eigenvector centrality and Seeley index consistent with Katz centrality and PageRank.

Consider a spread of some entity (money, virus, information, probability of visit etc.) through a network.
At the beginning, at time $t=0$, each node has some initial amount equal to its node weight.
Then, in each step the whole entity is multiplied by a scalar $a$, which is a parameter of the process, and then moved to direct successors.
Here, we consider two variants of the process: \emph{distributed} and \emph{parallel}.
In the distributed process, the entity is spread among the outgoing edges.
In the parallel process, the entity is duplicated and moved along all edges.
Specifically, when the entity is moved along an edge $(u,v)$ it is either
\begin{itemize}
\item multiplied by $c(u,v)/\deg_u^+(G)$ (distributed process), or
\item multiplied by $c(u,v)$ (parallel process).
\end{itemize}

As an example, consider a surfer on the World Wide Web. 
Node weights correspond to the probability that the surfer starts surfing from a specific page and weights of edges represent the number of links from one page to another.
A link is chosen by the surfer uniformly at random which means the process is distributed.
Finally, parameter $a$ is the probability that the surfer stops surfing altogether.
In result, the ``entity'' is the probability that the surfer visits a specific page at some time.

Consider the distributed process with parameter $a$. 
The amount of entity in node $v$ at time $t$ equals:
\[ p^a_{v,G}(t) = \sum_{\omega \in \Omega_t(G) : \omega(t)=v} b(\omega(0)) \cdot \prod_{i=0}^{t-1} \frac{a \cdot c(\omega(i),\omega(i+1))}{\deg^+_{\omega(i)}(G)}. \]
Now, if $a < 1$, PageRank is the total amount of entity in node $v$ through the whole process:
\begin{equation}\label{eq:walk:pr}
PR^a_v(G) = \sum_{t=0}^\infty p^a_{v,G}(t)
\end{equation}
We know that this sum converges, because $a < 1$.
Now, if $a=1$, the sum does not converge, as the amount of entity in the network does not change over time.
In such a case we can look at the average amount of entity in node $v$.
In this way we get the definition of Seeley index:
\begin{equation}\label{eq:walk:si}
SI_v(G) = \lim_{T \rightarrow \infty} \sum_{t=0}^T \frac{p^1_{v,G}(t)}{T}
\end{equation}
Hence, Seeley index is the stationary distribution of the process multiplied by the sum of node weights.

Let us turn our attention to the parallel process with parameter $a$.
The amount of entity in node $v$ at time $t$ equals:
\[ w^a_{v,G}(t) = \!\! \sum_{\omega \in \Omega_t(G) : \omega(t)=v} \!\! b(\omega(0)) \cdot \prod_{i=0}^{t-1} \left(a \cdot c(\omega(i),\omega(i+1))\right). \]
Now, if $a < 1/\lambda$, Katz centrality is the total amount of entity in node $v$ through the whole process:
\begin{equation}\label{eq:walk:katz}
K^a_v(G) = \sum_{t=0}^\infty w^a_{v,G}(t)
\end{equation}
This sum converges, because $a < 1/\lambda$.
Now, if $a=1/\lambda$, the sum does not converge.
In such a case, Eigenvector centrality can be obtained as the average amount of entity in node $v$:
\begin{equation}\label{eq:walk:ev}
EV_v(G) = \lim_{T \rightarrow \infty} \sum_{t=0}^T \frac{w^{1/\lambda}_{v,G}(t)}{T}
\end{equation}
In the appendix we show that measures defined in Equations~\eqref{eq:walk:pr}--\eqref{eq:walk:ev} indeed satisfy recursive equations from Equations~\eqref{eq:rec:ev}--\eqref{eq:rec:pr} (Propositions~\ref{proposition:walk:rec:ev}--\ref{proposition:walk:rec:si}).
From now on, we will use Equations~\eqref{eq:walk:pr}--\eqref{eq:walk:ev} as the definitions of all four centralities.


\section{Axioms}

In this section, we present seven axioms used in our axiomatic characterization.

All centrality measures except for PageRank are defined only for a subclass of all graphs.
Hence, for them we will consider restricted versions of our axioms.
Specifically, an axiom restricted to class $\mathcal{G}^*$ is obtained by adding an assumption that all graphs appearing in the axiom statement belong to $\mathcal{G}^*$.
In this way, we obtain a weaker version of the axiom.

Most of our axioms are invariant axioms. 
They identify simple graph operations that do not affect centralities of all or most nodes in a graph.
The last two axioms serve as a borderline: they specify centralities in very simple graphs.
We use three axioms proposed in the axiomatization of PageRank by \citet{Was:Skibski:2023:pagerank}.
Instead of the remaining three axioms, we use Locality and Node Combination which are more meaningful for
considered classes of graphs.

Before we proceed, let us introduce an operation of \emph{proportional combining} of two nodes used in one of the axioms.
Proportional combining differs from a simple merging of nodes as it preserves the significance of outgoing edges.
Take a centrality measure $F$, graph $G = (V,E,b,c)$ and two nodes $u,w \in V$.
Graph $C^F_{u \rightarrow w}(G)$ is a graph obtained in two steps:
\begin{itemize}
\item scaling weights of outgoing edges of $u$ and $w$ proportionally to their centralities: multiplying weights of outgoing edges of $u$ by $F_u(G)/(F_u(G)+F_w(G))$ and $w$ by $F_w(G)/(F_u(G) + F_w(G))$; and 
\item merging node $u$ into node $w$, i.e., deleting node $u$, transferring its incoming and outgoing edges to node $w$ and adding the weight of node $u$ to node $w$.
\end{itemize}
See Figure~\ref{figure:combining} for an illustration.

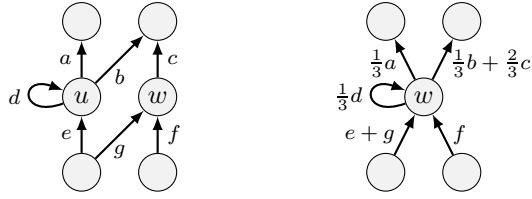
\begin{figure}[t]
\centering
\begin{tikzpicture}[x=1cm,y=1cm]
  \tikzset{     
    e4c node/.style={fill = black!05, circle,draw,minimum size=0.5cm,inner sep=0}, 
    e4c edge/.style={ -latex, left,font=\footnotesize}
  }
  
  \def\x{0}
  \node[e4c node] (1) at (\x+0.00, 2.00) {}; 
  \node[e4c node] (2) at (\x+1.00, 2.00) {}; 
  \node[e4c node] (3) at (\x+0.00, 1.00) {$u$}; 
  \node[e4c node] (4) at (\x+1.00, 1.00) {$w$}; 
  \node[e4c node] (5) at (\x+0.00, 0.00) {}; 
  \node[e4c node] (6) at (\x+1.00, 0.00) {}; 

  \path[->,draw,thick]
  (3) edge[e4c edge] node {$a$} (1)
  (3) edge[e4c edge, below] node {$b$} (2)
  (4) edge[e4c edge, right] node {$c$} (2)
  (3) edge[e4c edge, out = 200, in = 160, looseness=9] node {$d$} (3)
  (5) edge[e4c edge] node {$e$} (3)
  (6) edge[e4c edge, right] node {$f$} (4)
  (5) edge[e4c edge, below] node {$g$} (4)
  ;

  \def\x{4}
  \node[e4c node] (1) at (\x+0.00, 2.00) {}; 
  \node[e4c node] (2) at (\x+1.00, 2.00) {}; 
  \node[e4c node] (3) at (\x+0.50, 1.00) {$w$}; 
  \node[e4c node] (5) at (\x+0.00, 0.00) {}; 
  \node[e4c node] (6) at (\x+1.00, 0.00) {}; 

  \path[->,draw,thick]
  (3) edge[e4c edge, left, pos=0.4] node {$\frac{1}{3}a$} (1)
  (3) edge[e4c edge, right, pos=0.4] node {$\frac{1}{3}b+\frac{2}{3}c$} (2)
  (3) edge[e4c edge, out = 200, in = 160, looseness=9] node {$\frac{1}{3}d$} (3)
  (5) edge[e4c edge] node {$e+g$} (3)
  (6) edge[e4c edge, right] node {$f$} (3)
  ;

\end{tikzpicture}
\caption{Graph $G$ (on the left) and the corresponding graph $C_{u\rightarrow w}^F(G)$ (on the right) assuming $F_u(G) = 1$ and $F_w(G) = 2$.} 
\label{figure:combining}
\end{figure}

We are now ready to present the first three axioms satisfied by all feedback centralities.

\begin{quote}\textit{\noindent
\textbf{Locality (LOC)}:
For every graph $G=(V,E,b,c)$ and graph $G'=(V',E',b',c')$ s.t. $V \cap V' = \emptyset$:
\[	F_v(G+G')=F_v(G) \quad \mbox{for } v \in V. \]
}
\end{quote}

\begin{quote}\textit{\noindent
\textbf{Edge Deletion (ED)}:
For every graph $G=(V,E,b,c)$ and edge $(u,w) \in E$:
\[ F_v(V,E \setminus \{(u,w)\},b,c)=F_v(G) \quad \mbox{for } v \in V \setminus S(u). \]
}
\end{quote}

\begin{quote}\textit{\noindent
\textbf{Node Combination (NC)}:
For every graph $G=(V,E,b,c)$ and nodes $u,w \in V$ s.t. $\deg^+_u(G)=\deg^+_w(G) = \deg^+_s(G)$ for every $s \in S(u)\cup S(w)$:
\[	F_v(C^F_{u \rightarrow w}(G))=F_v(G) \quad \mbox{for } v \in V \setminus \{u,w\} \]
and $F_w(C^F_{u \rightarrow w}(G)) = F_u(G) + F_w(G)$.
}
\end{quote}

Locality and Edge Deletion are standard axioms from the literature.
Locality, proposed in \citep{Skibski:etal:2019:attachment}, states that the centrality of a node depends solely on the part of the graph a node is connected to.
In other words, removing part of the graph not connected to a node does not affect its centrality.
Edge Deletion, proposed in \citep{Was:Skibski:2023:pagerank} for PageRank, states that removing an edge from the graph does not affect nodes which cannot be reached from the start of the edge. 
Node Combination is a new axiom.
Assume two nodes $u,w \in V$ and their successors have the same out-degree, but possibly different centralities. 
Node Combination states that in a graph obtained from proportional combining of $u$ into $w$, the centrality of $w$ is the sum of centralities of both nodes and centralities of other nodes do not change.
This property is characteristic for feedback centralities which associate a benefit from an incoming edge with the importance of a node this edge comes from.
We note that PageRank, Seeley index and Katz centrality satisfy also the axiom without the assumption about equal out-degrees of successors. 
However, it is necessary for Eigenvector centrality.

\begin{figure}[t]
  \centering
  \begin{tikzpicture}[x=1cm,y=1cm] 
    \tikzset{     
      e4c node/.style={fill = black!05, circle,draw,minimum size=0.5cm,inner sep=0}, 
      e4c edge/.style={ -latex, left,font=\footnotesize}
    }
  
    \def\x{0}
    \node[e4c node] (1) at (\x+0.00, 2.00) {}; 
    \node[e4c node] (2) at (\x+1.00, 2.00) {}; 
    \node[e4c node] (3) at (\x+0.00, 1.00) {$u$}; 
    \node[e4c node] (4) at (\x+1.00, 1.00) {$w$}; 
    \node[e4c node] (5) at (\x+0.00, 0.00) {}; 
    \node[e4c node] (6) at (\x+1.00, 0.00) {}; 
  
    \path[->,draw,thick]
    (3) edge[e4c edge] node {$xa$} (1)
    (3) edge[e4c edge, below] node {$xb$} (2)
    (4) edge[e4c edge, right] node {$c$} (2)
    (3) edge[e4c edge, out = 200, in = 160, looseness=9] node {$xd$} (3)
    (5) edge[e4c edge] node {$e$} (3)
    (6) edge[e4c edge, right] node {$f$} (4)
    (5) edge[e4c edge, below] node {$g$} (4)
    ;
  
    \def\x{4}
    \node[e4c node] (1) at (\x+0.00, 2.00) {}; 
    \node[e4c node] (2) at (\x+1.00, 2.00) {}; 
    \node[e4c node] (3) at (\x+0.00, 1.00) {$u$}; 
    \node[e4c node] (4) at (\x+1.00, 1.00) {$w$}; 
    \node[e4c node] (5) at (\x+0.00, 0.00) {}; 
    \node[e4c node] (6) at (\x+1.00, 0.00) {}; 
  
    \path[->,draw,thick]
    (3) edge[e4c edge] node {$xa$} (1)
    (3) edge[e4c edge, below] node {$xb$} (2)
    (4) edge[e4c edge, right] node {$c$} (2)
    (3) edge[e4c edge, out = 200, in = 160, looseness=9] node {$d$} (3)
    (5) edge[e4c edge] node {$e/x$} (3)
    (6) edge[e4c edge, right] node {$f$} (4)
    (5) edge[e4c edge, below] node {$g$} (4)
    ;

  \end{tikzpicture}
  \caption{Graphs considered in Edge Multiplication (on the left) and Edge Compensation (on the right) obtained from $G$ from Figure~\ref{figure:combining}.} 
  \label{figure:em_ec}
  \end{figure}
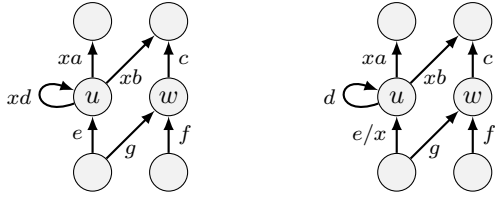

Our next two axioms concern a modification of one node: its weight and weights of its incident edges.

\begin{quote}\textit{\noindent
\textbf{Edge Multiplication (EM)}:
For every graph $G=(V,E,b,c)$, node $u \in V$, constant $x>0$:
\[ F_v(V,E,b,c_{-\Gamma^+_u(G)}+x\cdot c_{\Gamma^+_u(G)})=F_v(G) \mbox{ \ for } v \in V. \]
}
\end{quote}

\begin{quote}\textit{\noindent
\textbf{Edge Compensation (EC)}:
For every graph $G=(V,E,b,c)$, node $u \in V$, constant $x>0$, 
$b' = b_{-u}+b_u/x$ and $c' = c_{-\Gamma^{\pm}_u(G) \setminus \{(u,u)\}} + c_{\Gamma^-_u(G) \setminus \{(u,u)\})}/ x + c_{\Gamma^+_u(G) \setminus \{(u,u)\}} \cdot x$:
\[ F_v(V,E,b',c')=F_v(G) \quad \mbox{for } v \in V \setminus \{u\} \]
and $F_u(V,E,b',c') = F_u(G)/x$.
}
\end{quote}

Edge Multiplication, proposed in \citep{Was:Skibski:2023:pagerank}, states that multiplying weights of all outgoing edges of node $u$ by some constant does not affect centralities.
This means that the absolute weight of edges does not matter, as long as the proportion to other outgoing edges of a node is the same. 
This property is satisfied by both PageRank and Seeley index.
However, it is not satisfied by Eigenvector and Katz centralities, as it increase the importance of modified edges $x$ times.
For them, we propose a similar axiom: Edge Compensation.
In this axiom, not only weights of outgoing edges of $u$ are multiplied by a constant, but at the same time weights of incoming edges and the node itself are divided by the same constant.
Edge Compensation states that this operation decreases the importance of $u$ $x$ times, but at the same time does not affect the importance of other nodes.
See Figure~\ref{figure:em_ec} for an illustration.

Finally, the last two axioms concern simple borderline cases.

\begin{quote}\textit{\noindent
\textbf{Baseline (BL)}:
For every graph $G = (V,E,b,c)$ and an isolated node $v \in V$ it holds $F_v(G) = b(v)$.
}
\end{quote}

\begin{quote}\textit{\noindent
\textbf{Cycle (CY)}:
For every out-regular cycle graph $G=(V,E,b,c)$ it holds 
\[
  F_v(G) = \sum_{u \in V} b(u) / |V|\quad \mbox{for every } v \in V.
\]
}
\end{quote}

Baseline, proposed in \citep{Was:Skibski:2023:pagerank}, states that the centrality of a node with no incident edges is equal to its baseline importance: its node weight.
Baseline is satisfied by PageRank and Katz centrality.
However, it does not make sense for strongly connected graphs.
Cycle, proposed as the borderline case for Eigenvector centrality and Seeley index, considers the simplest strongly connected graph: a cycle.
Specifically, if weight of all edges in a cycle graph are equal, then centralities of all nodes are also equal. 
Moreover, Cycle normalizes the sum of centralities to be equal to the sum of node weights.
See Figure~\ref{figure:bl_cy} for an illustration.

\begin{figure}[t]
\centering
\begin{tikzpicture}[x=5cm,y=5cm] 
  \tikzset{     
    e4c node/.style={fill = black!05, circle,draw,minimum size=0.5cm,inner sep=0}, 
    e4c edge/.style={ -latex, above,font=\footnotesize}
  }
  \node[e4c node] (1) at (0.08, 1.19) {$v$}; 
  \node[e4c node] (2) at (0.30, 1.19) {}; 
  \node[e4c node] (3) at (0.38, 1.00) {}; 
  \node[e4c node] (4) at (0.19, 0.86) {}; 
  \node[e4c node] (5) at (0.00, 1.00) {}; 
  \node[e4c node] (6) at (-0.65, 1.14) {}; 
  \node[e4c node] (7) at (-0.45, 1.00) {}; 
  \node[e4c node] (8) at (-0.69, 0.93) {}; 
  \node[e4c node] (9) at (-0.40, 1.20) {$v$};

  \path[->,draw,thick]
  (1) edge[e4c edge,above] node {$x$} (2)
  (2) edge[e4c edge,right,pos=0.4] node {$x$} (3)
  (3) edge[e4c edge,below,pos=0.3] node {$x$} (4)
  (4) edge[e4c edge,below,pos=0.6] node {$x$} (5)
  (5) edge[e4c edge,left,pos=0.6] node {$x$} (1)
  (6) edge[e4c edge,above,pos=0.6] (7)
  (7) edge[e4c edge,below,pos=0.4] (8)
  (6) edge[e4c edge,left,pos=0.4] (8)
  ;
\end{tikzpicture}

\caption{Graphs considered in Baseline (on the left) and Cycle (on the right).} 
\label{figure:bl_cy}
\end{figure}
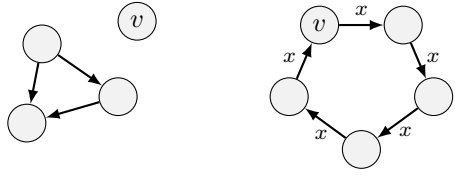

As we will show, these seven axioms are enough to obtain the axiomatizations of all four feedback centralities.

\section{Main results}
In  this section, we present our main results:
We show that if a centrality satisfies three general axioms (LOC, ED and NC), one of the one-node-modification axioms (EM or EC) and one of the borderline axioms (BL or CY) then it must be one of the four feedback centralities.

The full proofs can be found in the supplementary materials. 
In this section, we present the main ideas behind them.

First, consider Seeley index and Eigenvector centrality.

\begin{theorem}
\label{theorem:si}
A centrality measure defined on $\mathcal{G}^{SI}$ satisfies LOC, ED, NC, EM and CY if and only if it is Seeley index (Equation~\eqref{eq:walk:si}).
\end{theorem}

\begin{theorem}
\label{theorem:ev}
A centrality measure defined on $\mathcal{G}^{EV}$ satisfies LOC, ED, NC, EC and CY if and only if it is Eigenvector centrality (Equation~\eqref{eq:walk:ev}).
\end{theorem}

We begin by showing that Seeley index and Eigenvector centrality are equal for strongly connected out-regular graphs (Lemma~\ref{lemma:ev-si:equal-on-regular}). 
Then, we prove that both centralities satisfy the corresponding axioms (Lemmas~\ref{lemma:axioms:si}--\ref{lemma:axioms:ev}).
Hence, it remains to prove the uniqueness of both axiomatizations.

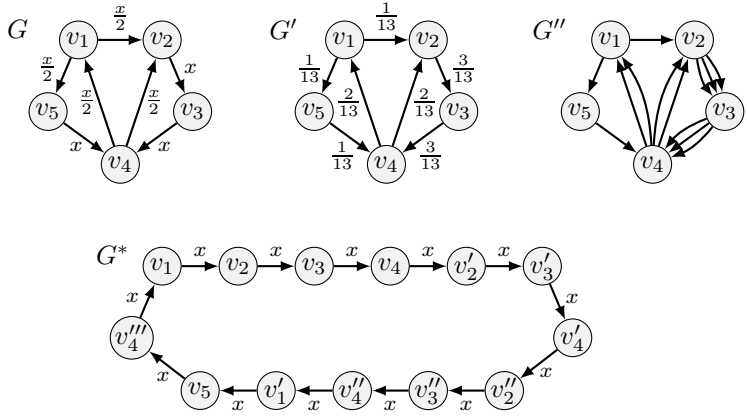
\begin{figure*}[t]
  \centering
  \begin{tikzpicture}[x=5cm,y=5cm] 
    \tikzset{     
      e4c node/.style={fill = black!05, circle,draw,minimum size=0.5cm,inner sep=0}, 
      e4c edge/.style={ -latex, above,font=\footnotesize}
    }
  
    \def\x{0}
    \node[e4c node] (1) at (\x+0.08, 1.19) {$v_1$}; 
    \node[e4c node] (2) at (\x+0.30, 1.19) {$v_2$}; 
    \node[e4c node] (3) at (\x+0.38, 1.00) {$v_3$};
    \node[e4c node] (4) at (\x+0.19, 0.86) {$v_4$}; 
    \node[e4c node] (5) at (\x+0.00, 1.00) {$v_5$}; 
    \node (x) at (\x-0.08, 1.22) {$G$}; 
  
    \path[->,draw,thick]
    (4) edge[e4c edge,left] node[xshift=0.08cm] {$\frac{x}{2}$} (1)
    (4) edge[e4c edge,right] node[xshift=-0.08cm] {$\frac{x}{2}$} (2)
    (1) edge[e4c edge] node {$\frac{x}{2}$} (2)
    (1) edge[e4c edge,left,pos=0.3] node {$\frac{x}{2}$} (5)
    (2) edge[e4c edge,right,pos=0.3] node {$x$} (3)
    (5) edge[e4c edge,below,pos=0.3] node {$x$} (4)
    (3) edge[e4c edge,below,pos=0.3] node {$x$} (4)
    ;
    
    \def\x{0.7}
    \node[e4c node] (1) at (\x+0.08, 1.19) {$v_1$}; 
    \node[e4c node] (2) at (\x+0.30, 1.19) {$v_2$}; 
    \node[e4c node] (3) at (\x+0.38, 1.00) {$v_3$};
    \node[e4c node] (4) at (\x+0.19, 0.86) {$v_4$}; 
    \node[e4c node] (5) at (\x+0.00, 1.00) {$v_5$}; 
    \node (x) at (\x-0.08, 1.22) {$G'$}; 
  
    \path[->,draw,thick]
    (4) edge[e4c edge,left] node[xshift=0.08cm] {$\frac{2}{13}$} (1)
    (4) edge[e4c edge,right] node[xshift=-0.08cm] {$\frac{2}{13}$} (2)
    (1) edge[e4c edge] node {$\frac{1}{13}$} (2)
    (1) edge[e4c edge,left,pos=0.3] node {$\frac{1}{13}$} (5)
    (2) edge[e4c edge,right,pos=0.3] node {$\frac{3}{13}$} (3)
    (5) edge[e4c edge,below,pos=0.3] node {$\frac{1}{13}$} (4)
    (3) edge[e4c edge,below,pos=0.3] node {$\frac{3}{13}$} (4)
    ;
    
    \def\x{1.4}
    \node[e4c node] (1) at (\x+0.08, 1.19) {$v_1$}; 
    \node[e4c node] (2) at (\x+0.30, 1.19) {$v_2$}; 
    \node[e4c node] (3) at (\x+0.38, 1.00) {$v_3$};
    \node[e4c node] (4) at (\x+0.19, 0.86) {$v_4$}; 
    \node[e4c node] (5) at (\x+0.00, 1.00) {$v_5$}; 
    \node (x) at (\x-0.08, 1.22) {$G''$}; 
  
    \path[->,draw,thick]
    (4) edge[e4c edge]  (1)
    (4) edge[e4c edge, bend right=15]  (1)
    (4) edge[e4c edge, bend left=15]  (2)
    (4) edge[e4c edge]  (2)
    (1) edge[e4c edge]  (2)
    (1) edge[e4c edge]  (5)
    (2) edge[e4c edge]  (3)
    (2) edge[e4c edge, bend left=15]  (3)
    (2) edge[e4c edge, bend right=15]  (3)
    (5) edge[e4c edge]  (4)
    (3) edge[e4c edge]  (4)
    (3) edge[e4c edge, bend left=15]  (4)
    (3) edge[e4c edge, bend right=15]  (4)
    ;
    
    \def\x{0.3}  
    \def\y{0.25}
    \def\z{-0.6}
    \node[e4c node] (1) at (\x+0*\y, 1.19+\z) {$v_1$}; 
    \node[e4c node] (2) at (\x+1*\y, 1.19+\z) {$v_2$}; 
    \node[e4c node] (3) at (\x+2*\y, 1.19+\z) {$v_3$}; 
    \node[e4c node] (4) at (\x+3*\y, 1.19+\z) {$v_4$}; 
    \node[e4c node] (5) at (\x+4*\y, 1.19+\z) {$v'_2$}; 
    \node[e4c node] (6) at (\x+5*\y, 1.19+\z) {$v'_3$}; 
    \node[e4c node] (7) at (\x+5*\y+0.08, 1.00+\z) {$v'_4$}; 
    \node[e4c node] (8) at (\x+5*\y-0.10, 0.86+\z) {$v''_2$}; 
    \node[e4c node] (9) at (\x+4*\y-0.10, 0.86+\z) {$v''_3$}; 
    \node[e4c node] (10) at (\x+3*\y-0.10, 0.86+\z) {$v''_4$}; 
    \node[e4c node] (11) at (\x+2*\y-0.10, 0.86+\z) {$v'_1$}; 
    \node[e4c node] (12) at (\x+1*\y-0.10, 0.86+\z) {$v_5$}; 
    \node[e4c node] (13) at (\x+0*\y-0.08, 1.00+\z) {$v'''_4$}; 
    \node (x) at (\x-0.13, 1.22+\z) {$G^*$};   
  
    \path[->,draw,thick]
    (1) edge[e4c edge] node {$x$} (2)
    (2) edge[e4c edge] node {$x$} (3)
    (3) edge[e4c edge] node {$x$} (4)
    (4) edge[e4c edge] node {$x$} (5)
    (5) edge[e4c edge] node {$x$} (6)
    (6) edge[e4c edge,right,pos=0.4] node {$x$} (7)
    (7) edge[e4c edge,below,pos=0.3] node {$x$} (8)
    (8) edge[e4c edge,below] node {$x$} (9)
    (9) edge[e4c edge,below] node {$x$} (10)
    (10) edge[e4c edge,below] node {$x$} (11)
    (11) edge[e4c edge,below] node {$x$} (12)
    (12) edge[e4c edge,below,pos=0.6] node {$x$} (13)
    (13) edge[e4c edge,left,pos=0.6] node {$x$} (1)
    ;

  \end{tikzpicture}
  \caption{The construction of a cycle graph $G^*$ from which $G$ can be obtained by proportional combining.} 
  \label{figure:proof_1}
  \end{figure*}

Let $F$ be a centrality measure defined on $\mathcal{G}^{EV}$ or $\mathcal{G}^{SI}$ that satisfied LOC, ED, NC and CY.
First, we show that $F$ is equal to both centralities for strongly connected out-regular graphs.
We divide it into two steps:
\begin{itemize}
\item First, we consider strongly connected out-regular graphs such that the proportion of weights of every two edges is rational (Lemma~\ref{lemma:ev-si:super-regular-graphs})
\item Then, we consider arbitrary strongly connected out-regular graphs (Lemma~\ref{lemma:ev-si:regular-graphs})
\end{itemize}
These proofs are based on the following observation:
for every graph $G$ that satisfies assumptions of Lemma~\ref{lemma:ev-si:super-regular-graphs} it is possible to construct a cycle graph $G^*$, from which $G$ can be obtained using proportional combining.
This implies that if some centrality measure satisfies CY and NC, then it is uniquely defined on graph $G$.

We will present this construction on the graph $G$ from Figure~\ref{figure:proof_1}. 
Node weights are arbitrary, but we assume they sum up to $1$.
This graph is $x$-out-regular and for every two edges the proportion of their weights is rational.
In such graphs, it can be shown that Seeley index of any node is rational; in this graph we have: $SI_{v_1}(G) = \frac{2}{13}$, $SI_{v_2}(G) = SI_{v_3}(G) = \frac{3}{13}$, $SI_{v_4}(G) = \frac{4}{13}$ and $SI_{v_5}(G) = \frac{1}{13}$.

Now, let us define an \emph{impact} of an edge $(u,v)$, denoted by $I(u,v)$, as Seeley index of node $u$ multiplied by $c(u,v)/\deg^+_u(G)$.
Impacts of edges are depicted in graph $G'$ in Figure~\ref{figure:proof_1}.
For example, $I(v_4,v_1) = K_{v_4}(G) \cdot 1/2 = 2/13$.
Clearly, the total impact of outgoing edges of every node is equal to its Seeley index. 
Also, from the Seeley index recursive equation (Equation~\eqref{eq:rec:si}), we see that the total impact of incoming edges of every node is also equal to its Seeley index.
We will use this fact later on.
What is also important impacts are rational; let $N$ be the least common multiple of all denominators of edge impacts. 
For graph $G$ we have $N=13$.

Based on the notion of impact we create a \emph{multigraph} $G''$ (a graph in which multiple edges exist between a pair of nodes).
This multigraph is obtained by replacing edge $(u,v)$ from the original graph by $N \cdot I(u,v)$ edges.
For example, in our sample graph edge $(v_4,v_1)$ with impact $I(v_4,v_1) = 2/13$ is replaced by two edges from $v_4$ to $v_1$.
Now, the key observation is that every node in the multigraph $G''$ has the same number of incoming and outgoing edges (as the total impact of its incoming edges equals the total impact of its outgoing edges).
Hence, from Euler's theorem, we know that in $G''$ there exists an Euler cycle---a walk that visits every edge exactly once and is a cycle: starts and ends in the same node.

Graph $G^*$ in Figure~\ref{figure:proof_1} is a cycle graph corresponding to some Euler cycle in graph $G''$. 
Here, nodes $v_i, v'_i, \dots$ correspond to node $v_i$ in $G''$.
We define node weights in a way that nodes corresponding to node $v_i$ sum up to $b(v_i)$.
The construction of the cycle graph is complete.

Let us determine the number of nodes in $G^*$.
In graph $G''$ node $v_i$ has $N \cdot SI_{v_i}(G)$ outgoing edges.
Hence, in $G^*$ there are $N \cdot SI_{v_i}(G)$ nodes corresponding to node $v_i$.
Since the sum of Seeley index of all nodes in $G$ is equal to the sum of node weights, i.e., $1$, we get that there are $N$ nodes in graph $G^*$.
In our example, there are indeed $13$ nodes in $G^*$ and $4$ nodes correspond to node $v_4$.

Now, from CY, we know that according to $F$ every node in $G^*$ has the same centrality equal to $1/N$.
It is easy to verify that if we merge using proportional combining for every node $v_i$ all nodes corresponding to $v_i$ we will obtain graph $G$.
Based on NC this implies that node $v_i$ has centrality in $G$ equal to $N \cdot SI_{v_i}(G) \cdot 1/N = SI_{v_i}(G)$ which concludes the proof.

So far, we considered out-regular graphs and four axioms: LOC, ED, NC and CY.
Here is where the proof splits:
\begin{itemize}
\item If $F$ satisfies EM, then it is equal to Seeley index for every graph; it is easy to see that every graph can be made out-regular if we divide weights of outgoing edges of every node by its outdegree (Lemma~\ref{lemma:si:all-graphs}).
\item If $F$ satisfies EC, then it is equal to Eigenvector centrality for every graph; here, more detail analysis is required, as the EC operation changes weights of both outgoing and incoming edges (Lemma~\ref{lemma:ev:all-graphs}).
\end{itemize}


\begin{figure*}[t]
\centering
\begin{tikzpicture}[x=4.5cm,y=5cm] 
  \tikzset{     
    e4c node/.style={fill = black!05, circle,draw,minimum size=0.5cm,inner sep=0}, 
    e4c edge/.style={ -latex, above,font=\footnotesize}
  }

  \def\x{0.0}
  \node[e4c node] (1) at (\x+0.15, 0.18) {$u$}; 
  \node[e4c node] (2) at (\x+0.50, 0.36) {$v$}; 
  \node[e4c node] (6) at (\x+0.75, 0.36) {}; 
  \node[e4c node] (7) at (\x+0.50, 0.12) {}; 
  \node[e4c node] (8) at (\x+0.15, 0.36) {}; 
  \node[e4c node] (9) at (\x+0.50, 0.00) {};   
  \node (x) at (\x+0.70, 0.05) {$G$};   

  \path[->,draw,thick]
  (1) edge[e4c edge, above] node {$a$} (2)
  (1) edge[e4c edge, above] node[yshift=-0.05cm] {$b$} (7)
  (1) edge[e4c edge, below] node {$c$} (9)
  (2) edge[e4c edge]  (6)
  (8) edge[e4c edge]  (1)
  (8) edge[e4c edge]  (2)
  (6) edge[e4c edge, bend right=25]  (8)
  ;

  \def\x{1.0}
  \node[e4c node] (1) at (\x+0.05, 0.18) {$u$}; 
  \node[e4c node] (2) at (\x+0.50, 0.36) {$v$}; 
  \node[e4c node] (3) at (\x+0.25, 0.18) {$u'$}; 
  \node[e4c node] (5) at (\x+0.05, 0.00) {$w$}; 
  \node[e4c node] (6) at (\x+0.75, 0.36) {}; 
  \node[e4c node] (7) at (\x+0.50, 0.12) {}; 
  \node[e4c node] (8) at (\x+0.05, 0.36) {}; 
  \node[e4c node] (9) at (\x+0.50, 0.00) {}; 
  \node (x) at (\x+0.70, 0.05) {$G'$};   

  \path[->,draw,thick]
  (1) edge[e4c edge, right] node {$a$+$b$+$c$} (5)
  (2) edge[e4c edge]  (6)
  (3) edge[e4c edge, above] node {$a$} (2)
  (3) edge[e4c edge, above] node[yshift=-0.05cm] {$b$} (7)
  (3) edge[e4c edge, below] node {$c$} (9)
  (8) edge[e4c edge]  (1)
  (8) edge[e4c edge]  (2)
  (6) edge[e4c edge, bend right=25]  (8)
  ;

  \def\x{2.0}
  \node[e4c node] (1) at (\x+0.05, 0.18) {$u$}; 
  \node[e4c node] (2) at (\x+0.50, 0.36) {$v$}; 
  \node[e4c node] (3) at (\x+0.25, 0.18) {$u'$}; 
  \node[e4c node] (4) at (\x+0.50, 0.24) {$v'$}; 
  \node[e4c node] (5) at (\x+0.05, 0.00) {$w$}; 
  \node[e4c node] (6) at (\x+0.75, 0.36) {}; 
  \node[e4c node] (7) at (\x+0.50, 0.12) {}; 
  \node[e4c node] (8) at (\x+0.05, 0.36) {}; 
  \node[e4c node] (9) at (\x+0.50, 0.00) {}; 
  \node (x) at (\x+0.70, 0.05) {$G''$};   

  \path[->,draw,thick]
  (1) edge[e4c edge, right] node {$a$+$b$+$c$} (5)
  (2) edge[e4c edge]  (6)
  (3) edge[e4c edge, above] node {$a$} (4)
  (3) edge[e4c edge, above, pos=0.7] node[yshift=-0.05cm] {$b$} (7)
  (3) edge[e4c edge, below] node {$c$} (9)
  (4) edge[e4c edge]  (6)
  (8) edge[e4c edge]  (1)
  (8) edge[e4c edge]  (2)
  (6) edge[e4c edge, bend right=25]  (8)
  ;

  \def\x{0.6}
  \def\z{-0.5}
  \node[e4c node] (3) at (\x+0.25, 0.18+\z) {$u'$}; 
  \node[e4c node] (4) at (\x+0.50, 0.24+\z) {$v'$}; 
  \node[e4c node] (7) at (\x+0.50, 0.12+\z) {}; 
  \node[e4c node] (9) at (\x+0.50, 0.00+\z) {}; 
  \node (x) at (\x+0.65, 0.05+\z) {$G^*$};   

  \path[->,draw,thick]
  (3) edge[e4c edge, above] node {$a$} (4)
  (3) edge[e4c edge, above, pos=0.7] node[yshift=-0.05cm] {$b$} (7)
  (3) edge[e4c edge, below] node {$c$} (9)
  ;

  \def\x{1.4}
  \node[e4c node] (3) at (\x+0.25, 0.18+\z) {$u'$}; 
  \node[e4c node] (4) at (\x+0.50, 0.24+\z) {$v'$}; 
  \node[e4c node] (7) at (\x+0.50, 0.06+\z) {}; 

  \path[->,draw,thick]
  (3) edge[e4c edge, above] node {$a$} (4)
  (3) edge[e4c edge, below, pos=0.3] node {$b$+$c$} (7)
  ;

\end{tikzpicture}
\caption{Graphs showing that the profit from edge $(u,v)$ for node $v$ is equal to $p_F(F_u(G), c(u,v), \deg^+_u(G))$.} 
\label{figure:proof_2}
\end{figure*}
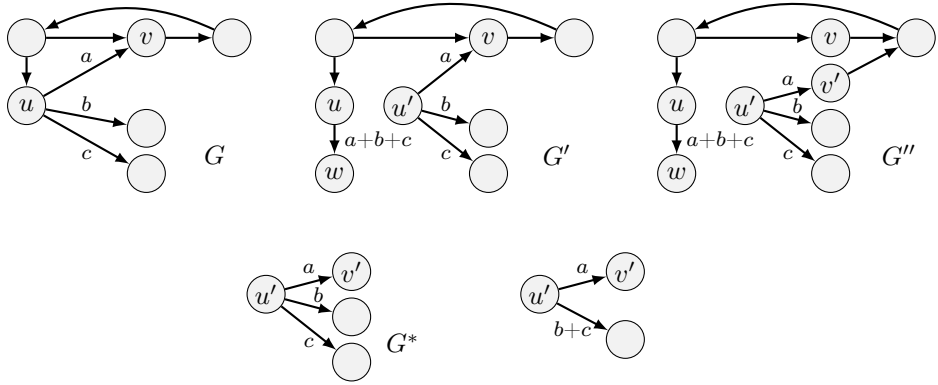

Now, let us turn our attention to PageRank and Katz centrality.
We have the following results:

\begin{theorem}
\label{theorem:pr}
A centrality measure defined on $\mathcal{G}$ satisfies LOC, ED, NC, EM and BL if and only if it is PageRank for some decay factor $a$ (Equation~\eqref{eq:walk:pr}).
\end{theorem}

\begin{theorem}
\label{theorem:katz}
A centrality measure defined on $\mathcal{G}^{K(a)}$ satisfies LOC, ED, NC, EC and BL if and only if it is Katz centrality for some decay factor $a$ (Equation~\eqref{eq:walk:katz}).
\end{theorem}

It is easy to check that PageRank and Katz centrality indeed satisfy the corresponding axioms (Lemmas~\ref{lemma:axioms:pr}--\ref{lemma:axioms:katz}).
Let us focus on the proof of uniqueness.

The key role in our proof will be played by \emph{semi-out-regular graphs} which is a class of graphs in which all nodes except for sinks (nodes with no outgoing edges) have equal out-degrees.
Most lemmas described below applies only to semi-out-regular graphs.

Let $F$ be a centrality measure defined on $\mathcal{G}$ or $\mathcal{G}^{K(a)}$ that satisfies LOC, ED, NC, BL and EM or EC.
First, we show several technical properties of $F$ that we use later in the proof:
Node Combination generalized to all semi-out-regular graphs (Lemma~\ref{lemma:pr-katz:nc-in-sor}), 
minimal centrality equal to the node weight (Lemmas~\ref{lemma:pr-katz:source-node}--\ref{lemma:pr-katz:positive-weight}), 
linearity with respect to node weights (Lemma~\ref{lemma:pr-katz:node-weights}) 
and existence of a constant $a_F$ such that $F_v(\{u,v\}, \{(u,v)\}, b, c) = a_F \cdot b(u)$ if $b(v)=0$ and $c(u,v)=1$ (Lemma~\ref{lemma:pr-katz:one-arrow-graphs}).

In the remainder of the proof we show that the profit from edge $(u,v)$ for node $v$ depends only on its weight and the centrality and out-degree of node $u$; moreover, if $F$ satisfies EM or EC, then this profit is equal to the profit in PageRank or Katz centrality.

To this end, for $x,y,z > 0$, $y \le z$, we define a \emph{profit function} $p_F(x,y,z)$ as follows: 
\[ p_F(x,y,z) = F_v(\{u,v,w\}, \{(u,v), (u,w)\}, b^x, c^{y,z}), \]
where $b^x(u) = x$, $b^x(v)=b^x(w)=0$, $c^{y,z}(u,v) = y$ and $c^{y,z}(u,w) = z-y$ (if $z=y$, then we remove edge $(u,w)$ from the graph).
To put in words, value $p_F(x,y,z)$ is the profit from an incoming edge with weight $y$ that starts in node with centrality $x$ and out-degree $z$ in the smallest such graph possible: with three nodes and two edges.
It is easy to check that if $F$ satisfies EM, then $p_F(x,y,z) = a_F \cdot x \cdot y/z$ as PageRank (Lemma~\ref{lemma:pr-katz:two-arrow-graphs:pr}) and if $F$ satisfies EC, then $p_F(x,y,z) = a_F \cdot x \cdot y$ as Katz centrality (Lemma~\ref{lemma:pr-katz:two-arrow-graphs:katz}). 

Now, it remains to prove that the profit from any edge $(u,v)$ equals $p_F(F_u, c(u,v), \deg^+_u(G))$.
More precisely, we prove that:
\[ F_v(G) = b(v) + \sum_{(u,v) \in \Gamma^+_u(G)} p_F(F_u(G), c(u,v), \deg^+_u(G)). \]
First, we do it for a self-loop which is the only incoming edge of a node (Lemma~\ref{lemma:pr-katz:loops}).
Then, we move on to the general proof for an arbitrary edge~(Lemma~\ref{lemma:pr-katz:recursive}).
We prove the thesis by induction over the number of incoming edges.

Let us illustrate the scheme of this proof on graph $G$ from Figure~\ref{figure:proof_2}.
Consider node $v$ in graph $G$ and pick one incoming edge, say $(u,v)$.
First we create graph $G'$ in which we extract from node $u$ all its outgoing edges and attach them to a new node $u'$ with the node weight $F_u(G)$.
To keep the out-degree of $u$ unchanged, we add a new node $w$ with an edge from $u$.
Using NC, it can be shown that $F_v(G') = F_v(G)$ and $F_{u'}(G') = F_u(G)$. 

In the second step, we split $v$ into two nodes with the same outgoing edges, but disjoint incoming edges: $v'$ has one edge $(u',v')$, and $v$ has the remaining edges.
From NC we know that $F_v(G) = F_v(G'') + F_{v'}(G'')$.
Since $v$ has less incoming edges in $G''$ than in $G$, we can use the inductive assumption.
Hence, it remains to determine the centrality of $v'$ in $G''$.

To this end, observe that only $u'$ is the predecessor of $v'$.
Hence, from ED, removing all outgoing edges of other nodes does not affect the centrality of $v$. 
Also, from LOC, removing nodes other than $u'$ and its direct predecessors does not change the centrality of $v$.
Hence, we can focus on graph $G^*$.
Using NC we can merge all of the nodes in $G^*$ except for $u'$ and $v'$.
In this way, we get a graph with three nodes and two edges, so we have $F_{v'}(G'') = p_F(F_u(G), c(u,v), \deg^+_u(G))$ what we needed to prove.

Combining the above results and using EM or EC we get uniqueness for PageRank (Lemmas~\ref{lemma:pr:semi-out-regular}--\ref{lemma:pr:final}) and Katz centrality (Lemmas~\ref{lemma:katz:semi-out-regular}--\ref{lemma:katz:final}).

\section{Conclusions}

We proposed the first joint axiomatization of four classic feedback centralities: Eigenvector centrality, Katz centrality, Seeley index and PageRank. 
We used seven axioms and proved that each centrality measure is uniquely characterize by a set of five axioms.
Our axiomatization highlights the similarities and differences between these measures which helps in making an informed choice of a centrality measure for a specific application at hand.

There are many possible directions of further research.
It would be interesting to extend our axiomatization to Degree centrality and Beta measure which constitute the borderline cases of feedback centralities.
Another direction is considering centrality measures based on random walks, such as Random Walk Closeness or Decay.
Also, several axioms considered in our paper can form a basis for the first axiomatization of distance-based centrality measures for directed graphs.

\bibliographystyle{abbrvnat}
\bibliography{bibliography}



\clearpage
\appendix

\newpage
\section{Proofs}
This appendix contains all the proofs.
First, in Section~\ref{sec:proofs:walks} we show that walk-based definitions of feedback centralities indeed satisfy their recursive equations.
Then, in Section~\ref{sec:proofs:ev:si} we focus on Seeley index and Eigenvector centrality (Theorems~\ref{theorem:si} and~\ref{theorem:ev}).
Finally, in Section~\ref{sec:proofs:katz:pr} we move to PageRank and Katz centrality (Theorems~\ref{theorem:pr} and~\ref{theorem:katz}).

\subsection{Walk-based definitions of feedback centralities}\label{sec:proofs:walks}
In this section, we show that centrality measures defined in Equations~\eqref{eq:walk:pr}--\eqref{eq:walk:ev} satisfy recursive equations from 
Equations~\eqref{eq:rec:ev}--\eqref{eq:rec:pr}.
Specifically:
\begin{itemize}
    \item In Proposition~\ref{proposition:walk:rec:ev} we show that Eigenvector centrality as defined in Equation~\eqref{eq:walk:ev} satisfies Eigenvector centrality recursive equation (Equation~\eqref{eq:rec:ev}).

    \item In Proposition~\ref{proposition:walk:rec:katz} we show that Katz centrality as defined in Equation~\eqref{eq:walk:katz} satisfies Katz centrality recursive equation (Equation~\eqref{eq:rec:katz}).
    
    \item In Proposition~\ref{proposition:walk:rec:si} we show that Seeley index as defined in Equation~\eqref{eq:walk:si} satisfies Seeley index recursive equation (Equation~\eqref{eq:rec:si}).
\end{itemize}

The proof that PageRank as defined in Equation~\eqref{eq:walk:pr} satisfies PageRank recursive equation (Equation~\eqref{eq:rec:pr}) can be found in~\citep{Was:Skibski:2023:pagerank}.

\begin{proposition}
\label{proposition:walk:rec:ev}
Eigenvector centrality defined on $\mathcal{G}^{EV}$ by Equation~\eqref{eq:walk:ev} satisfies Eigenvector centrality recursive equation (Equation~\eqref{eq:rec:ev}).
\end{proposition}
\begin{proof}
We will prove that the centrality defined as $F_v(G) = \lim_{T \rightarrow \infty} \sum_{t=0}^T w^{1/\lambda}_{v,G}(t)/T$ for every $G=(V,E,b,c)$ and every $v \in V$ satisfies Eigenvector recursive equation (Equation~\eqref{eq:rec:ev}) i.e., 
$x_v=\sum_{(u,v) \in \Gamma^-_v(G)} \frac{1}{\lambda} c(u,v) \cdot x_u.$

To this end, observe that for every walk $\omega \in \Omega_t(G)$ such that $\omega(t)=v$ the walk must have visited a direct predecessor of $v$, say $u$, in step $t-1$ and then move through edge $(u,v)$. Hence, for value $w^{1/\lambda}_{v,G}(t)$ we get
\[
w^{1/\lambda}_{v,G}(t) = \frac{1}{\lambda} \sum_{(u,v) \in \Gamma^-_v(G)} c(u,v) w^{1/\lambda}_{u,G}(t-1).
\]
If we sum both sides for all $t \in \{1,\dots, T\}$ we obtain
\[
\sum_{t=1}^T w^{1/\lambda}_{v,G}(t) =  \sum_{(u,v) \in \Gamma^-_v(G)} \frac{c(u,v)}{\lambda} \sum_{t=0}^{T-1} w^{1/\lambda}_{u,G}(t)
\]
Adding $w^{1/\lambda}_{v,G}(0) = b(v)$ and dividing both sides by $T$ yields
\begin{equation}
\label{eq:proposition:walk:rec:ev}
\sum_{t=0}^T \frac{w^{1/\lambda}_{v,G}(t)}{T} = \frac{b(v)}{T} + \sum_{(u,v) \in \Gamma^-_v(G)}  \frac{c(u,v)}{\lambda} \sum_{t=0}^{T-1} \frac{w^{1/\lambda}_{u,G}(t)}{T}.
\end{equation}
As $T$ approaches infinity $b(v)/T$ approaches zero.
In order to show that value $w^{1/\lambda}_{u,G}(T)/T$ approaches zero as well, observe that for each step $t \in \mathbb{N}$, the vector of values $w^{1/\lambda}_{u,G}(t)$ for all nodes $u \in V$ is a vector of values $w^{1/\lambda}_{u,G}(0)$ for all $u \in V$ multiplied $t$ times by adjacency matrix of $G$ and divided $t$ times by $\lambda$.
The norm of such vector is bounded~\citep{Mises:Pollaczek-Geiringer:1929}, thus its coordinates are also bounded.
Hence, the value $w^{1/\lambda}_{u,G}(T)/T$ indeed approaches zero as $T$ approaches infinity.
Thus, taking limit in Equation~\eqref{eq:proposition:walk:rec:ev} we obtain
\[
F_v(G) = \sum_{(u,v) \in \Gamma^-_v(G)} \frac{1}{\lambda} c(u,v)  F_u(G),
\]
which is exactly Eigenvector recursive equation.
\end{proof}

\begin{proposition}
\label{proposition:walk:rec:katz}
Katz centrality measure defined on $\mathcal{G}^{K(a)}$ by Equation~\eqref{eq:walk:katz} satisfies Katz centrality recursive equation (Equation~\eqref{eq:rec:katz}).
\end{proposition}
\begin{proof}
We will prove that centrality measure defined as $F_v(G) = \sum_{t=0}^T w^a_{v,G}(t)$ for every graph $G=(V,E,b,c)$ and every $v \in V$ satisfies Katz Recursive Equation~\eqref{eq:rec:katz}, i.e., 
\[
  x_v=a \left( \sum_{(u,v) \in \Gamma^-_v(G)} c(u,v) \cdot x_u \right) + b(v).
\]

To this end, observe that for each walk $\omega \in \Omega_t(G)$ such that $\omega(t)=v$ in order to arrive in $v$ at step $t \ge 1$, it must have visited a direct predecessor of $v$, say $u$, in step $t-1$ and then follow edge $(u,v)$.
Thus, for $w^a_{v,G}(t)$ we obtain that
\(
w^a_{v,G}(t) = a \cdot \sum_{(u,v) \in \Gamma^-_v(G)} w^a_{u,G}(t-1) \cdot c(u,v) .
\)
Summing both sides for all $t \in \{1,2,\dots\}$ we get
\[
\sum_{t=1}^\infty w^a_{v,G}(t) =  a \cdot \sum_{(u,v) \in \Gamma^-_v(G)} c(u,v) \sum_{t=0}^{\infty} w^a_{u,G}(t)
\]
Let us add $w^a_{v,G}(0) = b(v)$ to both sides of equation to obtain
\(
F_v(G) = a \cdot ( \sum_{(u,v) \in \Gamma^-_v(G)}   c(u,v) \cdot F_u(G) ) + b(v)
\)
which is exactly Katz centrality recursive equation.
\end{proof}

\begin{proposition}
\label{proposition:walk:rec:si}
Seeley index defined on $\mathcal{G}^{SI}$ by Equation~\eqref{eq:walk:si} satisfies Seeley index recursive equation (Equation~\eqref{eq:rec:si}) and for every $G=(V,E,b,c) \in \mathcal{G}^{SI}$ it holds that
$$\sum_{v \in V} SI_v(G) = \sum_{v \in V} b(v).$$
\end{proposition}
\begin{proof}
First, let us focus on the second part of the thesis, i.e., that the sum of centralities of all nodes is equal to the sum of weights of all nodes.
To this end, observe that for every $t>0$ every walk that $\omega \in \Omega_t(G)$ that ends in $v$, i.e., $\omega(t)=v$, must have visited one of the direct predecessors of $v$, say $u$, in step $t-1$ and then move to $v$ through edge $(u,v)$.
Thus, if we look at the value $p^1_{v,G}(t)$ we obtain that
\begin{equation}
    \label{eq:proposition:walk:rec:si:1}
    p^1_{v,G}(t) = \sum_{(u,v) \in \Gamma^-_v(G)} \frac{c(u,v)}{\deg^+_u(G)} p^1_{u,G}(t-1) .    
\end{equation}
Let us sum both sides of Equation~\ref{eq:proposition:walk:rec:si:1} for all nodes $v \in V$.
Observe that each edge $(u,v) \in E$ appears on the right hand side of the equation exactly once, hence for every $u \in V$ all fractions $c(u,v)/\deg^+_u(G)$ sums to 1.
In this way, we obtain
\begin{equation}
    \label{eq:proposition:walk:rec:si:2}
    \sum_{v \in V}  p^1_{v,G}(t) = \sum_{u \in V} p^1_{u,G}(t-1).
\end{equation}
This means that the total amount of entity in all nodes is the same in each step $t \in \mathbb{N}$ of the walk process.
Since $p^1_{v,G}(0) = b(v)$ for every $v \in V$, this total amount is always equal to the sum of node weights.
Thus, when we sum it for $t \in \{0,\dots,T\}$, divide by $T$ and take a limit we still obtain that
\[
    \sum_{v \in V} SI_v(G) = 
    \lim_{T \rightarrow \infty} \sum_{t=0}^T \frac{\sum_{v \in V}  p^1_{v,G}(t)}{T} = 
    \sum_{v \in V} b(v).
\]

Now, let us move to the first part of the thesis, i.e., that centrality measure defined as $F_v(G) = \lim_{T \rightarrow \infty} \sum_{t=0}^T p^1_{v,G}(t)/T$ for every graph $G=(V,E,b,c)$ satisfies recursive equation
\[
  x_v=\sum_{(u,v) \in \Gamma^-_v(G)} x_u \cdot c(u,v) /\deg^+_u(G).
\]
To this end, we sum both sides of Equation~\eqref{eq:proposition:walk:rec:si:1} for all $t \in \{1,\dots, T\}$, to get
\[
\sum_{t=1}^T p^1_{v,G}(t) =  \sum_{(u,v) \in \Gamma^-_v(G)} \frac{c(u,v)}{\deg^+_u(G)} \sum_{t=0}^{T-1} p^1_{u,G}(t)
\]
Now, if we add $p^1_{v,G}(0) = b(v)$ to both sides of the equation and divide each side by $T$, we obtain
\begin{equation*}
\sum_{t=0}^T \frac{p^1_{v,G}(t)}{T} = \frac{b(v)}{T} + \sum_{(u,v) \in \Gamma^-_v(G)}  \frac{c(u,v)}{\deg^+_u(G)} \sum_{t=0}^{T-1} \frac{p^1_{u,G}(t)}{T}    
\end{equation*}
When $T$ approaches infinity, $b(v)/T$ approaches zero. Hence,
\begin{equation}
\label{eq:proposition:walk:rec:si:3}
    F_v(G) = \sum_{(u,v) \in \Gamma^-_v(G)}  \frac{c(u,v)}{\deg^+_u(G)} \lim_{T \rightarrow \infty} \sum_{t=0}^{T-1} \frac{p^1_{u,G}(t)}{T}    
\end{equation}
Thus, it remains to show that $p^1_{u,G}(T)/T$ approaches zero as well.
To this end, observe that from Equation~\eqref{eq:proposition:walk:rec:si:2} the sum of $p^1_{u,G}(t)$ for all $u \in V$ is constant for all $t \in \mathbb{N}$.
Thus, the value of $p^1_{u,G}(t)$ is bounded.
Hence, term $p^1_{u,G}(T)/T$ indeed approaches zero
Therefore, from Equation~\eqref{eq:proposition:walk:rec:si:3} we get
\(
F_v(G) = \sum_{(u,v) \in \Gamma^-_v(G)}  F_u(G) \cdot c(u,v)/\deg^+_u(G)
\)
which is exactly Seeley index recursive equation.
\end{proof}


\subsection{Proof of Theorems~\ref{theorem:si} and~\ref{theorem:ev}}
\label{sec:proofs:ev:si}

In this section, we present the full proofs of Theorems~\ref{theorem:si} and~\ref{theorem:ev} which state that our axioms uniquely characterize Seeley index and Eigenvector centrality.
More in detail, we begin with Lemma~\ref{lemma:ev-si:equal-on-regular} in which we prove that Eigenvector centrality is equal to Seeley index in every out-regular graph.
Then, in Lemmas~\ref{lemma:axioms:si} and~\ref{lemma:axioms:ev} we show that Seeley index and Eigenvector centrality indeed satisfy our axioms.
Lemmas~\ref{lemma:ev-si:super-regular-graphs} and~\ref{lemma:ev-si:regular-graphs} state that any centrality measure satisfying our axioms indicate the same centralities for particular subclasses of graphs.
Finally, in Lemma~\ref{lemma:si:all-graphs} we prove that LOC, ED, NC, EM, and CY uniquely characterize Seeley index and in Lemma~\ref{lemma:ev:all-graphs} that LOC, ED, NC, EC, and CY uniquely characterize Eigenvector centrality.

\begin{lemma}
\label{lemma:ev-si:equal-on-regular}
For every graph $G=(V,E,b,c) \in \mathcal{G}^{EV}$ that is out-regular we have
$$EV_v(G)=SI_v(G) \quad \mbox{for every } v \in V.$$
\end{lemma}
\begin{proof}
Since there exist $\lambda \in \mathbb{R}_{\ge 0}$ such that $\deg^+_u(G)=\lambda$ for every $u \in V$, for every node $v \in V$ and step $t \in \mathbb{N}$ we get
$$w^{1/\lambda}_{v,G}(t) \! = \! \sum_{\substack{\omega \in \Omega_t(G) :\\ \omega(t)=v}} \!  b(\omega(0) \! ) \! \prod_{i=0}^{t-1} \! \frac{c(\omega(i),\! \omega(i\! + \! 1) \! )}{\deg^+_{\omega(i)}(G)} \!  = \! p^1_{v,G}(t).$$
Thus, Equations~\eqref{eq:walk:si} and~\eqref{eq:walk:ev} yield
$EV_v(G)=SI_v(G)$.
\end{proof}

\begin{lemma}
\label{lemma:axioms:si}
Seeley index defined on $\mathcal{G}^{SI}$ by Equation~\eqref{eq:walk:si} satisfies LOC, ED, NC, EM, and CY.
\end{lemma}
\begin{proof}
Let us take arbitrary graph $G=(V,E,b,c) \in \mathcal{G}^{SI}$ and consider axioms one by one.

For LOC consider graph $G'=(V',E',b',c')$ s.t. $V \cap V' = \emptyset$ and arbitrary node $v \in V$.
Observe that in $G+G'$ any walk that starts in one of the nodes in $V'$ cannot visit nodes in $V$ and vice versa.
Thus, for any $t \in \mathbb{N}$ we have that $\{\omega \in \Omega_t(G+G') : \omega(t)=v\} = \{\omega \in \Omega_t(G) : \omega(t)=v\}$.
This implies that also $p^{1}_{v,G+G'}(t)=p^{1}_{v,G}(t)$ since weights of edges in $E$ and out-degrees of nodes in $V$ are the same in both $G$ and $G+G'$.
Hence, LOC follows from Equation~\eqref{eq:walk:ev}.

For ED consider edge $(u,w) \in E$ and arbitrary node $v \in V \setminus S(u)$.
Observe that for $G \in \mathcal{G}^{SI}$ this is only possible if $u$ and $v$ belong to different strongly connected components, i.e., there exist graphs 
$$G_v = (V_v,E_v,b_{V_v},c_{E_v}) \quad \mbox{and} \quad G_u = (V_u,E_u,b_{V_u},c_{E_u})$$
such that $V_v \cap V_u = \emptyset$ and $G_v + G_u = G$.
Since Seeley index satisfies LOC, we get that $SI_v(G)=SI_v(G_v)$.
Now, if we take $G_u'=(V_u,E_u \setminus \{(u,w)\},b_{V_u},c_{E_u \setminus \{(u,w)\}})$, then still from LOC we get that $SI_v(G_v + G_u')=SI_v(G_v)$ and ED follows.

For NC consider nodes $u,w \in V$ such that $\deg^+_u(G)=\deg^+_w(G)=\deg^+_s(G)$ for every $s \in S(u)\cup S(w)$.
Observe that proportional combining of $u$ into $w$ does not affect the sum of node weights in the graph.
Thus, since Seeley index can be equivalently defined as the solution to the system of recursive equations and normalisation equation $\sum_{v \in V} SI_v(G) = \sum_{v \in V} b(v)$ (Proposition~\ref{proposition:walk:rec:si}), it suffices to show that $(x_v)_{v \in V \setminus \{u\}}$ defined as $x_v = SI_v(G)$ for every $v \in V \setminus \{u,w\}$ and $x_w = SI_u(G) + SI_w(G)$ satisfies Seeley index recursive equation (Equation~\eqref{eq:rec:si}) for graph $G' = (V',E',b',c') = C^{SI}_{u \rightarrow w}(G)$ and every $v \in V \setminus \{u\}$.
To this end, fix $v \in V \setminus \{w\}$ and observe that from Seeley index recursive equation (Equation~\eqref{eq:rec:si}) for graph $G$ we have
\begin{equation}
    \label{eq:lemma:axioms:si:nc:1}
    SI_v(G) = \sum_{(s,v) \in \Gamma^-_v(G)} SI_u(G) \cdot \frac{c(s,v)}{\deg^+_s(G)}.
\end{equation}
If $v$ is neither $w$ nor a direct successor of $u$ or $w$ in $G$, i.e., $(u,v),(w,v) \not \in \Gamma^-_v(G)$,
then proportional combining of $u$ into $w$ does not affect the incoming edges of $v$, hence $\Gamma^-_v(G)=\Gamma^-_v(G')$.
Moreover, for every edge $(s,v) \in \Gamma^-_v(G)$ its weight is unchanged, i.e., $c(s,v)=c'(s,v)$, as well as the out-degree of node $s$, i.e., $\deg^+_s(G)=\deg^+_s(G')$.
Thus, from Equation~\eqref{eq:lemma:axioms:si:nc:1} we get
\[
    x_v = \sum_{(s,v) \in \Gamma^-_v(G')} x_u \cdot \frac{c'(s,v)}{\deg^+_s(G')}
\]
which is Seeley index recursive equation for graph $G'$ and node $v \in V$.
Let us move to a case when $v$ is a direct successor of either $u$ or $w$ in $G$, i.e., $(u,v) \in \Gamma^-_v(G)$ or $(w,v) \not \in \Gamma^-_v(G)$, but $v \neq w$.
Then, Equation~\eqref{eq:lemma:axioms:si:nc:1} can be transformed into
\begin{equation}
    \label{eq:lemma:axioms:si:nc:2}
    x_v = \frac{SI_u(G) \tilde{c}(u,v)}{\deg^+_u(G)} + \frac{SI_w(G) \tilde{c}(w,v)}{\deg^+_w(G)} +
    \sum_{(s,v) \in \Gamma^-_v(G):s \not \in \{u,w\}} x_s \cdot \frac{c(s,v)}{\deg^+_s(G)},
\end{equation}
where we define $\tilde{c}$ in such a way that $\tilde{c}(e)=c(e)$ if $e \in E$ and $\tilde{c}(e)=0$ otherwise.
As we have assumed: $\deg^+_u(G)=\deg^+_w(G)$.
Moreover, observe that $\deg^+_w(G) = \deg^+_w(G')$.
Furthermore, for every $(s,v) \in \Gamma^-_v(G)$ such that $s \not \in \{u,w\}$ we have $c'(s,v)=c(s,v)$ and $\deg^+_s(G)=\deg^+_s(G')$.
Also, we have that
$\{ \!(s,v) \! \in \! \Gamma^-_v(G) \! :\! s \! \not \in \! \{u,w\} \!\} \! =
\! \{\! (s,v) \! \in \! \Gamma^-_v(G')\! :\! s \! \not \in \! \{u,w\}\! \}$.
Finally, by the definition of proportional combining of nodes:
$c'(w,v) = (SI_u(G)\tilde{c}(u,v) + SI_w(G)\tilde{c}(w,v))/(SI_u(G) + SI_w(G))$.
Combining these facts with Equation~\eqref{eq:lemma:axioms:si:nc:2} we get
\[
    x_v = \frac{x_w c'(w,v)}{\deg^+_w(G')} + \sum_{(s,v) \in \Gamma^-_v(G'):s \neq w} x_s \cdot \frac{c'(s,v)}{\deg^+_s(G')},
\]
which is Seeley index recursive equation for graph $G'$ and node $v$.
Therefore, it remains to consider node $w$.
From Seeley index recursive equation (Equation~\eqref{eq:rec:si}) for graph $G$ we have
\begin{equation}
    \label{eq:lemma:axioms:si:nc:3}
    SI_w(G) = \frac{SI_u(G) \tilde{c}(u,w)}{\deg^+_u(G)} + \frac{SI_w(G) \tilde{c}(w,w)}{\deg^+_w(G)} +
    \sum_{(s,w) \in \Gamma^-_w(G):s \not \in \{u,w\}} x_s \cdot \frac{c(s,w)}{\deg^+_s(G)}
\end{equation}
and analogous equation for $SI_u(G)$.
From the definition of proportional combining we obtain
$c'(w,w) = (SI_u(G) (\tilde{c}(u,w) + \tilde{c}(u,u)) + SI_w(G) (\tilde{c}(w,w)+ \tilde{c}(w,u)))/(SI_u(G) + SI_w(G))$.
Thus, taking into account that  $\deg^+_u(G)=\deg^+_w(G)= \deg^+_w(G')$,
$\{\! (s,w) \!\in \!\Gamma^-_w(G')\! :\! s \neq w\} \! =\! \{\! (s,u) \!\in \!\Gamma^-_u(G):s \!\not \in \!\{u,w\}\!\} \cup
\{\!(s,w) \!\in \!\Gamma^-_w(G)\! :\! s \not \in \{u,w\}\!\}$,
and that
$c'(s,w)=\tilde{c}(s,u) + \tilde{c}(s,w)$ and $\deg^+_s(G)\!=\!\deg^+_s(G')$ for every $(s,w) \!\in \!\Gamma^-_w(G')$ such that $s \neq w$, we get
\[
    x_w = \frac{x_w c(w,w)}{\deg^+_w(G')} + \sum_{(s,w) \in \Gamma^-_w(G'):s \neq w} x_s \cdot \frac{c'(s,w)}{\deg^+_s(G')},
\]
which is Seeley index recursive equation for graph $G'$ and node $w$.

For EM consider arbitrary nodes $u,v \in V$, constant $x \in \mathbb{R}_{\ge 0}$ and graph
$G' = (V,E,b,c_{-\Gamma^+_u(G)} + x \cdot c_{\Gamma^+_u(G)})$.
Observe that for every $t \in \mathbb{N}$ and walk $\omega \in \Omega_t(G)$ such that $\omega(t)=v$ the value
$$b(\omega(0))\cdot \prod_{i=0}^{t-1} \frac{c(\omega(i),\omega(i+1))}{\deg^+_{\omega(i)}(G)}$$
is the same for both $G$ and $G'$ (both numerator and denominator for all $\omega(i)=u$ is multiplied by $x$).
Thus, we have that $p^1_{v,G}(t) = p^1_{v,G'}(t)$ does not change as well.
Summing for all $t \in \mathbb{N}$, we get that $SI_v(G)=SI_v(G')$ from Equation~\eqref{eq:walk:si}.

Finally, for CY observe that if graph $G$ is a cycle graph, then from Equation~\eqref{eq:rec:si} for every node $v \in V$ we have that $SI_v(G) = SI_u(G)$
where $u$ is the node that precedes $v$ in the cycle, i.e., $\Gamma^-_v(G)=\{(u,v)\}$.
Thus, all nodes have equal centralities.
From Proposition~\ref{proposition:walk:rec:si} we get $|V| \cdot SI_v(G) = \sum_{v \in V} SI_v(G) = \sum_{v \in V}b(v)$, hence $SI_v(G) = \sum_{v \in V}b(v)/|V|$.
\end{proof}

\begin{lemma}
\label{lemma:axioms:ev}
Eigenvector centrality defined on $\mathcal{G}^{EV}$ by Equation~\eqref{eq:walk:ev} satisfies LOC, ED, NC, EC, and CY.
\end{lemma}
\begin{proof}
Let us take an arbitrary graph $G=(V,E,b,c) \in \mathcal{G}^{EV}$ and consider axioms one by one.

For LOC and ED the proof is analogous to the proof that Seeley index satisfies LOC and ED (Lemma~\ref{lemma:axioms:si}).


For NC consider nodes $u,w \in V$ such that $\deg^+_u(G)=\deg^+_w(G)=\deg^+_s(G)$ for every $s \in S(u) \cup S(w)$.
If all nodes in $V$ are successors of either $u$ or $w$, i.e., $V = S(u) \cup S(w)$, then the graph is out-regular and from Lemma~\ref{lemma:ev-si:equal-on-regular} $EV_v(G)=SI_v(G)$ for every $v \in V$.
Thus, since proportional combining preserves out-regularity, Eigenvector centrality satisfies NC because Seeley index satisfies NC (Lemma~\ref{lemma:axioms:si}).
If there are nodes in $V$ that are not successors of neither $u$ nor $w$, then since $G$ is a disjoint sum of strongly connected components it can be decomposed as a sum of two graphs:
$G_{uw} = (V_{uw}, E_{uw}, b_{V_{uw}},c_{E_{uw}})$  and $G' = (V', E', b_{V'},c_{E'})$ 
s.t. $V_{uw} \!=\! S(v) \cup S(u)$, $V_{uw} \cap V' \!=\! \emptyset$, and $G_{uw} + G' \!=\! G$.
Observe that $C^{EV}_{u \rightarrow w}(G) = C^{EV}_{u \rightarrow w}(G_{uw}) + G'$.
Also, $G_{uw}$ is out-regular, hence from Lemma~\ref{lemma:ev-si:equal-on-regular} we have $EV_v(G_{uw})=SI_v(G_{uw})$.
Therefore, NC follows from LOC and the fact that Seeley index satisfies NC (Lemma~\ref{lemma:axioms:si}).

For EC take arbitrary node $u \in V$ and constant $x > 0$, and denote $G=(V,E,b',c')$ such that $b'=b_{-u} + b_u/x$ and $c' = c_{-\Gamma^{\pm}_u(G) \setminus \{(u,u)\}} + c_{\Gamma^{+}_u(G) \setminus \{(u,u)\}} \cdot x + c_{\Gamma^{+}_u(G) \setminus \{(u,u)\}}/x$. 
Fix $v \in V \setminus \{u\}$.
Observe that for every $t \in \mathbb{N}$ and walk $\omega \in \Omega_t(G)$ that ends in $v$, i.e., s.t. $\omega(t) = v$, the value
$b(\omega(0)\!) \cdot \prod_{i=1}^{t} \frac{1}{\lambda}c(\omega(i-1),\omega(i)\! )$
is equal for graph $G$ and $G'$.
It holds because, since $\omega(t) \neq u$, for every step $i \in \{0,\dots,t-1\}$ in which the walk arrives at node $u$, i.e., $\omega(i-1) \neq u$ and $\omega(i)=u$, there exist step $j > i$ in which the walk departs from $u$, i.e., $\omega(k)=u$ for every $k \in \{i,\dots,j-1\}$ and $\omega(j)\neq u$.
Now, the factor for step $i$ decreases by $x$, i.e., $c'(\omega(i-1),\omega(i))=c(\omega(i-1),\omega(i))/x$ (or $b'(i) = b(i)/x$ if $i=0$), but at the same time the factor for step $j$ increases by $x$, i.e., $c'(\omega(j-1),\omega(j))=c(\omega(j-1),\omega(j)) \cdot x$.
Thus, indeed $b(\omega(0)\!) \cdot \prod_{i=1}^{t} \frac{1}{\lambda}c(\omega(i-1),\omega(i)\! )$
is equal for graph $G$ and $G'$.
Hence, we have also that $w^{1/\lambda}_{v,G}(t) = w^{1/\lambda}_{v,G'}(t)$ and from Equation~\eqref{eq:walk:ev} we obtain $EV_v(G)=EV_v(G')$.
Now, for node $u$, for every $t \in \{1,\dots\}$ and walk $\omega \in \Omega_t(G)$ such that $\omega(t)=u$ the value 
$b(\omega(0)\!) \cdot \prod_{i=0}^{t-2} \frac{1}{\lambda}c(\omega(i),\omega(i-1)\! )$
also does not change for the same reason.
However, since $\omega(t)=u$ we have that $c'(\omega(t-1),\omega(t))=c(\omega(t-1),\omega(t))/x$.
Thus, $w^{1/\lambda}_{u,G'}(t) = w^{1/\lambda}_{v,G}(t) /x$.
Similarly, for $t=0$ we have
$w^{1/\lambda}_{u,G'}(0) = b'(u) = b(u)/x = w^{1/\lambda}_{v,G}(t) /x.$
Hence, from Equation~\eqref{eq:walk:ev} we have $EV_u(G') = EV_u(G)/x$.

Finally, for CY from Lemmas~\ref{lemma:ev-si:equal-on-regular} and~\ref{lemma:axioms:si} for every $v \in V$ we have 
\[
  EV_v(G)=SI_v(G)=\sum_{u \in V}b(u)/|V|,
\]
which concludes the proof.
\end{proof}



\begin{lemma}
\label{lemma:ev-si:super-regular-graphs}
If a centrality measure $F$ defined on $\mathcal{G}^{SI}$ (or $\mathcal{G}^{EV}$) satisfies LOC, ED, NC, and CY, then for every $\lambda>0$ and every strongly connected $\lambda$-out-regular graph $G=(V,E, b, c)$ such that $G \in \mathcal{G}^{SI}$ (or $G \in \mathcal{G}^{EV}$), $c(e)/c(e') \in \mathbb{Q}$ for every $e,e' \in E$, and $\sum_{v \in V}b(v)=1$, we have
\begin{equation}
    \label{eq:lemma:ev-si:super-regular-graphs:0}
    F_v(G) = SI_v(G) = EV_v(G) \quad \mbox{for every } v \in V
\end{equation}
and $F_v(V,E,x \cdot b, c) = x \cdot F_v(G)$ for every $x > 0$.
\end{lemma}
\begin{proof}
Let us begin with Equation~\eqref{eq:lemma:ev-si:super-regular-graphs:0}.
The second equality, i.e., that $SI_v(G)=EV_c(G)$, comes from Lemma~\ref{lemma:ev-si:equal-on-regular}, thus let us focus on proving that $F_v(G)=SI_v(G)$ for every $v \in V$.
To this end, let us first define the \emph{impact} of an edge.
For any strongly connected graph $G=(V,E,b,c)$ and edge $(u,v) \in E$ let the impact of $(u,v)$ be equal to
$I_G(u,v) = SI_u(G) \cdot c(u,v) / \deg^+_u(G)$.
Intuitively, impact measures the amount of centrality that node $u$ transfers to node $v$.
Indeed, from Seeley index recursive equation (Equation~\eqref{eq:rec:si}) we see that the centrality of a node is equal to both the sum of impacts of its outgoing edges and the sum of impacts of its incoming edges, i.e.,
\begin{equation}
    \label{eq:lemma:ev-si:super-regular-graphs:1}
    \sum_{e \in \Gamma^-_v(G)} I_G(e) = SI_v(G) = \sum_{e \in \Gamma^+_v(G)} I_G(e).
\end{equation}
Another property that we will use in the proof is that proportional combining preserves the impact of edges.
More in detail, for any $(u,v),(u',v') \in E$ such that $v' \not \in \{u,u'\}$ and graph $G'=C^{SI}_{u' \rightarrow u}(G)$ we have
\begin{equation}
    \label{eq:lemma:ev-si:super-regular-graphs:2}
    I_{G'}(u,v) =
    \begin{cases}
        I_G(u,v) & \mbox{if } v \neq v', \\
        I_G(u,v) + I_G(u',v')  & \mbox{otherwise.}
    \end{cases}
\end{equation}
This comes from the fact, that the node resulting from the combination has the weight of its outgoing edges decreased by $SI_u(G)/(SI_u(G)+SI_{u'}(G))$, but at the same time its centrality increases by the same value.
Hence, the impact of its outgoing edges is unaffected.

From Proposition~\ref{proposition:walk:rec:si} we know that Seeley index of nodes in $V$ can be equivalently defined as the solution of the system of recursive equations and normalization equation $\sum_{v \in V} SI_v(G) = \sum_{v \in V} b(v) = 1$.
Observe that since proportions of the weights of edges are rational, i.e., it holds that $c(e)/c(e') \in \mathbb{Q}$ for every $e,e' \in E$, then also the coefficients in the system of equations, i.e., $c(u,v)/\deg^+_u(G)$, are rational (they are reciprocal of  $\deg^+_u(G)/c(u,v)$, which are the sums of proportions $c(e)/c(u,v)$ for all $e \in \Gamma^+_u(G)$).
If all coefficients are rational, then the solution, i.e., $SI_v(G)$ for every $v \in V$, is also rational.
Moreover, since both $SI_v(G)$ for every $v \in V$ and $c(u,v)/\deg^+_u(G)$ for every $(u,v) \in E$ are rational,
the impact of every edge is rational as well, i.e., $I_G(e) \in \mathbb{Q}$ for every $e \in E$.
Building upon this, we will consider a walk on graph $G$ that follows each edge the number of times that is proportional to its impact.
Next, we will construct the cycle graph based on this walk and by proportional combining of its nodes transform it into the original graph $G$.
Hence, based on CY and NC, we will establish centrality $F$ of each node.

\begin{figure}[t]
\centering
\begin{tikzpicture}
  \def\sx{0.7cm} 
  \def\sy{0.6cm} 
  \def\x{0cm} 
  \def\y{0cm} 
  \def\arrdist{0.35cm}

  \tikzset{
    node_blank/.style={circle,draw,minimum size=0.5cm,inner sep=0, color=white}, 
    node/.style={circle,draw,minimum size=0.5cm,inner sep=0, fill = black!05, font=\footnotesize}, 
    node_emph/.style={circle, minimum size=0.65cm, inner sep=0, fill = black!15, font=\footnotesize}, 
    edge/.style={-latex,above,font=\footnotesize}, 
    el/.style={below,font=\footnotesize}, 
    operation/.style={sloped,>=stealth,above,font=\footnotesize},
    arrow/.style={draw, single arrow, minimum width = 0.9cm, minimum height=\x-6*\s+\s, fill=black!10},
    blank/.style={}
  } 
  
    
  \node[node] (a) at (\x+1*\sx, 4*\sy + \y) {$a$};
  \node[node] (b) at (\x+3*\sx, 4*\sy + \y) {$b$};
  \node[node] (c) at (\x+4*\sx, 2*\sy + \y) {$c$};
  \node[node] (d) at (\x+3*\sx, 0*\sy + \y) {$d$};
  \node[node] (e) at (\x+1*\sx, 0*\sy + \y) {$e$};
  \node[node] (f) at (\x+0*\sx, 2*\sy + \y) {$f$};
  \node[blank] (blank) at (\x+4*\sx+0.2cm, -0.75cm + \y) {$G$};
  
  \path[->,draw,thick]
  (a) edge[edge, bend left=20, looseness = 0.9] node[above, pos=0.4] {$\lambda$} (b)
  (b) edge[edge, bend left=20, looseness = 0.9] node[right, pos=0.2] {$\lambda$} (c)
  (c) edge[edge, bend left=20, looseness = 0.9] node[right] {$\lambda$} (d)
  (d) edge[edge, bend left=20, looseness = 0.9] node[below] {$\lambda/2$} (e)
  (e) edge[edge, bend left=20, looseness = 0.9] node[left, pos =0.2] {$\lambda/2$} (f)
  (d) edge[edge] node[right, pos=0.6] {$\lambda/2$} (a)
  (e) edge[edge] node[above, pos=0.1] {$\lambda/2$} (c)
  (f) edge[edge] node[above, pos=0.1] {$\lambda$} (b)
  ;

  \def\x{4.5cm} 
  
  \node[node] (a) at (\x+1*\sx, 4*\sy + \y) {$a$};
  \node[node] (b) at (\x+3*\sx, 4*\sy + \y) {$b$};
  \node[node] (c) at (\x+4*\sx, 2*\sy + \y) {$c$};
  \node[node] (d) at (\x+3*\sx, 0*\sy + \y) {$d$};
  \node[node] (e) at (\x+1*\sx, 0*\sy + \y) {$e$};
  \node[node] (f) at (\x+0*\sx, 2*\sy + \y) {$f$};
  \node[blank] (blank) at (\x+4*\sx+0.2cm, -0.75cm + \y) {$I_{G}$};
  
  \path[->,draw,thick]
  (a) edge[edge, bend left=20, looseness = 0.9] node[above, pos=0.4] {$\frac{2}{16}$} (b)
  (b) edge[edge, bend left=20, looseness = 0.9] node[right, pos=0.2] {$\frac{3}{16}$} (c)
  (c) edge[edge, bend left=20, looseness = 0.9] node[right] {$\frac{4}{16}$} (d)
  (d) edge[edge, bend left=20, looseness = 0.9] node[below] {$\frac{2}{16}$} (e)
  (e) edge[edge, bend left=20, looseness = 0.9] node[left, pos =0.2] {$\frac{1}{16}$} (f)
  (d) edge[edge] node[right, pos=0.6] {$\frac{2}{16}$} (a)
  (e) edge[edge] node[above, pos=0.1] {$\frac{1}{16}$} (c)
  (f) edge[edge] node[above, pos=0.1] {$\frac{1}{16}$} (b)
  ;
  
  
  
   \def\x{9cm} 

  \node[node] (a) at (\x+1*\sx, 4*\sy + \y) {$a$};
  \node[node] (b) at (\x+3*\sx, 4*\sy + \y) {$b$};
  \node[node] (c) at (\x+4*\sx, 2*\sy + \y) {$c$};
  \node[node] (d) at (\x+3*\sx, 0*\sy + \y) {$d$};
  \node[node] (e) at (\x+1*\sx, 0*\sy + \y) {$e$};
  \node[node] (f) at (\x+0*\sx, 2*\sy + \y) {$f$};
  \node[blank] (blank) at (\x+4*\sx+0.2cm, -0.75cm + \y) {$\hat{G}$};
  
  \path[->,draw,thick]
  (c) edge[edge, bend right=30, looseness = 0.9] (d)
  (c) edge[edge, bend right=10, looseness = 0.9] (d)
  (c) edge[edge, bend right=-10, looseness = 0.9] (d)
  (c) edge[edge, bend right=-30, looseness = 0.9] (d)
  (b) edge[edge, bend right=20, looseness = 0.9] (c)
  (b) edge[edge] (c)
  (b) edge[edge, bend right=-20, looseness = 0.9] (c)
  (a) edge[edge, bend right=-20, looseness = 0.9] (b)
  (a) edge[edge] (b)
  (f) edge[edge] (b)
  (e) edge[edge, bend right=-20, looseness = 0.9] (f)
  (e) edge[edge, bend right=-15, looseness = 0.9] (c)
  (d) edge[edge, bend right=-20, looseness = 0.9] (e)
  (d) edge[edge] (e)
  (d) edge[edge, bend right=10, looseness = 0.9](a)
  (d) edge[edge, bend right=-10, looseness = 0.9](a)
  ;

  

\end{tikzpicture}
\caption{An illustration to the proof of Lemma~\ref{lemma:ev-si:super-regular-graphs}.
The leftmost graph, $G$, is a $\lambda$-out-regular graph with the weight of each edge shown.
The middle graph is graph $G$ as well, but with impact of each edge shown instead of its weight.
The rightmost graph, $\hat{G}$, is an unweighted multi-graph obtained from $G$.}
\label{fig:lemma:impacts}
\end{figure}
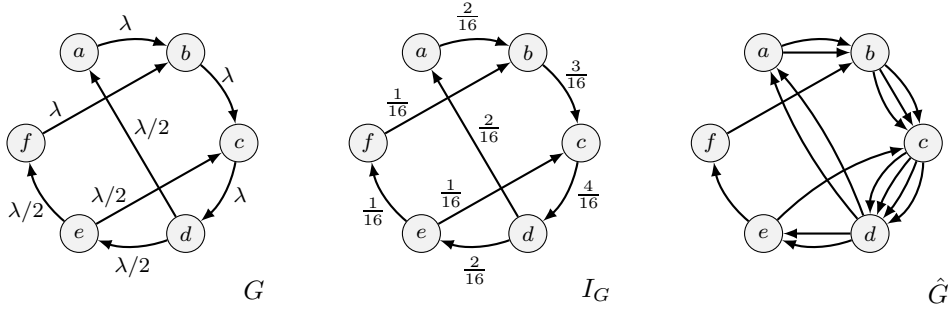

To this end, observe that since the impact of each edge is rational, there exist $N \in \mathbb{N}$ such that for every $e \in E$ the product $N \cdot I_G(e)$ is an integer.
Building upon this, let us define an auxiliary unweighted multi-graph, $\hat{G} = (V, \hat{E})$.
Its nodes are the nodes of graph $G$ and its multi-set of edges, $\hat{E} = (E,m)$, consists of edges in $E$ with the multiplicity of each edge $e \in E$ equal to $m(e) = N \cdot I_G(e)$ (see Figure~\ref{fig:lemma:impacts}).
Observe that
\begin{equation*}
    |\hat{E}| = \sum_{e \in E} m(e) = N \cdot \sum_{e \in E} I_G(e) = N \cdot \sum_{v \in V} SI_v(G) = N.
\end{equation*}

Now, from Equation~\eqref{eq:lemma:ev-si:super-regular-graphs:1} we get that in $\hat{G}$ every node has equal number of incoming and outgoing edges (when accounted for their multiplicity).
Hence, from Euler theorem for directed graphs, $\hat{G}$ is an Euler multi-graph.
This means that there exists an Euler walk $\varepsilon=(\varepsilon(0),\varepsilon(1),\dots,\varepsilon(N))$ of length $N$ in which $\varepsilon(0) = \varepsilon(N)$ and each edge is followed exactly once, i.e,
\[
	\big| \{ i : \varepsilon(i) \!= \! u \ \land \ \varepsilon(i+1)\! =\! v\}\big| \! =\! m(u,v) \ \ \mbox{for every } (u,v) \! \in  \! E.
\]
For each node $v \in V$ denote the indexes on which walk $\varepsilon$ visits node $v$, i.e., let
$E_v = \{ i \in \{1,\dots,N\} : \varepsilon(i)=v\}.$
Observe that the number of visits in $v$, is equal to the in-degree (or out-degree as it is equal) of $v$ in multi-graph $\hat{G}$, i.e., $|E_v| = N \cdot SI_v(G)$.

Next, based on Euler walk $\varepsilon$, let us construct a $\lambda$-out-regular cycle graph (see Figure~\ref{fig:lemma:euler-cycle-graph} for an illustration).
To this end, let us consider set of $N$ pairwise-distinct nodes $V' = \{v_1,\dots,v_N\}$ that will correspond to consecutive steps of $\varepsilon$.
For later convenience, let us take them in such a way that some of them are equal to particular nodes from $V$.
More in detail, for every $v \in V$ let node with index $i = \min (E_v)$, i.e., the index of the first step in which walk $\varepsilon$ visits node $v$, be equal to node $v$, i.e., $v_i=v$.
Now, the graph is given by $G'=(V',E',b',c')$,
where $E' = \{(v_1,v_2),\dots, (v_{N-1},v_N),(v_N,v_1)\}$,
node weights are weights of particular node visited by $\varepsilon$ divided by the total number of its visits, i.e.,
$b'(v_i) = b(\varepsilon(i))/|E_{\varepsilon(i)}|$ for every $i \in \{1,\dots,N\}$, and
$c'(e) = \lambda$ for every $e \in E'$.
Since $G'$ is a $\lambda$-out-regular cycle graph, from CY we get that
\begin{equation}
    \label{eq:lemma:ev-si:super-regular-graphs:3}
    F_v(G') = \sum_{u \in V'} b'(u) / N  = 1/N
\end{equation}

\begin{figure}[t]
\centering
\begin{tikzpicture}
  \def\sx{0.7cm} 
  \def\sy{0.6cm} 
  \def\x{0cm} 
  \def\y{0cm} 
  \def\arrdist{0.35cm}

  \tikzset{
    node_blank/.style={circle,draw,minimum size=0.5cm,inner sep=0, color=white}, 
    node/.style={circle,draw,minimum size=0.5cm,inner sep=0, fill = black!05, font=\footnotesize}, 
    node_emph/.style={circle, minimum size=0.65cm, inner sep=0, fill = black!15, font=\footnotesize}, 
    edge/.style={-latex,above,font=\footnotesize}, 
    el/.style={below,font=\footnotesize}, 
    operation/.style={sloped,>=stealth,above,font=\footnotesize},
    arrow/.style={draw, single arrow, minimum width = 0.9cm, minimum height=\x-6*\s+\s, fill=black!10},
    blank/.style={}
  }

  \node[node] (a) at (\x+1*\sx, 4*\sy + \y) {$a$};
  \node[node] (b) at (\x+3*\sx, 4*\sy + \y) {$b$};
  \node[node] (c) at (\x+5*\sx, 4*\sy + \y) {$c$};
  \node[node] (d) at (\x+7*\sx, 4*\sy + \y) {$d$};
  \node[node] (e) at (\x+9*\sx, 4*\sy + \y) {$e$};
  \node[node] (f) at (\x+11*\sx, 4*\sy + \y) {$f$};
  \node[node] (b1) at (\x+13*\sx, 4*\sy + \y) {$b'$};
  \node[node] (c1) at (\x+14*\sx, 2*\sy + \y) {$c'$};
  \node[node] (d1) at (\x+13*\sx, 0*\sy + \y) {$d'$};
  \node[node] (e1) at (\x+11*\sx, 0*\sy + \y) {$e'$};
  \node[node] (c2) at (\x+9*\sx, 0*\sy + \y) {$c''$};
  \node[node] (d2) at (\x+7*\sx, 0*\sy + \y) {$d''$};
  \node[node] (a1) at (\x+5*\sx, 0*\sy + \y) {$a'$};
  \node[node] (b2) at (\x+3*\sx, 0*\sy + \y) {$b''$};
  \node[node] (c3) at (\x+1*\sx, 0*\sy + \y) {$c'''$};
  \node[node] (d3) at (\x+0*\sx, 2*\sy + \y) {$d'''$};
  \node[blank] (blank) at (\x+14*\sx+0.2cm, -0.75cm + \y) {$G'$};
  
  \path[->,draw,thick]
  (c) edge[edge] (d)
  (c1) edge[edge, bend left=20, looseness = 0.9] (d1)
  (c2) edge[edge] (d2)
  (c3) edge[edge, bend left=20, looseness = 0.9] (d3)
  (b) edge[edge] (c)
  (b1) edge[edge, bend left=20, looseness = 0.9] (c1)
  (b2) edge[edge] (c3)
  (a) edge[edge] (b)
  (a1) edge[edge] (b2)
  (f) edge[edge] (b1)
  (e) edge[edge] (f)
  (e1) edge[edge] (c2)
  (d) edge[edge] (e)
  (d1) edge[edge] (e1)
  (d2) edge[edge] (a1)
  (d3) edge[edge, bend left=20, looseness = 0.9] (a)
  ;
  
\end{tikzpicture}
\caption{Cycle graph, $G'$, corresponding to an example Euler cycle on multi-graph $\hat{G}$ from  Fig.~\ref{fig:lemma:impacts}.
By proportional combining of nodes that are labeled with the same letter, we can obtain graph $G$ from Fig.~\ref{fig:lemma:impacts}.}
\label{fig:lemma:euler-cycle-graph}
\end{figure}

Now, we sequentially combine nodes in $G'$ that correspond to the same node in walk $\varepsilon$ to obtain graph isomorphic to $G$.
More in detail, for every $v \in V$, let us sequentially combine every node in $\{ v_i : i \in E_v \} \setminus \{v\}$ into $v$ (recall that $v$ is also $v_i$ with $i$ being the minimal index in $E_v$).
By $G''=(V'',E'',b'',c'')$ let us denote the graph resulting from conducting this sequential combining for all nodes $v \in V$.
Observe that from NC and Equation~\eqref{eq:lemma:ev-si:super-regular-graphs:3} we get
\begin{equation}
    \label{eq:lemma:ev-si:super-regular-graphs:4}
    F_{v}(G'') = \sum_{i \in E_v} F_{v_i}(G') = |E_v| \cdot 1/N = SI_v(G)    
\end{equation}
for every $v \in V.$
Hence, in order to prove that $F_v(G)=SI_v(G)$ it remains to prove that $G=G''$.

To this end, observe that indeed $V'' = V$ since all other nodes in $V'$ have been combined into one of the nodes in $V$.
As for edges, observe that for any edge $(u,v) \in E''$ there exists $i \in \{1,\dots,N\}$ such that in the construction of graph $G''$ node $v_{i-1}$ was combined into $u$ (or $u = v_{i-1}$) and node $v_i$ was combined into $v$ (or $v = v_i$).
In result, $(\varepsilon(i-1),\varepsilon(i))=(u,v)$, hence there exist $(u,v) \in E$.
Converse reasoning is analogous.
For node weights, for every $v \in V$ we have that 
$b''(v)=\sum_{v_i : i \in E_v} b'(v_i) = b(v).$
Finally, for edge weights observe that since combining nodes preserves the impact of edges (Equation~\eqref{eq:lemma:ev-si:super-regular-graphs:2}), the impact of edge $(u,v) \in E''$ is the sum of impacts of edges $(v_{i-1},v_i) \in E'$ such that $v_{i-1}$ has been combined into $u$ (or $v_{i-1}=u)$ and $v_i$ into $v$ (or $v_i=v$).
There are exactly $m(u,v)/N$ of such edges and the impact of every edge in graph $G'$ is equal to $1/N$, thus
\[
    I_{G''}(u,v) = m(u,v)/N = I_G(u,v).    
\]
Since, $(u,v)$ in $G''$ and $G$ have the same impact,
$u$ has the same Seeley index in both graphs (Equation~\eqref{eq:lemma:ev-si:super-regular-graphs:4}),
and both graphs are $\lambda$-out-regular (combining nodes preserves out-regularity),
we get that $c''(u,v)=c(u,v)$.
Thus, indeed $G''=G$ and $F_v(G)=SI_v(G)$.

It remains to prove that for any $x>0$ we have that $F_v(V,E,x \cdot b,c) = x \cdot F_v(G)$.
To this end, consider graph $x \cdot G'=(V',E',x \cdot b',c')$ which is just graph $G'$ with node weights scaled by $x$.
Then, from CY we get that $F_v(x \cdot G)=x/N$ for every $v \in V'$.
Thus, when we perform identical sequential proportional combining of nodes as to obtain graph $G$ from $G'$, but we start from $x \cdot G'$, then we obtain graph $(V,E,x \cdot b, c)$, which is $G$ with node weights scaled by $x$.
Therefore, from NC we get that $F_v(V,E,x \cdot b, c)= x \cdot F_v(G)$ for every $v \in V$.
\end{proof}



\begin{lemma}
\label{lemma:ev-si:regular-graphs}
If a centrality measure $F$ defined on $\mathcal{G}^{SI}$ (or $ \mathcal{G}^{EV}$) satisfies LOC, ED, NC, and CY, then for every $\lambda>0$ and every strongly connected $\lambda$-out-regular graph $G=(V,E, b, c)$ such that $G \in \mathcal{G}^{SI}$ (or $G \in \mathcal{G}^{EV}$), we have
$$F_v(G) = SI_v(G) = EV_v(G) \quad \mbox{for every } v \in V$$
and $F_v(V,E,x \cdot b,c) = x \cdot F_v(G)$ for every $x>0$.
\end{lemma}
\begin{proof}
First, let us restrict ourselves only to strongly connected $\lambda$-out-regular graphs with unit node weights, i.e., graphs $G=(V,E, \1_{\hat{v}}, c)$, where $\hat{v} \in V$, $\1_{\hat{v}}(\hat{v})=1$ and $\1_{\hat{v}}(u)=0$ for every $u \in V \setminus \{\hat{v}\}$.
For every such graph, let us distinguish one outgoing edge of each node $u \in V$, and denote it by $e_u$, in such a way that:
(1) there exist walk $\omega$ that begins with edge $e_u$, ends in $\hat{v}$ and does not visit $u$ again before reaching $\hat{v}$,
(2) among edges satisfying condition (1) the number of other outgoing edges of $u$ with weights that are not a rational multiple of the weight of $e_u$, i.e., $|\{ e \in \Gamma^+_u(G) : c(e)/c(e_u) \not \in \mathbb{Q}\}|$, is minimal.
By $k_G$ let us denote the sum of the numbers of not-rationally-proportional edges for all nodes, i.e., let
$$ k_G = \sum_{u \in V} |\{ e \in \Gamma^+_u(G) : c(e)/c(e_u) \not \in \mathbb{Q}\}|. $$
We will prove the thesis for every strongly connected $\lambda$-out-regular graph with unit node weights by induction on $k_G$.

To this end, observe that if $k_G = 0$, then for every node $u \in V$ the weight of each edge $e \in \Gamma^+_u(G)$ can be written as $c(e_u) \cdot q_e$ for some $q_e \in \mathbb{Q}$.
Thus,
$$\lambda = \deg^+_u(G) = \sum_{e \in \Gamma^+_u(G)} c(e) = c(e_u) \cdot \left(\sum_{e \in \Gamma^+_u(G)} q_e\right).$$
Hence, $\lambda/c(e_u) \in \mathbb{Q}$ which implies that also $\lambda/c(e) \in \mathbb{Q}$ for every $e \in \Gamma^+_u(G)$ for every $u \in V$.
Therefore, if we take any $e,e' \in E$, then $c(e)/c(e') = (c(e)/ \lambda)\cdot(\lambda/c(e') \in \mathbb{Q}$.
In result, the thesis follows from Lemma~\ref{lemma:ev-si:super-regular-graphs}.

Therefore, let us focus on the case in which $k_G > 0$.
Then, there exists a node $u \in V$ and its outgoing edge $\hat{e} \in E$ such that $c(\hat{e})/c(e_u) \not \in \mathbb{Q}$.
In what follows, we will construct two additional graphs: $G'$, in which edge $\hat{e}$ is removed (possibly along with a number of nodes), and $G''$, in which the weights of outgoing edges of $u$ are adjusted so that the weight of $\hat{e}$ is a rational multiple of the weight of $e_u$.
Next, we will construct graph $G$ from the combination of $G'$ and $G''$ and since both $k_{G'}$ and $k_{G''}$ are smaller than $k_G$, this will lead us to thesis from the inductive assumption.
See Figure~\ref{fig:lemma:regular} for an illustration.

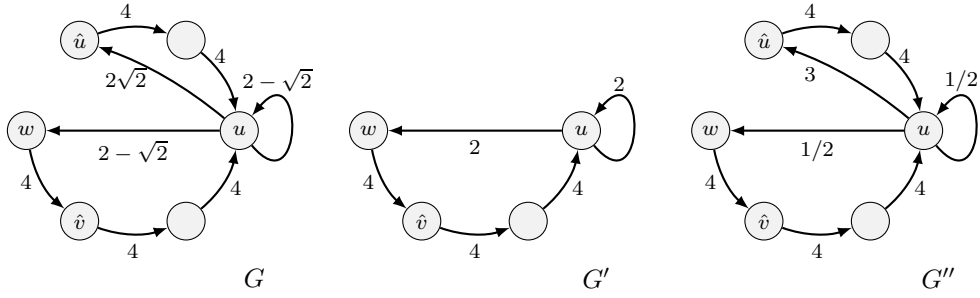
\begin{figure}[t]
\centering
\begin{tikzpicture}
  \def\sx{0.7cm} 
  \def\sy{0.6cm} 
  \def\x{0cm} 
  \def\y{0cm} 
  \def\arrdist{0.35cm}

  \tikzset{
    node_blank/.style={circle,draw,minimum size=0.5cm,inner sep=0, color=white}, 
    node/.style={circle,draw,minimum size=0.5cm,inner sep=0, fill = black!05, font=\footnotesize}, 
    node_emph/.style={circle, minimum size=0.65cm, inner sep=0, fill = black!15, font=\footnotesize}, 
    edge/.style={-latex,above,font=\footnotesize}, 
    el/.style={below,font=\footnotesize}, 
    operation/.style={sloped,>=stealth,above,font=\footnotesize},
    arrow/.style={draw, single arrow, minimum width = 0.9cm, minimum height=\x-6*\s+\s, fill=black!10},
    blank/.style={}
  } 
  
    
  \node[node] (a) at (\x+1*\sx, 4*\sy + \y) {$\hat{u}$};
  \node[node] (b) at (\x+3*\sx, 4*\sy + \y) {};
  \node[node] (c) at (\x+4*\sx, 2*\sy + \y) {$u$};
  \node[node] (d) at (\x+3*\sx, 0*\sy + \y) {};
  \node[node] (e) at (\x+1*\sx, 0*\sy + \y) {$\hat{v}$};
  \node[node] (f) at (\x+0*\sx, 2*\sy + \y) {$w$};
  \node[blank] (blank) at (\x+4*\sx+0.2cm, -0.75cm + \y) {$G$};
  
  \path[->,draw,thick]
  (a) edge[edge, bend left=20, looseness = 0.9] node[above, pos=0.4] {$4$} (b)
  (b) edge[edge, bend left=20, looseness = 0.9] node[above, pos=0.4] {$4$} (c)
  (d) edge[edge, bend left=-20, looseness = 0.9] node[right, pos =0.4] {$4$} (c)
  (e) edge[edge, bend left=-20, looseness = 0.9] node[below] {$4$} (d)
  (f) edge[edge, bend left=-20, looseness = 0.9] node[left] {$4$} (e)
  (c) edge[edge, bend left=-10, looseness = 0.9] node[below, pos = 0.8] {$2\sqrt{2}$} (a)
  (c) edge[edge] node[below] {$2 - \sqrt{2}$} (f)
  (c) edge[edge, out = -50, in = 50, looseness = 7] node[above, pos=0.8] {$2 - \sqrt{2}$} (c)
  ;

  \def\x{4.5cm} 
  
  \node[node] (c) at (\x+4*\sx, 2*\sy + \y) {$u$};
  \node[node] (d) at (\x+3*\sx, 0*\sy + \y) {};
  \node[node] (e) at (\x+1*\sx, 0*\sy + \y) {$\hat{v}$};
  \node[node] (f) at (\x+0*\sx, 2*\sy + \y) {$w$};
  \node[blank] (blank) at (\x+4*\sx+0.2cm, -0.75cm + \y) {$G'$};
  
  \path[->,draw,thick]
  (d) edge[edge, bend left=-20, looseness = 0.9] node[right, pos =0.4] {$4$} (c)
  (e) edge[edge, bend left=-20, looseness = 0.9] node[below] {$4$} (d)
  (f) edge[edge, bend left=-20, looseness = 0.9] node[left] {$4$} (e)
  (c) edge[edge] node[below] {$2$} (f)
  (c) edge[edge, out = -50, in = 50, looseness = 7] node[above, pos=0.8] {$2$} (c)
  ;
  
   \def\x{9cm} 

  \node[node] (a) at (\x+1*\sx, 4*\sy + \y) {$\hat{u}$};
  \node[node] (b) at (\x+3*\sx, 4*\sy + \y) {};
  \node[node] (c) at (\x+4*\sx, 2*\sy + \y) {$u$};
  \node[node] (d) at (\x+3*\sx, 0*\sy + \y) {};
  \node[node] (e) at (\x+1*\sx, 0*\sy + \y) {$\hat{v}$};
  \node[node] (f) at (\x+0*\sx, 2*\sy + \y) {$w$};
  \node[blank] (blank) at (\x+4*\sx+0.2cm, -0.75cm + \y) {$G''$};
  
  \path[->,draw,thick]
  (a) edge[edge, bend left=20, looseness = 0.9] node[above, pos=0.4] {$4$} (b)
  (b) edge[edge, bend left=20, looseness = 0.9] node[above, pos=0.4] {$4$} (c)
  (d) edge[edge, bend left=-20, looseness = 0.9] node[right, pos =0.4] {$4$} (c)
  (e) edge[edge, bend left=-20, looseness = 0.9] node[below] {$4$} (d)
  (f) edge[edge, bend left=-20, looseness = 0.9] node[left] {$4$} (e)
  (c) edge[edge, bend left=-10, looseness = 0.9] node[below, pos = 0.8] {$3$} (a)
  (c) edge[edge] node[below] {$1/2$} (f)
  (c) edge[edge, out = -50, in = 50, looseness = 7] node[above, pos=0.8] {$1/2$} (c)
  ;

\end{tikzpicture}
\caption{An illustration to the proof of Lemma~\ref{lemma:ev-si:regular-graphs}.
The proportion of weights of edges $\hat{e}=(u,\hat{u})$ and $e_u = (u,w)$ in $4$-out-regular graph $G$ is not rational.
Graph $G'$ is graph $G$ with edge $\hat{e}$ removed along with a part of a graph that is then disconnected from $\hat{v}$.
Graph $G''$ is graph $G$ with edge weights adjusted so that the proportion of weights of $\hat{e}$ and $e_u$ is now rational.
Since $3 > 2\sqrt{2}$, it is possible to take graphs $G'$ and $G''$ and ``combine'' them (with their node weights properly scaled) to obtain graph $G$.}
\label{fig:lemma:regular}
\end{figure}

Let us begin with graph $G'$.
Since removing just $\hat{e}$ can result in a graph that is not strongly connected, we remove $\hat{e}$ and all nodes that would not be in the same strongly connected component of the graph as node $\hat{v}$.
Formally, let $G_{-\hat{e}} = (V,E \setminus \{\hat{e}\}, \1_{\hat{v}}, c_{-\hat{e}})$ be a graph with just $\hat{e}$ removed.
Observe that since in $G$ there exists a walk that begins with $e_u$, ends in $\hat{v}$, and not passes through $u$ before it reaches $\hat{v}$, then after removal of $\hat{e}$ it is still possible to reach $\hat{v}$ from any other node, i.e., $P_{G_{-\hat{e}}}(\hat{v}) = P_G(\hat{v})=V$.
Thus, we want to remove exactly these nodes that cannot be reached from $\hat{v}$ without edge $\hat{e}$.
Hence, let $V' = S_{G_{-\hat{e}}}(\hat{v})$.
Since we remove outgoing edge of $u$, then for sure $u \in V'$.
Moreover, for every node $w \in V' \setminus \{u\}$ its successors in $G_{-\hat{e}}$ are also successors of $\hat{v}$.
Thus, we preserve all of the outgoing edges of all nodes in $V' \setminus \{u\}$, i.e., let $E' = \{(s,t) \in E : s \in V'\} \setminus \{\hat{e}\}$.
Building upon this, let us define graph $G'=(V',E',\1_{\hat{v}},c')$ in which weights of outgoing edges of $u$ are scaled so that graph is still $\lambda$-out-regular, i.e.,  $c'(e) = c(e) \cdot \lambda/(\lambda - c(\hat{e}))$ for $e \in \Gamma^+_u(G')$ and the weights of remaining edges remain unchanged, i.e, $c'(e) = c(e)$ for every $e \in E' \setminus \Gamma^+_u(G')$.
Since $\hat{e}$ is not an edge in $G'$ and the proportions of weights between remaining outgoing edges of $u$ are unchanged, we obtain that $k_{G'} < k_{G}$.
Hence, from the inductive assumption
\begin{equation}
    \label{eq:lemma:ev-si:regular-graphs:1}
    F_v(G')=SI_v(G')=EV_v(G') \quad \mbox{for every } v \in V.
\end{equation}

Now, let us construct graph $G''=(V,E,\1_{\hat{v}},c'')$ in which weight of edge $\hat{e}$ is scaled by $x>1$ and the weights of remaining outgoing edges of $u$ are scaled by $y<1$ in such a way that: (1) proportion $(x \cdot c(\hat{e}))/(y \cdot c(e_u))$ is now rational, and (2) the sum of the weights of outgoing edges of $u$ is still equal to $\lambda$ so that $G''$ is still $\lambda$-out-regular.
To this end, take any $q \in \mathbb{Q}$ such that $q > c(\hat{e})/c(e_u)$.
Then, the new edge weights are given by: $c''(\hat{e})=q \cdot c(e_u) \cdot \lambda / (\lambda - c(\hat{e}) + qc(e_u))$, $c''(e)= c(e) \cdot \lambda / (\lambda - c(\hat{e}) + qc(e_u))$ for every $e \in \Gamma^+_u(G) \setminus \{\hat{e}\}$ and $c''(e)=c(e)$ for every $e \in E \setminus \Gamma^+_u(G)$.
Observe that indeed $c''(\hat{e})/c''(e_u)=q \in \mathbb{Q}$ and that the sum of the weights of the outgoing edges of $u$ is equal to $\lambda$ which makes $G''$ $\lambda$-our-regular.
Also, it can be calculated that since $q> c(\hat{e})/c(e_u)$, we have that $c''(\hat{e})>c(\hat{e})$.
As $c''(\hat{e})/c''(e_u) \in \mathbb{Q}$ and the proportions of other edge weights did not change, it means that $k_{G''}<k_G$.
Hence, from the inductive assumption
\begin{equation}
    \label{eq:lemma:ev-si:regular-graphs:2}
    F_v(G'')=SI_v(G'')=EV_v(G'') \quad \mbox{for every } v \in V.
\end{equation}

Now, through a combination of graph $G'$ and $G''$ we will obtain graph $G$.
To this end, let us denote graph $p \cdot G' = (V',E',p \cdot \1_{\hat{v}},c')$, i.e., graph $G'$ with node weights scaled by $p$.
From inductive assumption we know that $F_u(p \cdot G') = p \cdot F_u(G')$.
In order to combine graphs $p \cdot G'$ and $G''$ we have to be able to add them together.
For this purpose, let us consider graph isomorphic to $p \cdot G'$, i.e.,
$p \cdot G^\dagger=(V^\dagger,E^\dagger,p \cdot \1_{\hat{v}^\dagger},c^\dagger)$, where
$V^\dagger = \{v^\dagger : v \in V'\}$ such that $V^\dagger \cap V = \emptyset$,
$E^\dagger = \{(s^\dagger,t^\dagger) : (s,t) \in E'\}$, and
$c^\dagger(s^\dagger,t^\dagger)=c(s,t)$ for every $(s,t) \in E'$.
It is clear that $SI_{v^\dagger}(p \cdot G^\dagger)=SI_v(p \cdot G')$ for every $v \in V'$.
Thus, from inductive assumption we have that also $F_{v^\dagger}(p \cdot G^\dagger) = p \cdot F_v(G')$.

Since $c(\hat{e}) < c''(\hat{e})$, we will combine graphs $p \cdot G^\dagger$ and $G''$ in order to obtain our original graph $G$.
To this end, let us consider the sum of graphs $p \cdot G^\dagger$ and $G''$, i.e., $p \cdot G^\dagger + G''$.
Since for every $v \in V'$ we have that $\deg^+_{v^\dagger}(G^\dagger) = \deg^+_v(G'')$,
let us sequentially combine each node $v^\dagger$ into $v$ and denote the obtained graph by $G^*=(V,E, (1+p)\1_{\hat{v}},c^*)$.
From LOC, NC and the fact that $F_{v^\dagger}(p \cdot G^\dagger)=p \cdot F_v(G')$ we get that
\begin{equation}
    \label{eq:lemma:ev-si:regular-graphs:3}
    F_v(G^*) =
    \begin{cases}
        p \cdot F_v(G') + F_v(G'')  & \mbox{if } v \in V',\\
        F_v(G'')                      & \mbox{otherwise.}
    \end{cases}
\end{equation}

In what follows, we will prove that if we take the value of $p = F_u(G'')/F_u(G') ( c''(\hat{e})/c(\hat{e}) - 1)$, then the edge weights in the obtained graph are equal to the edge weights of graph $G$, i.e., $c^* = c$.
For every $(s,t) \in E \setminus \Gamma^+_u(G^*)$ observe that $c(s,t)=c'(s,t)=c''(s,t)$.
Thus, when we combine both nodes $s^\dagger$ into $s$ and $t^\dagger$ into $t$ in $p \cdot G^\dagger + G''$, the weight of edge $(s,t)$ will be preserved.
Hence, $c^*(s,t)=c(s,t)$.
For $\hat{e}$ we have
\begin{align*}
    c^*(\hat{e}) &= \frac{F_u(G'') \cdot c''(\hat{e})}{p \cdot F_u(G') + F_u(G'')} =\\
             &= \frac{F_u(G'') \cdot c''(\hat{e})}{F_u(G'')(\frac{c''(\hat{e})}{c(\hat{e})} - 1) + F_u(G'')} =\\
             &= c(\hat{e}).
\end{align*}
For other outgoing edges of $u$, i.e., $e, e' \in \Gamma^+_u(G^*) \setminus \{\hat{e}\}$, observe that the proportions of their weights are equal in all three graphs, i.e., $c(e)/c(e') = c'(e)/c'(e') = c''(e)/c''(e')$.
Thus, similarly when we combine node $u^\dagger$ into $u$ and the corresponding ends of this edges in graph $p \cdot G^\dagger + G''$ this proportions are also preserved, i.e., $c^*(e)/c^*(e')=c(e)/c(e')$.
Moreover, observe that $p \cdot G^\dagger + G''$ is $\lambda$-out-regular and proportional combining preserves out-regularity, thus graph $G^*$ is $\lambda$-out-regular as well.
Hence, the sum of weights of edges in $\Gamma^+_u(G^*) \setminus \{\hat{e}\}$ is equal to $\lambda - c(\hat{e})$.
Since the sum and the proportions of the weights of these edges are the same in both $G$ and $G^*$, weights themselves are equal as well.
In result, we obtain that $c = c^*$ which means that graph $G^*=(V,E,(1+p)\cdot \1_{\hat{v}}, c)$ is graph $G$ with node weights scaled by $(1+p)$.

Now, if we repeat the same operation, but instead of graphs $p \cdot G'$ and $G''$ we take graphs $x \cdot p \cdot G'=(V',E',x\cdot p\cdot \1_{\hat{v}},c')$ and $x \cdot G'' = (V,E,x \cdot \1_{\hat{v}},c'')$ in the same way we obtain graph $x \cdot G^* = (V,E,x \cdot (1+p) \cdot \1_{\hat{v}},c)$.
From inductive assumption we get that $F_v(x \cdot p \cdot G')=x \cdot p \cdot F_v(G')$ and $F_v(x \cdot G'')=x \cdot F_v(G'')$.
Thus, in the same way we obtained Equation~\eqref{eq:lemma:ev-si:regular-graphs:3} we get that
\begin{equation}
    \label{eq:lemma:ev-si:regular-graphs:4}
    F_v(x \cdot G^*) =
    \begin{cases}
        x \cdot p \cdot F_v(G') + x \cdot F_v(G'')  & \mbox{if } v \in V',\\
        x \cdot F_v(G'')                      & \mbox{otherwise.}
    \end{cases}
\end{equation}
In particular, for $x = \frac{1}{1+p}$ we get that $ \frac{1}{1+p} \cdot G^* = G$.
Therefore,
\begin{equation}
    \label{eq:lemma:ev-si:regular-graphs:5}
    F_v(G) =
    \begin{cases}
        (p \cdot F_v(G') + F_v(G''))/(1+p)  & \mbox{if } v \in V',\\
        F_v(G'')/(1+p)                      & \mbox{otherwise.}
    \end{cases}
\end{equation}
Since $p = F_u(G'')/F_u(G')\left(c''(\hat{e})/c(\hat{e}) - 1\right)$, then based on inductive assumption (Equations~\eqref{eq:lemma:ev-si:regular-graphs:1} and~\eqref{eq:lemma:ev-si:regular-graphs:2}) its value does not depend on the choice of centrality $F$.
Therefore, from Equations~\eqref{eq:lemma:ev-si:regular-graphs:1}, \eqref{eq:lemma:ev-si:regular-graphs:2} and~\eqref{eq:lemma:ev-si:regular-graphs:5} and the fact that Seeley index and Eigenvector centrality also satisfy our axioms (Lemmas~\ref{lemma:axioms:si} and~\ref{lemma:axioms:ev}), we get that
$F_v(G) = SI_v(G) = EV_v(G)$
for every $v \in V$.
Moreover, from Equation~\eqref{eq:lemma:ev-si:regular-graphs:4} we obtain that
$F_v(V,E,x \cdot \1_{\hat{v}},c) = x \cdot F_v(G)$ for every $x > 0$ and $v \in V$.

It remains to relax the additional assumption of unit node weights.
Let us consider an arbitrary strongly connected $\lambda$-out-regular graph $G=(V,E,b,c)$ and two cases: the first in which $\sum_{v \in V}b(v)>0$ (I), and the second in which $\sum_{v \in V}b(v)=0$ (II).

(I) If $\sum_{v \in V}b(v)>0$, then let us denote the set of nodes that have positive weight by $V^* = \{ v \in V : b(v)>0\}$.
For every $v \in V^*$ let us construct $G_v = (V,E, b(v) \cdot \1_v ,c)$.
Observe that each graph $G_v$ is a strongly connected and $\lambda$-out-regular with node weights multiplied by a constant.
Thus, from the previous part of the proof for every $u \in V$ we have
$$F_u(G_v) = b(v) \cdot F_u(V,E,\1_v,c) = SI_u(G_v) = EV_u(G_v).$$

In order to combine all graphs $G_v$ into one graph $G$, for each graph $G_v$ let us define graph $G'_v$ isomorphic to it.
More in detail, let $V'$ be a set of nodes such that $V \cap V' = \emptyset$ and that $V' = \{ u' : u \in V\}$.
Let $G'_v = (V',E',b(v) \cdot \1_{v'},c')$, where $E'=\{(u',w') : (u,w) \in E\}$ and $c'(u',w')=c(u,w)$ for every $(u,w) \in E$.
Graph $G'_v$ is also $\lambda$-out-regular with unit node weights with node weights multiplied by a constant, thus from first part of the proof we have
$$F_u(G'_v) = SI_u(G'_v) = SI_u(G_v) = F_u(G_v).$$

Building upon this, let us consider the following operation:
Let us choose one node $v \in V^*$ and take graph $G_v$ and say that at the beginning it is our current graph.
Next, for node $u \in V^* \setminus \{v\}$ let us take graph $G'_u$,
add it to the current graph,
sequentially combine node $w'$ into $w$ for all $w \in V$, and
say that the resulting graph is now the current graph.
Then, let us perform this for all nodes $u \in V^* \setminus \{v\}$ ones.
Observe that after each such addition of graph $G'_u$, the nodes, edges and edge weight of the current graph remain unchanged, only the node weights of the current graph are summed with node weights of just added graph $G'_u$
Hence, the graph that we obtain after such operation for all $u \in V^* \setminus \{v\}$ is the original graph $G$.
Now, from LOC and NC we obtain that
$$F_u(G) = \sum_{v \in V^*} F_u(G_v) = \sum_{v \in V^*} SI_u(G_v) = SI_u(G)$$
for every $u \in V$.

(II) Finally, let us consider strongly connected $\lambda$-out-regular graph $G=(V,E,b,c)$ such that $\sum_{v \in V}b(v)=0$.
For such graph Seeley index and Eigenvector centrality is equal to zero for every node.
We prove that the same is true for centrality $F$.
Assume otherwise, i.e., there exists strongly connected $\lambda$-out-regular graph $G=(V,E,b,c)$ and node $v \in V$ such that $\sum_{u \in V}b(u)=0$ and $F_v(G) > 0$.
Then, let us take node $v' \not \in V$ and consider graph 
\[
  G' = G+(\{v'\},\{(v',v')\},\1_{v'},c'),
\]
i.e., graph $G$ with an additional node with a loop.
Let $c'(v',v') = \lambda$, so that $G'$ is still $\lambda$-out-regular.
From CY we have that $F_{v'}(\{v'\},\{(v',v')\},\1_{v'},c') = 1$, thus also $F_{v'}(G')=1$ from LOC.
Now, let us combine node $v'$ into node $v$ in graph $G'$, i.e., let us take $G''=C^F_{v' \rightarrow v}(G')$.
From NC we have that $F_v(G'')=F_{v'}(G')+F_v(G')=1+F_v(G)$.
Since $G''$ is a strongly connected $\lambda$-out-regular graph with unit node weights, then from first part of the proof we get that
$1 + F_v(G) = F_v(G'')=SI_v(G'')$.
However, we know that the sum of Seeley index in a graph is the sum of weights in that graph (Proposition~\ref{proposition:walk:rec:si}), thus $\sum_{u \in V} SI_u(G'') = 1 + \sum_{u \in V} b(v) = 1$.
As $SI_v(G'')= 1 + F_v(G) > 1$ we arrive at a contradiction.
\end{proof}



\begin{lemma}
\label{lemma:si:all-graphs}
If a centrality measure $F$ defined on $\mathcal{G}^{SI}$ satisfies LOC, ED, NC, EM, and CY, then for every graph $G=(V,E, b, c) \in \mathcal{G}^{SI}$, we have
$$F_v(G) = SI_v(G) \quad \mbox{for every } v \in V.$$
\end{lemma}
\begin{proof}
Because $F$ satisfies LOC, without loss of generality, we can assume that the graph consists of one connected component.
Since we consider graphs in $\mathcal{G}^{SI}$ this means that it is strongly connected.
Let us then, take an arbitrary such graph, $G=(V,E,b,c)$, and consider $G'=(V,E,b,c')$ in which $c'(u,v) = c(u,v) / \deg^+_u(G)$ for every $(u,v) \in E$ (see Figure~\ref{fig:lemma:si:all-graphs}).
Observe that graph $G'$ is $1$-out-regular, hence from Lemma~\ref{lemma:ev-si:regular-graphs} we get that
$F_v(G')=SI_v(G')$ for every $v \in V$.
Now, graph $G$ can be obtained from $G'$ by multiplying outgoing edges of every node $v \in V$ by $\deg^+_v(G)$.
Hence, since both $F$ and Seeley index satisfy EM, we get that
$F_v(G)=F_v(G')=SI_v(G')=SI_v(G)$ for every $v \in V$.
This concludes the proof.
\end{proof}

\begin{figure}[t]
\centering
\begin{tikzpicture}
  \def\sx{0.7cm} 
  \def\sy{0.6cm} 
  \def\x{0cm} 
  \def\y{0cm} 
  \def\arrdist{0.35cm}

  \tikzset{
    node_blank/.style={circle,draw,minimum size=0.5cm,inner sep=0, color=white}, 
    node/.style={circle,draw,minimum size=0.5cm,inner sep=0, fill = black!05, font=\footnotesize}, 
    node_emph/.style={circle, minimum size=0.65cm, inner sep=0, fill = black!15, font=\footnotesize}, 
    edge/.style={-latex,above,font=\footnotesize}, 
    el/.style={below,font=\footnotesize}, 
    operation/.style={sloped,>=stealth,above,font=\footnotesize},
    arrow/.style={draw, single arrow, minimum width = 0.9cm, minimum height=\x-6*\s+\s, fill=black!10},
    blank/.style={}
  } 
    
  \node[node] (a) at (\x+1*\sx, 4*\sy + \y) {$a$};
  \node[node] (b) at (\x+3*\sx, 4*\sy + \y) {$b$};
  \node[node] (c) at (\x+4*\sx, 2*\sy + \y) {$c$};
  \node[node] (d) at (\x+3*\sx, 0*\sy + \y) {$d$};
  \node[node] (e) at (\x+1*\sx, 0*\sy + \y) {$e$};
  \node[node] (f) at (\x+0*\sx, 2*\sy + \y) {$f$};
  \node[blank] (blank) at (\x+4*\sx+0.2cm, -0.75cm + \y) {$G$};
  
  \path[->,draw,thick]
  (a) edge[edge, bend left=20, looseness = 0.9] (b)
  (a) edge[edge, bend left=20, looseness = 0.9] (f)
  (b) edge[edge, bend left=20, looseness = 0.9] (c)
  (c) edge[edge] (f)
  (d) edge[edge, bend left=-20, looseness = 0.9] (c)
  (e) edge[edge] (c)
  (e) edge[edge, bend left=-20, looseness = 0.9] (d)
  (f) edge[edge, bend left=20, looseness = 0.9] (a)
  (f) edge[edge, bend left=-20, looseness = 0.9] (e)
  (f) edge[edge, out = 230, in = 130, looseness = 7] (f)
  ;

   \def\x{6cm} 

  \node[node] (a) at (\x+1*\sx, 4*\sy + \y) {$a$};
  \node[node] (b) at (\x+3*\sx, 4*\sy + \y) {$b$};
  \node[node] (c) at (\x+4*\sx, 2*\sy + \y) {$c$};
  \node[node] (d) at (\x+3*\sx, 0*\sy + \y) {$d$};
  \node[node] (e) at (\x+1*\sx, 0*\sy + \y) {$e$};
  \node[node] (f) at (\x+0*\sx, 2*\sy + \y) {$f$};
  \node[blank] (blank) at (\x+4*\sx+0.2cm, -0.75cm + \y) {$G'$};
  
  \path[->,draw,thick]
  (a) edge[edge, bend left=20, looseness = 0.9] node[above] {$\frac{1}{2}$} (b)
  (a) edge[edge, bend left=20, looseness = 0.9] node[right] {$\frac{1}{2}$} (f)
  (b) edge[edge, bend left=20, looseness = 0.9] node[right, pos =0.3] {$1$} (c)
  (c) edge[edge] node[above] {$1$} (f)
  (d) edge[edge, bend left=-20, looseness = 0.9] node[right, pos =0.3] {$1$} (c)
  (e) edge[edge] node[above, pos=0.2] {$\frac{1}{2}$} (c)
  (e) edge[edge, bend left=-20, looseness = 0.9] node[below] {$\frac{1}{2}$} (d)
  (f) edge[edge, bend left=20, looseness = 0.9] node[left, pos=0.7] {$\frac{1}{3}$} (a)
  (f) edge[edge, bend left=-20, looseness = 0.9] node[left] {$\frac{1}{3}$} (e)
  (f) edge[edge, out = 230, in = 130, looseness = 7] node[left] {$\frac{1}{3}$} (f)
  ;

\end{tikzpicture}
\caption{An illustration to the proof of Lemma~\ref{lemma:si:all-graphs}.
The graph on the left hand side, $G$, is a strongly connected graph that is not out-regular.
Each edge of $G)$ has weight 1.
The graph on the right hand side, $G'$, is a graph obtained from $G$ by dividing the weights of outgoing edges of $v$ by $\deg^+_v(G)$ for every node $v$.
Note that $G'$ is now $1$-out-regular.}
\label{fig:lemma:si:all-graphs}
\end{figure}
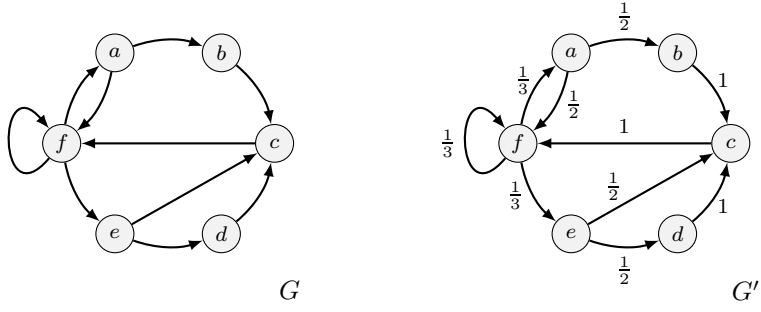


\begin{lemma}
\label{lemma:ev:all-graphs}
If a centrality measure $F$ defined on $\mathcal{G}^{EV}$ satisfies LOC, ED, NC, EC, and CY, then for every graph $G=(V,E, b, c) \in \mathcal{G}^{EV}$, we have
$$F_v(G) = EV_v(G) \quad \mbox{for every } v \in V.$$
\end{lemma}
\begin{proof}
Because $F$ satisfies LOC, without loss of generality, we can assume that the graph consists of one connected component.
Since we consider graphs in $\mathcal{G}^{EV}$ this means that it is strongly connected.
Let us then, take an arbitrary strongly connected graph $G=(V,E,b,c)$ and consider the opposite graph, i.e., graph $\bar{G} = (V, \bar{E}, b, \bar{c})$, where $\bar{E} = \{ (u,v) : (v,u) \in E\}$ and $\bar{c}(u,v) = c(v,u)$ for every $(u,v) \in \bar{E}$ (see Figure~\ref{fig:lemma:ev:all-graphs}).
Now, in graph $\bar{G}$ let us multiply the weights of outgoing edges of node $v \in V$ by $EV_v(\bar{G})$ and divide the weights of its incoming edges as well as the weight of $v$ also by $EV_v(\bar{G})$.
Because Eigenvector centrality satisfies EC, we know that this operation does not affect the centralities of nodes other than $v$ and divides the centrality of $v$ by $EV_v(\bar{G})$, making it equal to 1.
If we proceed with this operation for each node $v \in V$, then we obtain graph $\bar{G}'$ in which all nodes have Eigenvector centrality equal to 1.
Formally, $\bar{G}' = (V,\bar{E},b',\bar{c}')$ where $b'(v) = b(v) / EV_v(\bar{G})$ for every $v \in V$ and $\bar{c}'(u,v)=\bar{c}(u,v) \cdot EV_u(\bar{G})/EV_v(\bar{G})$.

Observe that if all nodes in a graph have equal centrality, then from Eigenvector centrality recursive equation~\eqref{eq:rec:ev} we get that in-degree of each node is equal to $\lambda$, i.e, $\deg^-_v(\bar{G}') = \lambda$ for every $v \in V$.
Hence, the opposite graph to $\bar{G}'$ would be $\lambda$-out-regular.
Let us define graph $G'$ as opposite to $\bar{G}'$, but with different node weights, i.e., let $G'=(V,E,b'',c')$ where $b''(v)= b(v) \cdot EV_v(\bar{G})$ for every $v \in V$ and $c'(u,v) = c(u,v) \cdot EV_v(\bar{G})/EV_u(\bar{G})$ for every $(u,v) \in E$ (see Figure~\ref{fig:lemma:ev:all-graphs} for an illustration).
Graph $G'$ is $\lambda$-out-regular, hence from Lemma~\ref{lemma:ev-si:regular-graphs} we get that
\begin{equation}
    \label{eq:lemma:ev:all-graphs:1}
    F_v(G') = EV_v(G') \quad \mbox{for every} v \in V.
\end{equation}

Now, to obtain $G$ from $G'$, for each node $v \in V$ we have to multiply the weights of outgoing edges of $v$ by $EV_v(\bar{G})$ and divide the weights of its incoming edges as well as its weight also by $EV_v(\bar{G})$.
Since both $F$ and Eigenvector centrality satisfy EC, from Equation~\eqref{eq:lemma:ev:all-graphs:1} we get that
$$ F_v(G) \! = \! F_v(G')/EV_v(\bar{G}) \! = \! EV_v(G')/EV_v(\bar{G}) \! = \! EV_v(G)$$
for every $v \in V.$
This concludes the proof.
\end{proof}

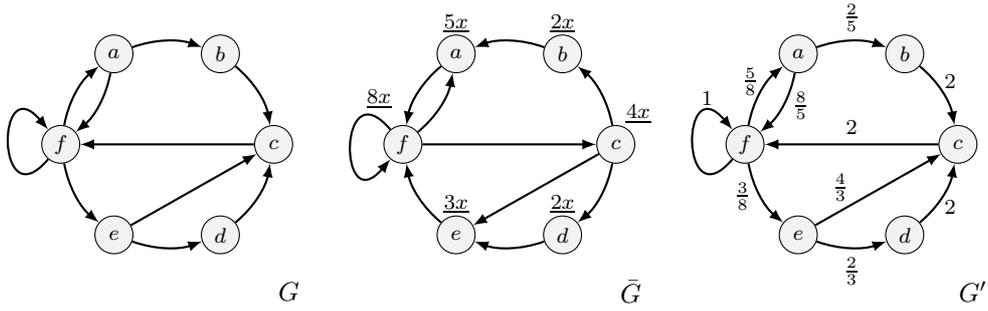
\begin{figure}[t]
\centering
\begin{tikzpicture}
  \def\sx{0.7cm} 
  \def\sy{0.6cm} 
  \def\x{0cm} 
  \def\y{0cm} 
  \def\arrdist{0.35cm}

  \tikzset{
    node_blank/.style={circle,draw,minimum size=0.5cm,inner sep=0, color=white}, 
    node/.style={circle,draw,minimum size=0.5cm,inner sep=0, fill = black!05, font=\footnotesize}, 
    node_emph/.style={circle, minimum size=0.65cm, inner sep=0, fill = black!15, font=\footnotesize}, 
    edge/.style={-latex,above,font=\footnotesize}, 
    el/.style={below,font=\footnotesize}, 
    operation/.style={sloped,>=stealth,above,font=\footnotesize},
    arrow/.style={draw, single arrow, minimum width = 0.9cm, minimum height=\x-6*\s+\s, fill=black!10},
    blank/.style={}
  } 
    
  \node[node] (a) at (\x+1*\sx, 4*\sy + \y) {$a$};
  \node[node] (b) at (\x+3*\sx, 4*\sy + \y) {$b$};
  \node[node] (c) at (\x+4*\sx, 2*\sy + \y) {$c$};
  \node[node] (d) at (\x+3*\sx, 0*\sy + \y) {$d$};
  \node[node] (e) at (\x+1*\sx, 0*\sy + \y) {$e$};
  \node[node] (f) at (\x+0*\sx, 2*\sy + \y) {$f$};
  \node[blank] (blank) at (\x+4*\sx+0.2cm, -0.75cm + \y) {$G$};
  
  \path[->,draw,thick]
  (a) edge[edge, bend left=20, looseness = 0.9] (b)
  (a) edge[edge, bend left=20, looseness = 0.9] (f)
  (b) edge[edge, bend left=20, looseness = 0.9] (c)
  (c) edge[edge] (f)
  (d) edge[edge, bend left=-20, looseness = 0.9] (c)
  (e) edge[edge] (c)
  (e) edge[edge, bend left=-20, looseness = 0.9] (d)
  (f) edge[edge, bend left=20, looseness = 0.9] (a)
  (f) edge[edge, bend left=-20, looseness = 0.9] (e)
  (f) edge[edge, out = 230, in = 130, looseness = 7] (f)
  ;

  \def\x{4.5cm} 
  
  \node[node, label = {[label distance=-0.1cm]\footnotesize \underline{$5x$}}] (a) at (\x+1*\sx, 4*\sy + \y) {$a$};
  \node[node, label = {[label distance=-0.1cm]\footnotesize \underline{$2x$}}] (b) at (\x+3*\sx, 4*\sy + \y) {$b$};
  \node[node, label = {[label distance=-0.1cm]85:\footnotesize \underline{$4x$}}] (c) at (\x+4*\sx, 2*\sy + \y) {$c$};
  \node[node, label = {[label distance=-0.1cm]\footnotesize \underline{$2x$}}] (d) at (\x+3*\sx, 0*\sy + \y) {$d$};
  \node[node, label = {[label distance=-0.1cm]\footnotesize \underline{$3x$}}] (e) at (\x+1*\sx, 0*\sy + \y) {$e$};
  \node[node, label = {[label distance=0.1cm]95:\footnotesize \underline{$8x$}}] (f) at (\x+0*\sx, 2*\sy + \y) {$f$};
  \node[blank] (blank) at ([label distance=-0.1cm]\x+4*\sx+0.2cm, -0.75cm + \y) {$\bar{G}$};
  
  \path[->,draw,thick]
  (b) edge[edge, bend left=-20, looseness = 0.9] (a)
  (f) edge[edge, bend left=-20, looseness = 0.9] (a)
  (c) edge[edge, bend left=-20, looseness = 0.9] (b)
  (f) edge[edge] (c)
  (c) edge[edge, bend left=20, looseness = 0.9] (d)
  (c) edge[edge] (e)
  (d) edge[edge, bend left=20, looseness = 0.9] (e)
  (a) edge[edge, bend left=-20, looseness = 0.9] (f)
  (e) edge[edge, bend left=20, looseness = 0.9] (f)
  (f) edge[edge, out = 130, in = 230, looseness = 7] (f)
  ;

   \def\x{9cm} 

  \node[node] (a) at (\x+1*\sx, 4*\sy + \y) {$a$};
  \node[node] (b) at (\x+3*\sx, 4*\sy + \y) {$b$};
  \node[node] (c) at (\x+4*\sx, 2*\sy + \y) {$c$};
  \node[node] (d) at (\x+3*\sx, 0*\sy + \y) {$d$};
  \node[node] (e) at (\x+1*\sx, 0*\sy + \y) {$e$};
  \node[node] (f) at (\x+0*\sx, 2*\sy + \y) {$f$};
  \node[blank] (blank) at (\x+4*\sx+0.2cm, -0.75cm + \y) {$G'$};
  
  \path[->,draw,thick]
  (a) edge[edge, bend left=20, looseness = 0.9] node[above] {$\frac{2}{5}$} (b)
  (a) edge[edge, bend left=20, looseness = 0.9] node[right] {$\frac{8}{5}$} (f)
  (b) edge[edge, bend left=20, looseness = 0.9] node[right, pos=0.3] {$2$} (c)
  (c) edge[edge] node[above] {$2$} (f)
  (d) edge[edge, bend left=-20, looseness = 0.9] node[right, pos=0.3] {$2$} (c)
  (e) edge[edge] node[above, pos=0.2] {$\frac{4}{3}$} (c)
  (e) edge[edge, bend left=-20, looseness = 0.9] node[below] {$\frac{2}{3}$} (d)
  (f) edge[edge, bend left=20, looseness = 0.9] node[left, pos=0.7] {$\frac{5}{8}$} (a)
  (f) edge[edge, bend left=-20, looseness = 0.9] node[left] {$\frac{3}{8}$} (e)
  (f) edge[edge, out = 230, in = 130, looseness = 7] node[above, pos=0.8] {$1$} (f)
  ;

\end{tikzpicture}
\caption{An illustration to the proof of Lemma~\ref{lemma:ev:all-graphs}.
The leftmost graph, $G$, is a strongly connected graph that is not out-regular.
Each edge of $G$ has weight 1.
The middle graph, $\bar{G}$, is an opposite graph to $G$, i.e., the direction of every edge is inverted. 
Underlined values near each node are equal to eigenvector centrality of this node (note that we only need the relative values).
The rightmost graph, $G'$, is a graph obtained from $G$ by dividing the weights of outgoing edges of $v$ by $EV_v(\bar{G})$ and multiplying the incoming edges of $v$ by the same value for every node $v$.
The weight of each edge is shown.
Note that graph $G'$ is now $2$-out-regular.}
\label{fig:lemma:ev:all-graphs}
\end{figure}


\subsection{Proof of Theorems~\ref{theorem:pr} and~\ref{theorem:katz}}
\label{sec:proofs:katz:pr}
In this section, we present the full proofs of Theorems~\ref{theorem:pr} and~\ref{theorem:katz} which state that our axioms uniquely characterize PageRank and Katz centrality.
We begin with Lemmas~\ref{lemma:axioms:pr} and~\ref{lemma:axioms:katz} in which we show that PageRank and Katz centrality indeed satisfy our axioms.
Then, we move to the part of the proof in which we show that our axioms are sufficient to uniquely characterize both centrality measures.
More in detail, in Lemmas~\ref{lemma:pr-katz:nc-in-sor}--\ref{lemma:pr-katz:one-arrow-graphs} we prove some basic properties that stem from our axioms.
Next, we introduce the \emph{profit function} in Definition~\ref{def:profit-function} and in Lemmas~\ref{lemma:pr-katz:two-arrow-graphs:pr}--\ref{lemma:pr-katz:recursive} show how it determines the centrality of every node.
Finally, we prove that LOC, ED, NC, EM, and BL uniquely characterize PageRank on semi-out-regular graphs (Lemma~\ref{lemma:pr:semi-out-regular}) and all graphs (Lemma~\ref{lemma:pr:final}) and that LOC, ED, NC, EC, and BL uniquely characterize Katz centrality on semi-out-regular graphs (Lemma~\ref{lemma:katz:semi-out-regular}) and all graphs (Lemma~\ref{lemma:katz:final}).


\begin{lemma}
\label{lemma:axioms:pr}
For every decay factor $a \in [0,1)$ PageRank defined by Equation~\eqref{eq:walk:pr} satisfies LOC, ED, NC, EM, and BL.
\end{lemma}
\begin{proof}
Let us take an arbitrary graph $G=(V,E,b,c) \in \mathcal{G}$ and consider axioms one by one.

For LOC the proof is analogous to the proof of LOC for Seeley index (Lemma~\ref{lemma:axioms:si}).

For ED consider edge $(u,w) \in E$ and fix $v \in V \setminus S(u)$.
By $G' = (V, E \setminus \{(u,w)\},b,c_{-(u,w)})$ let us denote graph $G$ with edge $(u,w)$ removed.
Observe that since $v$ is not a successor of $u$ in graph $G$, then a walk on $G$ that has visited node $u$ may not visit $v$ later on.
Thus, removing edge $(u,w)$ does not affect the walks of length $t$ that ends in $v$, i.e.,
$\{\omega \in \Omega_t(G) : \omega(t)=v\} =\{\omega \in \Omega_t(G') : \omega(t)=v\}$.
Moreover, for each walk $\omega \in \Omega_t(G)$ that does not visit $u$, i.e., $\omega(i)\neq u$ for every $i \in \{0,\dots,t\}$, we have that the value
\(
    b(\omega(0))\prod_{i=0}^{t-1}a \cdot c(\omega(i),\omega(i+1))/\deg^+_{\omega(i)}(G)
\)
does not change between $G$ and $G'$.
Thus, also $p^a_{v,G}(t)=p^a_{v,G'}(t)$.
Hence, $PR^a_v(G)=PR^a_v(G')$ from Equation~\eqref{eq:walk:pr}.

For NC consider nodes $u,w \in V$ such that $\deg^+_u(G)=\deg^+_w(G)=\deg^+_s(G)$ for every $s \in S(u) \cup S(v)$ and graph $G'=(V',E',b',c')=C^{PR^a}_{u \rightarrow w}(G)$.
Observe that PageRank is uniquely characterised by PageRank recursive equation (Equation~\eqref{eq:rec:pr}), thus it suffices if we prove that $(x_v)_{v \in V \setminus \{u\}}$ defined as $x_v = PR^a_v(G)$ for every $v \in V \setminus \{u,w\}$ and $x_w = PR^a_u(G) + PR^a_w(G)$ satisfy PageRank recursive equation for graph $G'$ and every node $v \in V \setminus \{u\}$.
To this end, take an arbitrary node $v \in V \setminus \{u\}$ and observe that from Equation~\eqref{eq:rec:pr} for graph $G$ and node $v$ we have
\begin{equation}
    \label{eq:lemma:axioms:pr:1}
    PR^a_v(G) = b(v) + \sum_{(s,v) \in \Gamma^-_v(G)} \frac{a \cdot c(s,v)}{\deg^+_s(G)}PR^a_s(G).
\end{equation}
If $v$ is not a direct successor of $u$ or $w$, i.e., $(u,v),(w,v) \not \in \Gamma^-_v(G)$ and $v \neq w$, then proportional combining of $u$ into $w$ does not affect the incoming edges of $v$, hence $\Gamma^-_v(G)=\Gamma^-_v(G')$.
Moreover, for every edge $(s,v) \in \Gamma^-_v(G)$ its weight is unchanged, $c'(s,v)=c(s,v)$, as well as the out-degree of node $s$, i.e., $\deg^+_s(G)=\deg^+_s(G)$.
Also, $b'(v)=b(v)$.
Thus, Equation~\eqref{eq:lemma:axioms:pr:1} leads to
\[
    x_v = b'(v) + \sum_{(s,v) \in \Gamma^-_v(G')} \frac{a c'(s,v)}{\deg^+_s(G')}x_s,
\]
which is the recursive equation for graph $G'$ and node $v$.
If, on the other hand, $v$ is a direct successor of either $u$ or $w$, but still $v \neq w$, then Equation~\eqref{eq:lemma:axioms:pr:1} can be transformed into
\begin{equation}
    \label{eq:lemma:axioms:pr:2}
    PR^a_v(G) = \frac{a  \tilde{c}(u,v)}{\deg^+_u(G)}PR^a_u(G) + \frac{a \tilde{c}(w,v)}{\deg^+_w(G)}PR^a_w(G) +
      b(v) + \sum_{\substack{(s,v) \in \Gamma^-_v(G):\\ s \not \in \{u,w\} }} \frac{a  c(s,v)}{\deg^+_s(G)}PR^a_s(G),
\end{equation}
where $\tilde{c}(e) = c(e)$ if $e \in E$ and $\tilde{c}(e)=0$ otherwise.
From the axiom condition we know that $\deg^+_u(G)=\deg^+_w(G)$.
Observe that this implies that also $\deg^+_w(G)=\deg^+_w(G')$.
Furthermore, for any $(s,v) \in \Gamma^-_v(G)$ such that $s \not \in \{u,w\}$ we have $c'(s,v)=c(s,v)$ and $\deg^+_s(G)=\deg^+_s(G')$.
Moreover, $\{(s,v) \in \Gamma^-_v(G) : s \not \in \{u,w\} \} = \{(s,v) \in \Gamma^-_v(G):s \not \in \{u,w\} \}$.
Also, $b'(v)=b(v)$.
Finally, by the definition of proportional combining we have 
$c'(w,v)=(PR^a_u(G)\tilde{c}(u,v) + PR^a_w(G)\tilde{c}(w,v))/(PR^a_u(G) + PR^a_w(G))$.
Combining these facts with Equation~\eqref{eq:lemma:axioms:pr:2} we get
\[
    x_v = b'(v) + \frac{a  c'(w,v)}{\deg^+_w(G')} x_w + \sum_{\substack{(s,v) \in \Gamma^-_v(G)\\ s \not \in \{u,w\} }} \frac{a  c'(s,v)}{\deg^+_s(G')}x_s,
\]
which is the recursive equation for graph $G'$ and node $v$.
It remains to consider node $w$.
Here, from Equation~\eqref{eq:rec:pr} we have
\[
    \label{eq:lemma:axioms:pr:3}
    PR^a_w(G) = \frac{a \tilde{c}(u,w)}{\deg^+_u(G)}PR^a_u(G) + \frac{a \tilde{c}(w,w)}{\deg^+_w(G)}PR^a_w(G) +
     b(w) + \sum_{\substack{(s,w) \in \Gamma^-_w(G):\\ s \not \in \{u,w\} }} \frac{a  c(s,w)}{\deg^+_s(G)}PR^a_s(G)
\]
and analogous equation for $u$.
From the definition of proportional combining we get
$c'(w,w) = (\! PR^a_u(G)(\tilde{c}(u,\! u) + \tilde{c}(u,\! w)\!) + PR^a_w(G)(\tilde{c}(w,\! u) + \tilde{c}(w,\! w)\!)\!)/(\! PR^a_u(G) + PR^a_w(G)\!)$.
Moreover, we have that $\deg^+_u(G)=\deg^+_w(G)=\deg^+_w(G')$,
$\{\! (s,w) \!\in \!\Gamma^-_w(G')\! :\! s \neq w\} \! =\! \{\! (s,u) \!\in \!\Gamma^-_u(G):s \!\not \in \!\{u,w\}\!\} \cup 
\{\!(s,w) \!\in \!\Gamma^-_w(G)\! :\! s \not \in \{u,w\}\!\}$,
$c'(s,w) = \tilde{c}(s,u) + \tilde{c}(s,w)$ and $\deg^+_s(G)=\deg^+_s(G')$ for every $(s,w) \in \Gamma^-_w(G')$ such that $s \neq w$.
Also, $b'(w) = b(u) + b(w)$
Hence,
\[
    x_w = b'(w) + \frac{a c'(w,w)}{\deg^+_w(G')} x_w + \sum_{\substack{(s,w) \in \Gamma^-_v(G')\\ s \neq w }} \frac{a  c'(s,w)}{\deg^+_s(G')}x_s,
\]
which is recursive equation for graph $G'$ and node $w$.

For EM the proof is analogous to the proof of EM for Seeley index (Lemma~\ref{lemma:axioms:si}).

Finally, BL follows directly from PageRank recursive equation (Equation~\eqref{eq:rec:pr}).
\end{proof}

\begin{lemma}
\label{lemma:axioms:katz}
For every decay factor $a \in \mathbb{R}_{\ge 0}$ Katz centrality defined on $\mathcal{G}^{K(a)}$ by Equation~\eqref{eq:walk:katz} satisfies LOC, ED, NC, EC, and BL.
\end{lemma}
\begin{proof}
Let us take an arbitrary graph $G=(V,E,b,c) \in \mathcal{G}^{K(a)}$ and consider axioms one by one.

For LOC the proof is analogous to the proof of LOC for Seeley index (Lemma~\ref{lemma:axioms:si}).

For ED and NC the proof is analogous to the proof that PageRank satisfies ED and NC (Lemma~\ref{lemma:axioms:pr}).

For EC the proof is analogous to the proof that Eigenvector centrality satisfies EC (Lemma~\ref{lemma:axioms:ev}).

Finally, BL follows directly from Katz centrality recursive equation (Equation~\eqref{eq:rec:katz}).
\end{proof}

\begin{definition}
Graph $G=(V,E,b,c)$ is \emph{semi-out-regular} if there exists constant $r \in \mathbb{R}_{>0}$ such that
for every $v \in V$ it holds that $\deg^+_v(G) = r$ or $\deg^+_v(G) = 0$.
\end{definition}

\begin{lemma}
\label{lemma:pr-katz:nc-in-sor}
If a centrality measure $F$ defined on $\mathcal{G}^{K(a)}$ (or $\mathcal{G}$) satisfies LOC, ED, and NC, then for every $G=(V,E,b,c)$ and nodes $u,w \in V$ such that $G \in \mathcal{G}^{K(a)}$ (or $G \in \mathcal{G}$), and
\begin{itemize}
    \item[a.] $\deg^+_u(G) = \deg^+_w(G)$ and for every $s \in S(u) \cup S(w)$ either $\deg^+_s(G)=\deg^+_u(G)$ or $\deg^+_s(G)=0$, or
    \item[b.] graph $G$ is semi-out-regular
\end{itemize}
we have that
$$F_v(C^F_{u \rightarrow w}(G))=F_v(G) \quad \mbox{for every } v \in V \setminus \{u,w\}$$
and $F_w(C^F_{u \rightarrow w}(G)) = F_u(G) + F_w(G)$.
\end{lemma}
\begin{proof}
Part \emph{b.} comes immediately from \emph{a.} since in semi-out-regular graph $G=(V,E,b,c)$ for any two nodes $u,w \in V$ we have that $\deg^+_u(G)=\deg^+_w(G)$ or $\deg^+_u(G)=0$.
Thus, let us focus on proving $\emph{a.}$

To this end, consider arbitrary nodes $u,w \in V$ such that $\deg^+_u(G) = \deg^+_w(G)$ and for every $s \in S(u) \cup S(w)$ either $\deg^+_s(G)=\deg^+_u(G)$ or $\deg^+_s(G)=0$.
If $\deg^+_u(G) = \deg^+_w(G)=0$, then the thesis comes directly from NC.
Assume otherwise, and let us denote $r = \deg^+_u(G) = \deg^+_w(G)$.

Consider nodes in $S(u)\cup S(w)$ that do not have outgoing edges, i.e., let $V' = \{v \in S(u)\cup S(w) : \Gamma^+_v(G) = \emptyset\}$.
Let us add a new node to the graph, $t \not \in V$, and an edge of weight $r$ from each node in $V'$ to $t$ and from $t$ to $t$ itself.
Formally, let $G^t = (V \cup \{t\}, E \cup E^t, b^t, c^t)$ where $E^t = \{ (v,t) : v \in V' \cup \{t\}\}$, $b^t_V=b_V$ and $b^t(t)=0$, and $c^t_E=c_E$ and $c^t(e)=r$ for every $e \in E^t$.
Observe that since nodes in $V'$ does not have outgoing edges in graph $G$, they does not have successors in $V$ i.e., $S(w) \cap V = \emptyset$ for every $w \in V'$.
Thus, from ED and LOC we have that
\begin{equation}
    \label{eq:lemma:pr-katz:semi-out-regular:1}
    F_v(G^t) = F_v(G) \quad \mbox{for every } v \in V.
\end{equation}

Now, let us perform the same operation on graph resulting from proportional combining of node $u$ int $w$.
More in detail, let us take
$\hat{G} = (\hat{V},\hat{E},\hat{b}, \hat{c}) = C^F_{u \rightarrow w}(G)$
and observe that since $\deg^+_u(G)=\deg^+_w(G)=r>0$, nodes in $V'$ are still present in graph $\hat{G}$ and still they do not have outgoing egdes.
Thus, let us consider graph
$\hat{G}^t = (\hat{V} \cup \{t\}, \hat{E} \cup E^t, \hat{b}^t, \hat{c}^t)$ where
$E^t = \{ (v,t) : v \in V' \cup \{t\}\}$,
$\hat{b}^t_V=\hat{b}_V$ and $\hat{b}^t(t)=0$ while
$\hat{c}^t_E=\hat{c}_E$ and $\hat{c}^t(e)=r$ for every $e \in E^t$.
Similarly, from ED and LOC we have that
\begin{equation}
    \label{eq:lemma:pr-katz:semi-out-regular:2}
    F_v(\hat{G}^t) = F_v(\hat{G}) \quad \mbox{for every } v \in \hat{V}.
\end{equation}
Observe that graph $\hat{G}^t$ is in fact the graph that we obtain from combining $u$ into $w$ in graph $G^t$, i.e., $\hat{G}^t = C^F_{u \rightarrow w}(G^t)$.
In $G^t$ all successors of nodes $u$ and $w$ have out-degree equal $r$, hence from NC we have
\begin{equation*}
    F_v(\hat{G}^t) =
    \begin{cases}
        F_v(G^t) & \mbox{if } v \in V \setminus \{u,w\}\\
        F_u(G^t) + F_w(G^t) & \mbox{otherwise.}
    \end{cases}
\end{equation*}
Combining this with Equations~\eqref{eq:lemma:pr-katz:semi-out-regular:1} and~\eqref{eq:lemma:pr-katz:semi-out-regular:2} yields the thesis.
\end{proof}

\begin{lemma}
\label{lemma:pr-katz:source-node}
If a centrality measure $F$ defined on $\mathcal{G}^{K(a)}$ (or $\mathcal{G}$) satisfies ED and BL, then for every graph $G=(V,E,b,c)$ and node $v \in V$ such that $G \in \mathcal{G}^{K(a)}$ (or $G \in \mathcal{G}$), and $\Gamma^-_v(G)=\emptyset$ we have
$$F_v(G) = b(v).$$
\end{lemma}
\begin{proof}
Observe that in graph $G$ node $v$ is not a successor of any node, i.e., for every $u \in V$ it holds that $v \not \in S(u)$.
Thus, in graph $G'=(V,\emptyset,b,c_\emptyset)$, i.e., graph $G$ with all edges removed, from ED we have
$F_v(G')=F_v(G)$.
In graph $G'$ node $v$ is isolated, thus the thesis follows from BL.
\end{proof}

\begin{lemma}
\label{lemma:pr-katz:positive-weight}
If a centrality measure $F$ defined on $\mathcal{G}^{K(a)}$ (or $\mathcal{G}$) satisfies LOC, ED, NC, and BL, then for every semi-out-regular graph $G=(V,E,b,c)$ and node $v \in V$ such that $G \in \mathcal{G}^{K(a)}$ (or $G \in \mathcal{G}$), and $b(v)>0$ we have
$$F_v(G) > 0.$$
\end{lemma}
\begin{proof}
Let us consider graph $G'$ with additional node $v'$ with exactly the same set of outgoing edges as $v$ in $G$, but without any incoming edges.
Also, let us transfer all node weight of node $v$ into $v'$.
Formally, let
$G'=(V \cup \{v'\}, E',b',c')$, where
$E' = E \cup \{(v',u) : (v,u) \in \Gamma^+_v(G)\}$,
$b'(v)=0$, $b'(v')=b(v)$ and $b'(w)=b(w)$ for every $w \in V \setminus \{v\}$ while
$c'_E = c_E$ and $c'(v',u)=c(v,u)$ for every $(v,u) \in \Gamma^+_v(G)$.
Clearly, $C^F_{v' \rightarrow v}(G')=G$.
Observe that $G'$ is still semi-out-regular, thus from Lemma~\ref{lemma:pr-katz:nc-in-sor}b we have $F_v(G)=F_v(G') + F_{v'}(G')$.
Now, from Lemma~\ref{lemma:pr-katz:source-node} we get that $F_{v'}(G')=b(v)>0$, thus from the fact that centrality is always non-negative we have $F_v(G)>0$.
\end{proof}

\begin{lemma}
\label{lemma:pr-katz:node-weights}
If a centrality measure $F$ defined on $\mathcal{G}^{K(a)}$ (or $\mathcal{G}$) satisfies LOC, ED, NC, and BL, then for every semi-out-regular graph $G=(V,E,b,c)$ such that $G \in \mathcal{G}^{K(a)}$ (or $G \in \mathcal{G}$) and node $v \in V$ it holds that
\begin{itemize}
    \item[a.] $F_v(V,E,b + b',c) = F_v(G) + F_v(V,E,b',c)$ for every node weights $b' : V \rightarrow \mathbb{R}_{\ge 0}$,
    \item[b.] $F_v(V,E, x \cdot b,c) = x \cdot F_v(G)$ for every $x \in \mathbb{R}_{\ge 0}$.
\end{itemize}
\end{lemma}
\begin{proof}
For \emph{a.} let us consider four semi-out-regular graphs: $G=(V,E,b,c)$, $G'=(V,E,b',c)$, $G''=(V,E,b + b',c)$ and isomorphic to them, but with different node weights $\hat{G}=(\hat{V},\hat{E},\textbf{1},\hat{c})$, where $\hat{V}=\{\hat{v}: v \in V\}$, $\hat{E}=\{(\hat{u},\hat{v}) : (u,v) \in E\}$, $\textbf{1}(\hat{v}) = 1$ for every $v \in V$, and $\hat{c}(\hat{u},\hat{v})=c(u,v)$ for every $(u,v) \in E$.

Now, using graph $\hat{G}$ we will combine together graphs $G$ and $G'$.
To this end, consider the sum of graphs $G + \hat{G}$ and then let us sequentially combine node $v$ into node $\hat{v}$ for every $v \in V$.
$G$ and $\hat{G}$ are isomorphic (when not accounting for node weights) and from Lemma~\ref{lemma:pr-katz:positive-weight} all nodes in $\hat{G}$ have positive centrality.
In result, we obtain graph $\hat{G}'=(\hat{V},\hat{E},b + \textbf{1}, \hat{c})$.
From LOC and Lemma~\ref{lemma:pr-katz:nc-in-sor}b we get
\begin{equation}
\label{eq:lemma:pr-katz:node-weights:1}
    F_{\hat{v}}(\hat{G}') = F_v(G) + F_{\hat{v}}(\hat{G}).
\end{equation}

Next, let us consider sum of graphs $\hat{G}'$ and $G'$ and this time let us sequentially combine node $\hat{v}$ into node $v$ for every $v \in V$.
Observe that both graphs are isomorphic as well (when not accounting for node weights) and from Equation~\eqref{eq:lemma:pr-katz:node-weights:1} centralities of all nodes in $\hat{G}'$ are still positive.
Hence, we obtain graph $G^*=(V,E,b + b' + \textbf{1},c)$.
Thus, from LOC, Lemma~\ref{lemma:pr-katz:nc-in-sor}b and Equation~\eqref{eq:lemma:pr-katz:node-weights:1} we have
\begin{equation}
\label{eq:lemma:pr-katz:node-weights:2}
    F_{v}(G^*)=F_v(G') + F_{\hat{v}}(\hat{G}') = F_v(G') + F_v(G) + F_{\hat{v}}(\hat{G}).
\end{equation}

On the other hand, as we will show, graph $G^*$ can be also obtained using graph $G''$ and $\hat{G}$.
To this end, consider the sum of graphs $G'' + \hat{G}$ and in this graph let us sequentially combine node $\hat{v}$ into $v$ for every $v \in V$.
Observe that in this way we also obtain graph $G^*$.
Thus, from LOC and Lemma~\ref{lemma:pr-katz:nc-in-sor}b we get
$F_{v}(G^*)=F_v(G'') + F_{\hat{v}}(\hat{G}').$
Combining this with Equation~\eqref{eq:lemma:pr-katz:node-weights:2} yields
\(
    F_v(G'') = F_v(G') + F_v(G)
\)
which concludes the proof of this part.

For \emph{b.} consider function $f(x)=F_v(V,E,x \cdot b,c)$.
From \emph{a.} we know that function $f$ is additive, i.e., $f(x+y)=f(x)+f(y)$ for every $x,y \in \mathbb{R}_{\ge 0}$.
From the definition of centrality measure, we know it is also non-negative, i.e., $f(x)\ge 0$ for every $x \in \mathbb{R}_{\ge 0}$.
This implies that the function is of the form $f(x)=x \cdot r$ for some $r \in \mathbb{R}_{\ge 0}$ \citep{Cauchy:1821}.
Since, $f(1)=F_v(G)$, we get that $F_v(V,E,x \cdot b,c) = x \cdot F_v(G)$ for every $v \in V$.
\end{proof}

\begin{lemma}
\label{lemma:pr-katz:one-arrow-graphs}
If a centrality measure $F$ defined on $\mathcal{G}^{K(a)}$ (or $\mathcal{G}$) satisfies LOC, NC, and BL, then there exists a constant $a_F \in \mathbb{R}_{\ge 0}$ such that for every $x \ge 0$, nodes $u,v$, and graph $G=(\{u,v\},\{(u,v)\},b,c)$ where $b(u)=x$, $b(v)=0$, and  $c(u,v)=1$ we have
$$F_v(G) = a_F \cdot x.$$
\end{lemma}
\begin{proof}
Let us denote $a_F = F_v(\{u,v\},\{(u,v)\},b,c)$ where $b(u)=c(u,v)=1$ and $b(v)=0$.
Since $G$ is semi-out-regular, the thesis follows from Lemma~\ref{lemma:pr-katz:node-weights}b.
\end{proof}



\begin{definition}
\label{def:profit-function}
The \emph{profit function} of centrality $F$ is a function that for every $x,y,z \in \mathbb{R}_{\ge 0}$ such that $y \le z$ returns the value $p_F(x,y,z) = F_v(G)$, where 
\[
    G= 
    \begin{cases}
        ( \{u,v,w\}, \{ (u,v), (u,w) \},b,c) & \mbox{if } y<z, \\
        ( \{u,v,w\}, \{ (u,v)\},b,c) & \mbox{otherwise,}
    \end{cases}
\]
with $b(u)=x$ and $b(v)=b(w)=0$ while $c(u,v)=y$ and if $y<z$ also $c(u,w)=z-y$.
\end{definition}

\begin{figure}[t]
\centering
\begin{tikzpicture}
  \def\x{0.7cm} 
  \def\y{0cm} 

  \tikzset{
    node_blank/.style={circle,draw,minimum size=0.5cm,inner sep=0, color=white}, 
    node/.style={circle,draw,minimum size=0.45cm,inner sep=0, fill = black!05, font=\footnotesize}, 
    edge/.style={-latex,above,font=\footnotesize},
    arrow/.style={draw, single arrow, minimum width = 0.9cm, minimum height=\y-6*\x+\s, fill=black!10},
    blank/.style={}
  }

  \node[node, label={180:\footnotesize $0$}] (v) at (\y + 0*\x, 3*\x) {$v$};
  \node[node, label={180:\footnotesize $x$}] (u) at (\y + 0*\x, 0*\x) {$u$};
  \node[node, label={0:\footnotesize $0$}] (w) at (\y + 1*\x, 3*\x) {$w$};
  
  \node[blank] (G) at (\y + 1*\x-0.02cm, -0.7cm) {$G$};

  \path[->,draw,thick]
  (u) edge[edge] node[left] {$y$}  (v)
  ;
  
  \def\y{3cm} 
  \node[blank] (I) at (0.5*\y+0.5*\x, 4*\x) {Case (I)};
  
  \node[node, label={180:\footnotesize $0$}] (v) at (\y + 0*\x, 3*\x) {$v$};
  \node[node, label={180:\footnotesize $x$}] (u) at (\y + 0*\x, 0*\x) {$u$};
  
  \node[blank] (G) at (\y + 1*\x-0.02cm, -0.7cm) {$G'$};

  \path[->,draw,thick]
  (u) edge[edge] node[left] {$1$}  (v)
  ;

  \def\y{7.5cm} 
  
  \node[node, label={180:\footnotesize $0$}] (v) at (\y + 0*\x, 3*\x) {$v$};
  \node[node, label={180:\footnotesize $\frac{x\cdot y}{z}$}] (u) at (\y + 0*\x, 0*\x) {$u$};
  \node[node, label={0:\footnotesize $0$}] (w) at (\y + 1*\x, 3*\x) {$w$};
  \node[node, label={0:\footnotesize $\frac{x\cdot (z-y)}{z}$}] (u_) at (\y + 1*\x, 0*\x) {$u'$};
  
  \node[blank] (G) at (\y + 1*\x-0.02cm, -0.7cm) {$G$};

  \path[->,draw,thick]
  (u) edge[edge] node[left] {$z$}  (v)
  (u_) edge[edge] node[right] {$z$}  (w)
  ;
  
  \def\y{11cm} 
  \node[blank] (II) at (9.25cm +0.5*\x, 4*\x) {Case (II)};
  
  \node[node, label={180:\footnotesize $0$}] (v) at (\y + 0*\x, 3*\x) {$v$};
  \node[node, label={180:\footnotesize $x$}] (u) at (\y + 0*\x, 0*\x) {$u$};
  \node[node, label={0:\footnotesize $0$}] (w) at (\y + 1*\x, 3*\x) {$w$};
  
  \node[blank] (G) at (\y + 1*\x-0.02cm, -0.7cm) {$G'$};

  \path[->,draw,thick]
  (u) edge[edge] node[left] {$y$}  (v)
  (u) edge[edge] node[right] {$z-y$}  (w)
  ;

\end{tikzpicture}
\caption{Graphs considered in the proof of Lemmas~\ref{lemma:pr-katz:two-arrow-graphs:pr} and~\ref{lemma:pr-katz:two-arrow-graphs:katz}.
The weight of each node and edge is shown.}
\label{fig:profit}
\end{figure}
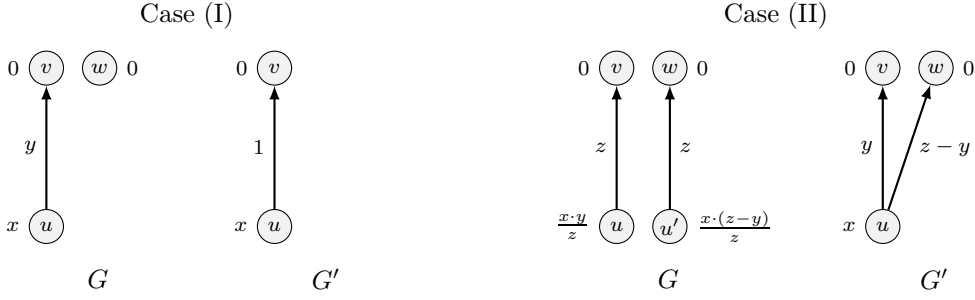

\begin{lemma}
\label{lemma:pr-katz:two-arrow-graphs:pr}
If a centrality measure $F$ defined on $\mathcal{G}$ satisfies LOC, ED, NC, EM, and BL, then for every $x,y,z \in \mathbb{R}_{\ge 0}$ such that $z \ge y$ we have
$$p_F(x,y,z) = p_{PR^{a_F}}(x,y,z)=a_F \cdot x \cdot y/z.$$
\end{lemma}
\begin{proof}
The second equality comes directly from PageRank recursive equation (Equation~\eqref{eq:rec:pr}).
Thus, let us focus on the first equality.
To this end, we will consider two cases: the first in which $y=z$, i.e., a graph with only one edge (I), and the second in which $y<z$, i.e., a graph with two edges (II).
See Figure~\ref{fig:profit} for an illustration.

(I) In the case where $y=z$, let us consider graph from Definition~\ref{def:profit-function} of the form
\[
  G=(\{u,v,w\},\{(u,v)\},b,c),
\]
where $b(u) = x$, $b(v)=b(w)=0$ and $c(u,v)=y$.
If we remove node $w$ and change the weight of edge $(u,v)$ to $1$, we obtain graph $G'=(\{u,v\},\{(u,v)\},b_{-w},1/y \cdot c).$
From Lemma~\ref{lemma:pr-katz:one-arrow-graphs} we have $F_v(G') = a_F \cdot x$.
Hence, from LOC and EM, we obtain
\begin{equation}
    \label{eq:lemma:pr-katz:two-arrow-graphs:pr:1}
    F_v(G) = F_v(G') = a_F \cdot x = PR^{a_F}_v(G).
\end{equation}

(II) In the case of $y<z$, let us begin with a graph consisting of two pairs of nodes connected by a single edge, i.e., let
$G = (\{ u, u', v, w\},\{(u,v),(u',w)\},b,c)$ where $b(u)= x \cdot y/z$, $b(u') = x \cdot (z-y)/z$, and $b(v)=b(w)=0$ while $c(u,v)=c(u',w)=z$.
Observe that from LOC and Equation~\eqref{eq:lemma:pr-katz:two-arrow-graphs:pr:1} we have that $F_v(G)= a_F \cdot x \cdot y/z$.

Now, based on Lemma~\ref{lemma:pr-katz:source-node}, we obtain that $F_u(G)=x \cdot y/z$ and also that $F_{u'}(G)=x \cdot (z-y)/z$.
Thus, if we combine node $u'$ into $u$ in $G$ we obtain graph
$G' = (\{ u,v,w\},\{(u,v),(u,w)\},b',c')$, where $b'(u)=x$, $b'(v)=b'(w)=0$, $c'(u,v)=y$, $c'(u,w)=z-y$, which is a graph from Definition~\ref{def:profit-function}.
Since $G$ is semi-out-regular, Lemma~\ref{lemma:pr-katz:nc-in-sor} yields
$p_F(x,y,z)= F_v(G') = F_v(G)= a_F \cdot x \cdot y/z= p_{PR^{a_F}}(x,y,z).$
This concludes the proof.
\end{proof}

\begin{lemma}
\label{lemma:pr-katz:two-arrow-graphs:katz}
If a centrality measure $F$ defined on $\mathcal{G}^{K(a)}$ satisfies LOC, ED, NC, EC, and BL, then for every $x,y,z \ge 0$ such that $z \ge y$ we have
$$p_F(x,y,z) = p_{K^{a_F}}(x,y,z)= a_F \cdot x \cdot y.$$
\end{lemma}
\begin{proof}
The proof follows in a similar fashion to the proof of Lemma~\ref{lemma:pr-katz:two-arrow-graphs:pr}.
The second equality comes directly from Katz centrality recursive equation (Equation~\eqref{eq:rec:katz}).
Hence, we focus on proving the first one.
To this end, we consider two cases: first in which $y=z$ (I), and second in which $y<z$ (II).
See Figure~\ref{fig:profit} for an illustration.

(I) For $y=z$, a graph from Definition~\ref{def:profit-function} is of the form $G=(\{u,v,w\},\{(u,v)\},b,c)$ where $b(u) = x$, $b(v)=b(w)=0$ and $c(u,v)=y$.
By removing node $w$ and changing the weight of edge $(u,v)$ to $1$, we obtain graph $G'=(\{u,v\},\{(u,v)\},b_{-w},1/y \cdot c).$
Lemma~\ref{lemma:pr-katz:one-arrow-graphs} yields $F_v(G') = a_F \cdot x$.
Since $F$ satisfies LOC and EC, we get
\begin{equation}
    \label{eq:lemma:pr-katz:two-arrow-graphs:katz:1}
    F_v(G) = y \cdot F_v(G') = a_F \cdot x \cdot y = K^{a_F}_v(G).
\end{equation}

(II) In the case where we have $y<z$, take graph
$G = (\{ u, u', v, w\},\! \{(u,v),\! (u',w)\},b,c)$ where $b(u)= x \cdot y/z$, $b(u') = x \cdot (z-y)/z$, and $b(v)=b(w)=0$ while $c(u,v)=c(u',w)=z$.
From Equation~\eqref{eq:lemma:pr-katz:two-arrow-graphs:katz:1} and LOC we have that $F_v(G)= a_F \cdot x \cdot y$.

From Lemma~\ref{lemma:pr-katz:source-node}, we get $F_u(G)=x \cdot y/z$ and $F_{u'}(G)=x \cdot (z-y)/z$.
Thus, by combining $u'$ into $u$ we obtain graph
$G' = (\{ u,v,w\},\{(u,v),(u,w)\},b',c')$, where $b'(u)=x$, $b'(v)=b'(w)=0$, $c'(u,v)=y$, $c'(u,w)=z-y$, which is a graph from Definition~\ref{def:profit-function}.
Since $G$ is semi-out-regular graph, from Lemma~\ref{lemma:pr-katz:nc-in-sor}b we obtain that
$p_F(x,y,z)= F_v(G') = F_v(G)= a_F \cdot x \cdot y=p_{K^{a_F}}(x,y,z).$
\end{proof}

\begin{lemma}
\label{lemma:pr-katz:recursive}
If a centrality measure $F$ defined on $\mathcal{G}^{K(a)}$ (or $\mathcal{G}$) satisfies LOC, ED, NC, BL, and EC (or EM) then for every semi-out-regular graph $G=(V,E,b,c)$ and node $v \in V$ such that $(v,v) \not \in E$ and $G \in \mathcal{G}^{K(a)}$ (or $G \in \mathcal{G}$) we have
$$F_v(G) = b(v) + \sum_{(u,v) \in E}p_F(F_u(G),c(u,v),\deg^+_u(G)).$$
\end{lemma}
\begin{proof}
We will prove the thesis by the induction on the number of incoming edges of node $v$.
If node $v$ does not have any incoming edges, then the thesis follows from Lemma~\ref{lemma:pr-katz:source-node}.
Therefore, we will focus on the case in which it has at least one edge from another node.

Let us denote one of the incoming edges of $v$ as $(u,v)$, where $u \neq v$.
In what follows, through the series of graph operation we will show that the centrality of node $v$ can be split between the centrality of sink in a graph from Definition~\ref{def:profit-function}, i.e., $p_F(F_u(G),c(u,v),\deg^+_u(G))$, and the rest that is known due to inductive assumption.

If $F_u(G)=0$, then both from Lemma~\ref{lemma:pr-katz:two-arrow-graphs:pr} and Lemma~\ref{lemma:pr-katz:two-arrow-graphs:katz} (depending on the satisfied axiom) the profit of node $v$ from edge $(u,v)$ is equal to zero, i.e., $p_F(F_u(G),c(u,v),\deg^+_u(G))=0$.
Thus, we will show that the centrality of node $v$ is equal to the sum of its profits from other edges plus its weight.
To this end, let us add small two-node graph to graph $G$, i.e., consider $G' = G + (\{u',v'\},\{(u',v'\},\textbf{1},c')$, where $\textbf{1}(u')=\textbf{1}(v')=1$ and $c'(u',v')=\deg^+_u(G)$.
From LOC we have that $F_u(G') = F_u(G) = 0$ and from Lemma~\ref{lemma:pr-katz:source-node} we have $F_{u'}(G')=1$.
Thus, when we combine node $u$ into $u'$, we get that the original outgoing edges of node $u$, including edge $(u,v)$, are removed.
Formally, let us denote the obtained graph as $G''=C^F_{u \rightarrow u'}(G')$.
Then, from Lemma~\ref{lemma:pr-katz:nc-in-sor}b and LOC we have
$$F_w(G'')=F_w(G')=F_w(G) \quad \mbox{for every } w \in V \setminus \{u\}.$$
Since in graph $G''$ node $v$ has one incoming edge less, the thesis follows from the inductive assumption.

\begin{figure}[t]
\centering
\begin{tikzpicture}
  \def\sx{0.7cm} 
  \def\sy{0.6cm} 
  \def\x{0cm} 
  \def\y{0cm} 
  \def\arrdist{0.35cm}

  \tikzset{
    node_blank/.style={circle,draw,minimum size=0.5cm,inner sep=0, color=white}, 
    node/.style={circle,draw,minimum size=0.5cm,inner sep=0, fill = black!05, font=\footnotesize}, 
    node_emph/.style={circle, minimum size=0.65cm, inner sep=0, fill = black!15, font=\footnotesize}, 
    edge/.style={-latex,above,font=\footnotesize}, 
    el/.style={below,font=\footnotesize}, 
    operation/.style={sloped,>=stealth,above,font=\footnotesize},
    arrow/.style={draw, single arrow, minimum width = 0.9cm, minimum height=\x-6*\s+\s, fill=black!10},
    blank/.style={}
  } 
    
  \node[node] (a) at (\x+1*\sx, 4*\sy + \y) {$v$};
  \node[node] (b) at (\x+3*\sx, 4*\sy + \y) {$u'$};
  \node[node] (c) at (\x+4*\sx, 2*\sy + \y) {$v'$};
  \node[node] (d) at (\x+3*\sx, 0*\sy + \y) {$u$};
  \node[node] (e) at (\x+1*\sx, 0*\sy + \y) {};
  \node[node] (f) at (\x+0*\sx, 2*\sy + \y) {};
  \node[blank] (blank) at (\x+2*\sx+0.2cm, -0.75cm + \y) {$G^*$};
  
  \path[->,draw,thick]
  (a) edge[edge, bend left=-20, looseness = 0.9] (f)
  (f) edge[edge, bend left=-20, looseness = 0.9] (e)
  (e) edge[edge, bend left=-20, looseness = 0.9] (d)
  (d) edge[edge, bend left=-20, looseness = 0.9] node[above] {$b$}  (e)
  (e) edge[edge] (a)
  (d) edge[edge] node[right] {$a$} (a)
  (b) edge[edge] node[right, pos=0.3] {$a+b$}  (c)
  ;

  \def\x{4.5cm} 
  
  \node[node] (a) at (\x+1*\sx, 4*\sy + \y) {$v$};
  \node[node] (c) at (\x+4*\sx, 2*\sy + \y) {$v'$};
  \node[node] (d) at (\x+3*\sx, 0*\sy + \y) {$u$};
  \node[node] (e) at (\x+1*\sx, 0*\sy + \y) {};
  \node[node] (f) at (\x+0*\sx, 2*\sy + \y) {};
  \node[blank] (blank) at (\x+2*\sx+0.2cm, -0.75cm + \y) {$G'' = G^\star$};
  
  \path[->,draw,thick]
  (a) edge[edge, bend left=-20, looseness = 0.9] (f)
  (f) edge[edge, bend left=-20, looseness = 0.9] (e)
  (e) edge[edge, bend left=-20, looseness = 0.9] (d)
  (d) edge[edge, bend left=-20, looseness = 0.9] node[above] {$\frac{b}{2}$}  (e)
  (e) edge[edge] (a)
  (d) edge[edge] node[right] {$\frac{a}{2}$} (a)
  (d) edge[edge, bend left=-20, looseness = 0.9] node[right, pos=0.3] {$\frac{a+b}{2}$}  (c)
  ;

   \def\x{9cm} 

  \node[node] (a) at (\x+1*\sx, 4*\sy + \y) {$v$};
  \node[node] (b) at (\x+3*\sx, 4*\sy + \y) {$u'$};
  \node[node] (c) at (\x+4*\sx, 2*\sy + \y) {$v'$};
  \node[node] (d) at (\x+3*\sx, 0*\sy + \y) {$u$};
  \node[node] (e) at (\x+1*\sx, 0*\sy + \y) {};
  \node[node] (f) at (\x+0*\sx, 2*\sy + \y) {};
  \node[blank] (blank) at (\x+2*\sx+0.2cm, -0.75cm + \y) {$G'$};
  
  \path[->,draw,thick]
  (a) edge[edge, bend left=-20, looseness = 0.9] (f)
  (f) edge[edge, bend left=-20, looseness = 0.9] (e)
  (e) edge[edge, bend left=-20, looseness = 0.9] (d)
  (b) edge[edge] node[left, pos=0.4] {$b$}  (e)
  (e) edge[edge] (a)
  (b) edge[edge, bend left=-20, looseness = 0.9] node[above] {$a$} (a)
  (d) edge[edge, bend left=-20, looseness = 0.9] node[right, pos=0.3] {$a+b$}  (c)
  ;

\end{tikzpicture}
\caption{An illustration to the first part of the proof of Lemma~\ref{lemma:pr-katz:recursive} for an example graph, $G$.
The weights of the outgoing edges of $u$ and $u'$ are shown.}
\label{fig:lemma:recursive1}
\end{figure}
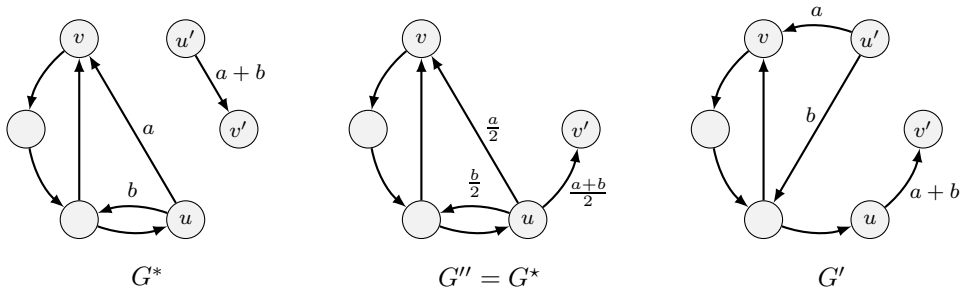

In the remainder of the proof, let us assume $F_u(G)>0$.
Consider graph $G'$ in which we split node $u$ into two nodes: $u'$ with all of its original outgoing edges, but no incoming edges, and $u$ with all of its original incoming edges, but only one new outgoing edge $(u,v')$ (see Figure~\ref{fig:lemma:recursive1}).
Formally, let $G'=(V',E',b',c')$ be a graph in which
$V'=V \cup \{u',v'\}$, $E'=E \setminus \Gamma^+_u(G) \cup \{(u',w) : (u,w) \in \Gamma^+_u(G)\} \cup \{(u,v')\}$, $b'(v')=0$, $b'(u')=F_u(G)$, and $b'_V=b_V$ while $c'(u,v')=\deg^+_G(u)$, $c'(u',w)=c(u,w)$ for every $(u,w) \in \Gamma^+_u(G)$, and $c'(e)=c(e)$ for every $e \in E \setminus \Gamma^+_u(G)$.

Now, let us combine node $u'$ into node $u$.
The graph that we obtain is not exactly our original graph $G$.
More in detail, it is graph $G''=(V \cup \{v'\}, E \cup \{(u,v')\},b'',c'')$, where $b''(v')=0$, $b''(u)=b(u)+F_u(G)$ and $b''(w)=b(w)$ for every $w \in V \setminus \{u\}$ (see Figure~\ref{fig:lemma:recursive1} for an illustration).
From Lemma~\ref{lemma:pr-katz:source-node} we have that $F_{u'}(G')=F_u(G)$, hence we know that also $c''(u,v')=\deg^+_u(G) \cdot F_u(G')/(F_u(G)+F_u(G'))$, $c''(u,w)=c(u,w) \cdot F_u(G)/(F_u(G)+F_u(G'))$ for every $(u,w) \in \Gamma^+_u(G)$ and $c''(e)=c(e)$ for every $e \in E \setminus \Gamma^+_u(G)$
(if $F_{u}(G')=0$, the edge $(u,v')$ that is supposed to have zero weight is removed, but all of the future claims still hold).
Next, we will show that graph $G''$ can be obtained also in another way.

To this end, consider adding to our original graph $G$ a simple graph consisting of nodes $u'$ and $v'$ and edge between them, and then combining node $u'$ into node $u$.
Formally, let us denote the graph with nodes added by $G^* = (V \cup \{u',v'\},E \cup \{(u',v')\},b^*,c^*)$, where $b^*_V = b_V$, $b^*(u)=F_u(G')$, and $b^*(v')=0$ while $c^*_E=c_E$ and $c^*(u,v')=\deg^+_G(u)$.
See Figure~\ref{fig:lemma:recursive1} for an illustration.
Consequently, let us combine node $u'$ into $u$ in graph $G^*$ and denote the resulting graph by $C^F_{u' \rightarrow u}(G^*)=G^\star=(V \cup \{v'\}, E \cup \{(u,v')\},b^\star, c^\star)$.
We have that $b^\star(v')=0$, $b^\star(u)=b(u)+F_u(G')$, and $b^\star(w)=b(w)$ for every $w \in V \setminus \{u\}$.
As for edge weights, from Lemma~\ref{lemma:pr-katz:source-node} we have that $F_{u'}(G^*)=F_u(G')$ and from LOC we have $F_u(G^*)=F_u(G)$.
Thus, $c^\star(u,v')=\deg^+_u(G) \cdot F_u(G')/(F_u(G)+F_u(G'))$,
$c^\star(u,w)=c(u,w) \cdot F_u(G)/(F_u(G)+F_u(G'))$ for every $(u,w) \in \Gamma^+_u(G)$, and
$c^\star(e)=c(e)$ for every $e \in E \setminus \Gamma^+_u(G)$
(if $F_{u}(G')=0$, the edge $(u,v')$ that is supposed to have zero weight is removed, but all claims still hold).

Therefore, all nodes, edges, node weights and edge weights of graphs $G^\star$ and $G''$ are identical,
except for the weight of node $u$ that in graph $G^\star$ equals $b^\star(u)=b(u)+F_u(G')$ and in graph $G''$ it is $b''(u)=b(u)+F_u(G)$.
However, the centrality of node $u$ in both graphs is the same: By Lemma~\ref{lemma:pr-katz:nc-in-sor}b, from combining $u'$ into $u$ in $G^*$ we have that $F_u(G^\star)=F_u(G')+F_u(G)$ and from combining $u'$ into $u$ in $G'$ we have that $F_u(G'')=F_u(G)+F_u(G')$.
Let us prove that this implies that also $b^\star(u)=b''(u)$.
Assume otherwise.
Without loss of generality, let us assume that $b^\star(u)>b''(u)$.
Then, consider graph $G^{\star}-G''=(V \cup \{v'\}, E \cup \{(u,v')\},b^\star-b'', c^\star)$ that is just graph $G^\star$ with node weights $b^\star-b''$ that are the difference between both node weights, i.e., $(b^\star-b'')(u)=b^\star(u)-b''(u)>0$ and $(b^\star-b'')(w)= 0$ for every $w \in V \setminus \{u\}$.
From Lemma~\ref{lemma:pr-katz:positive-weight} we have that $F_u(G^{\star}-G'')>0$.
However, from Lemma~\ref{lemma:pr-katz:node-weights}a, we have that $F_u(G'')+F_u(G^{\star}-G'')=F_u(G^\star)$.
Since $F_u(G'')=F_u(G^\star)$ we arrive at a contradiction.
Thus, $b^\star(u)=b''(u)$ which implies that $G^\star = G''$.
Therefore, from Lemma~\ref{lemma:pr-katz:nc-in-sor}b we have that
\begin{equation}
    \label{eq:lemma:pr-katz:recursive:1}
    F_w(G) = F_w(G^\star)=F_w(G'')=F_w(G') \quad \mbox{for every } w \in V.
\end{equation}

Now, let us focus on node $v$ in graph $G'$.
First, let us remove incoming edge $(u',v)$ from $v$ and add it to a new node, $v^\dagger$, with identical set of outgoing edges (see Figure~\ref{fig:lemma:recursive2}).
Formally, let us consider graph $G^\dagger=(V^\dagger,E^\dagger,b^\dagger,c^\dagger)$,
where $V^\dagger=V' \cup v^\dagger$,
$E^\dagger \! =E' \setminus \{(u',v)\}  \cup \{(u',v^\dagger)\} \cup \{(v^\dagger,w) \! : \! (v,w) \! \in \Gamma^+_v(G')\}$,
node weights remain unchanged, i.e, $b^\dagger_{V'}=b'_{V'}$, and node $v^\dagger$ has node weight $b^\dagger(v^\dagger)=0$.
Edge weights are also unchanged, i.e., $c^\dagger(u',v^\dagger)=c'(u',v)$, $c^\dagger(v^\dagger,w)=c'(v,w)$ for every $(v,w) \in \Gamma^+_v(G)$, and $c^\dagger(e)=c'(e)$ for every $e \in E' \setminus \{(u',v)\}$.
Clearly, if we combine node $v^\dagger$ into node $v$ in graph $G^\dagger$ we obtain graph $G'$ (we will prove that $v^\dagger$ has positive centrality in graph $G^\dagger$).
Thus, from Lemma~\ref{lemma:pr-katz:nc-in-sor}b we have $F_v(G')=F_v(G^\dagger) + F_{v^\dagger}(G^\dagger)$.
Node $v$ has less incoming edges in graph $G^\dagger$ than it had in graph $G$, hence from inductive assumption and Equation~\eqref{eq:lemma:pr-katz:recursive:1} we have
\[
    F_v(G) = F_{v^\dagger}(G^\dagger) + b^\dagger(v) +
    \sum_{(w,v) \in E : w \neq u} p_F(F_w(G),c^\dagger(w,v),\deg^+_w(G^\dagger)).
\]
Since $c^\dagger(w,v)=c(w,v)$ and $\deg^+_w(G^\dagger)=\deg^+_w(G)$ for every $(w,v) \in \Gamma^-_v(G)$ such that $w \neq u$ and $b^\dagger(v)=b(v)$, it remains to prove $F_{v^\dagger}(G^\dagger)=p_F(F_u(G),c(u,v),\deg^+_u(G))$.

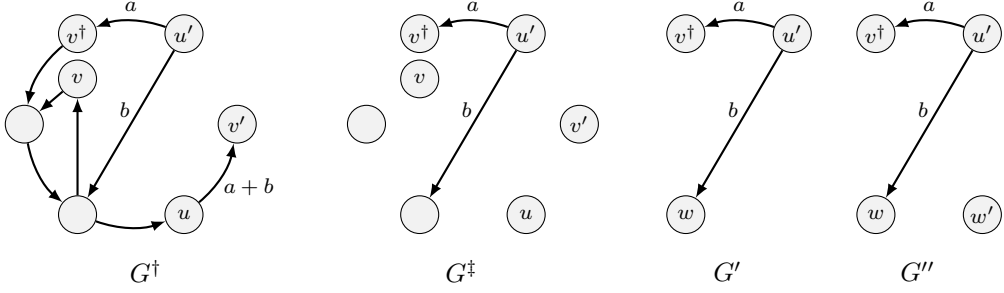
\begin{figure}[t]
\centering
\begin{tikzpicture}
  \def\sx{0.7cm} 
  \def\sy{0.6cm} 
  \def\x{0cm} 
  \def\y{0cm} 
  \def\arrdist{0.35cm}

  \tikzset{
    node_blank/.style={circle,draw,minimum size=0.5cm,inner sep=0, color=white}, 
    node/.style={circle,draw,minimum size=0.5cm,inner sep=0, fill = black!05, font=\footnotesize}, 
    node_emph/.style={circle, minimum size=0.65cm, inner sep=0, fill = black!15, font=\footnotesize}, 
    edge/.style={-latex,above,font=\footnotesize}, 
    el/.style={below,font=\footnotesize}, 
    operation/.style={sloped,>=stealth,above,font=\footnotesize},
    arrow/.style={draw, single arrow, minimum width = 0.9cm, minimum height=\x-6*\s+\s, fill=black!10},
    blank/.style={}
  }

  \def\x{0cm} 

  \node[node] (a) at (\x+1*\sx, 4*\sy + \y) {$v^\dagger$};
  \node[node] (a_) at (\x+1*\sx, 3*\sy + \y) {$v$};
  \node[node] (b) at (\x+3*\sx, 4*\sy + \y) {$u'$};
  \node[node] (c) at (\x+4*\sx, 2*\sy + \y) {$v'$};
  \node[node] (d) at (\x+3*\sx, 0*\sy + \y) {$u$};
  \node[node] (e) at (\x+1*\sx, 0*\sy + \y) {};
  \node[node] (f) at (\x+0*\sx, 2*\sy + \y) {};
  \node[blank] (blank) at (\x+2*\sx+0.2cm, -0.75cm + \y) {$G^\dagger$};
  
  \path[->,draw,thick]
  (a) edge[edge, bend left=-20, looseness = 0.9] (f)
  (a_) edge[edge] (f)
  (f) edge[edge, bend left=-20, looseness = 0.9] (e)
  (e) edge[edge, bend left=-20, looseness = 0.9] (d)
  (b) edge[edge] node[left, pos=0.4] {$b$}  (e)
  (e) edge[edge] (a_)
  (b) edge[edge, bend left=-20, looseness = 0.9] node[above] {$a$} (a)
  (d) edge[edge, bend left=-20, looseness = 0.9] node[right, pos=0.3] {$a+b$}  (c)
  ;

  \def\x{4.5cm} 

  \node[node] (a) at (\x+1*\sx, 4*\sy + \y) {$v^\dagger$};
  \node[node] (a_) at (\x+1*\sx, 3*\sy + \y) {$v$};
  \node[node] (b) at (\x+3*\sx, 4*\sy + \y) {$u'$};
  \node[node] (c) at (\x+4*\sx, 2*\sy + \y) {$v'$};
  \node[node] (d) at (\x+3*\sx, 0*\sy + \y) {$u$};
  \node[node] (e) at (\x+1*\sx, 0*\sy + \y) {};
  \node[node] (f) at (\x+0*\sx, 2*\sy + \y) {};
  \node[blank] (blank) at (\x+1.5*\sx+0.2cm, -0.75cm + \y) {$G^\ddagger$};
  
  \path[->,draw,thick]
  (b) edge[edge] node[left, pos=0.4] {$b$}  (e)
  (b) edge[edge, bend left=-20, looseness = 0.9] node[above] {$a$} (a)
  ;
  
  \def\x{8cm} 

  \node[node] (a) at (\x+1*\sx, 4*\sy + \y) {$v^\dagger$};
  \node[node] (b) at (\x+3*\sx, 4*\sy + \y) {$u'$};
  \node[node] (e) at (\x+1*\sx, 0*\sy + \y) {$w$};
  \node[blank] (blank) at (\x+1.5*\sx+0.2cm, -0.75cm + \y) {$G'$};
  
  \path[->,draw,thick]
  (b) edge[edge] node[left, pos=0.4] {$b$}  (e)
  (b) edge[edge, bend left=-20, looseness = 0.9] node[above] {$a$} (a)
  ;
  
  \def\x{10.5cm} 

  \node[node] (a) at (\x+1*\sx, 4*\sy + \y) {$v^\dagger$};
  \node[node] (b) at (\x+3*\sx, 4*\sy + \y) {$u'$};
  \node[node] (d) at (\x+3*\sx, 0*\sy + \y) {$w'$};
  \node[node] (e) at (\x+1*\sx, 0*\sy + \y) {$w$};
  \node[blank] (blank) at (\x+1.5*\sx+0.2cm, -0.75cm + \y) {$G''$};
  
  \path[->,draw,thick]
  (b) edge[edge] node[left, pos=0.4] {$b$}  (e)
  (b) edge[edge, bend left=-20, looseness = 0.9] node[above] {$a$} (a)
  ;

\end{tikzpicture}
\caption{An illustration to the second part of the proof of Lemma~\ref{lemma:pr-katz:recursive} for graph $G$ from Fig.~\ref{fig:lemma:recursive1}.}
\label{fig:lemma:recursive2}
\end{figure}

To this end, observe that the only predecessor of node $v^\dagger$ in graph $G^\dagger$ is node $u'$.
Let us denote set of edges $E^\ddagger=\Gamma^+_{u'}(G^\dagger)$, i.e., the outgoing edges of $u'$ in $G^\dagger$ and graph
$G^\ddagger=(V^\dagger,E^\ddagger, b^\dagger,c^\dagger_{E^\ddagger})$, which is graph $G^\dagger$ with all edges removed except for the outgoing edges of $u'$.
From ED, the centrality of node $v^\dagger$ is unchanged, i.e., $F_{v^\dagger}(G^\dagger)=F_{v^\dagger}(G^\ddagger)$.

If $u'$ has one outgoing edge, i.e., only edge $(u',v^\dagger)$, then by LOC and Definition~\ref{def:profit-function} we get that
$F_{v^\dagger}(G^\dagger)=p_F(F_u(G),c(u,v),deg^+_u(G))$ and the thesis follows from induction.
Hence, let us assume otherwise.
Then, observe that in graph $G^\ddagger$ every node except for $u'$ does not have outgoing edges, i.e., $\deg^+_{s}(G^\ddagger)=0$ for every $s \in V^\dagger \setminus \{u'\}$.
Therefore, let us add new isolated node to the graph, i.e., $w \not \in V$ of weight 1, and sequentially combine all nodes in $V^\dagger \setminus \{u',v^\dagger\}$ into $w$.
Formally, denote the resulting graph by $G'=(\{u',v^\dagger,w\},\{(u',v^\dagger),(u',w)\},b', c')$ where $b'(u')=F_u(G)$, $b'(v^\dagger)=0$ and $b'(w)=1 + \sum_{w' \in V^\dagger}b^\dagger(w')$ while $c'(u',v^\dagger)=c(u,v)$ and $c'(u',w)=\deg^+_u(G) - c(u,v)$.
Lemma~\ref{lemma:pr-katz:nc-in-sor}b yields $F_{v^\dagger}(G^\ddagger)=F_{v^\dagger}(G')$.

Finally, the only difference between graph $G'$ and a graph from Definition~\ref{def:profit-function} is the weight of node $w$.
Thus, let us split it into two nodes: $w$ and $w'$ such that $w$ has incoming edge $(u',w)$ and zero weight, whereas $w'$ has the weight of original node $w$, but no incoming edges.
Formally, let $G''=(\{u',v^\dagger,w,w'\},\{(u',v^\dagger),(u',w)\},b'',c')$, where $b''(u')=F_u(G)$, $b''(v^\dagger)=b''(w)=0$ and $b''(w')=b'(w)$.
Clearly, if we combine $w'$ into $w$ in this graph, then we obtain graph $G'$.
Hence, from Lemma~\ref{lemma:pr-katz:nc-in-sor}b we have $F_{v^\dagger}(G')=F_{v^\dagger}(G'')$.
On the other hand, graph $G''$ is the type of graph described in Definition~\ref{def:profit-function}, thus $F_{v^\dagger}(G'')=p_F(F_u(G),c(u,v),\deg^+_u(G)).$
This concludes the proof.
\end{proof}

\begin{lemma}
\label{lemma:pr-katz:loops}
If a centrality measure $F$ defined on $\mathcal{G}^{K(a)}$ (or $\mathcal{G}$) satisfies LOC, ED, NC, BL, and EC (or EM), then for every semi-out-regular graph $G=(V,E,b,c)$ and node $v \in V$ such that $G \in \mathcal{G}^{K(a)}$ (or $G \in \mathcal{G}$) we have
$$F_v(G) = b(v) + p_F(F_v(G),c(v,v),\deg^+_v(G)).$$
\end{lemma}
\begin{proof}
First, we focus on a graph with only two nodes and two edges: one connecting the nodes and a loop around the start (see Figure~\ref{fig:loop1}).
Formally, let $G=(\{v, w\},\{(v,v), (v,w)\},b,c)$ in which $b(w)=0$.
If $b(v)=0$ as well, then $F_v(G)=0$ from Lemma~\ref{lemma:pr-katz:node-weights}b.
Observe that both Katz centrality and PageRank profit function for $x=0$ is equal to 0,
hence the thesis follows from Lemma~\ref{lemma:pr-katz:two-arrow-graphs:pr} or~\ref{lemma:pr-katz:two-arrow-graphs:katz} (depending on the satisfied axiom).
Thus, assume $b(v)>0$ which by Lemma~\ref{lemma:pr-katz:positive-weight} means that also $F_v(G)>0$.
Let us denote the weight of edge $(v,v)$ by $y$ and the total out-degree of $v$ by $z$, i.e., $c(v,v)=y$ and $c(v,w)=z-y$.

\begin{figure}[t]
\centering
\begin{tikzpicture}
  \def\x{0.7cm} 
  \def\y{0cm} 

  \tikzset{
    node_blank/.style={circle,draw,minimum size=0.5cm,inner sep=0, color=white}, 
    node/.style={circle,draw,minimum size=0.45cm,inner sep=0, fill = black!05, font=\footnotesize}, 
    edge/.style={-latex,above,font=\footnotesize},
    arrow/.style={draw, single arrow, minimum width = 0.9cm, minimum height=\y-6*\x+\s, fill=black!10},
    blank/.style={}
  }

  \node[node, label={270:\footnotesize $b(v)$}] (v) at (\y + 0*\x, 0*\x) {$v$};
  \node[node, label={180:\footnotesize $0$}] (w) at (\y + 0*\x, 3*\x) {$w$};
  
  \node[blank] (G) at (\y + 0*\x-0.02cm, -1.3cm) {$G$};

  \path[->,draw,thick]
  (v) edge[edge] node[left] {$z-y$}  (w)
  (v) edge[edge, out=230, in=130, looseness = 5] node[left] {$y$} (v)
  ;
  
  \def\y{1.5cm} 
  
  \node[node, label={270:\footnotesize $x\cdot F_v(G)$}] (v_) at (\y + 0*\x, 0*\x) {$v'$};
  \node[node, label={270:\footnotesize $0$}] (u) at (\y + 2*\x, 0*\x) {$u$};
  \node[node, label={180:\footnotesize $0$}] (w_) at (\y + 0*\x, 03*\x) {$w'$};
  
  \node[blank] (G) at (\y + 1*\x-0.02cm, -1.3cm) {$G^x$};

  \path[->,draw,thick]
  (v_) edge[edge] node[above] {$y$}  (u)
  (v_) edge[edge] node[left] {$z-y$}  (w_)
  ;

  \def\y{5.7cm} 
  
  \node[node, label={270:\footnotesize $b(v) + x\cdot F_v(G)$}] (v) at (\y + 0*\x, 0*\x) {$v$};
  \node[node, label={270:\footnotesize $0$}] (u) at (\y + 2*\x, 0*\x) {$u$};
  \node[node, label={180:\footnotesize $0$}] (w) at (\y + 0*\x, 3*\x) {$w$};
  
  \node[blank] (G) at (\y + 0.5*\x-0.02cm, -1.3cm) {${G^x}'$};

  \path[->,draw,thick]
  (v) edge[edge] node[above] {$\frac{y\cdot x}{1+x}$}  (u)
  (v) edge[edge, out=230, in=130, looseness = 5] node[left] {$\frac{y}{1+x}$} (v)
  (v) edge[edge] node[left] {$z-y$}  (w)
  ;
  
  \def\y{9cm} 
 
  \node[node, label={270:\footnotesize $0$}] (v_) at (\y + 2*\x, 0*\x) {$v'$};
  \node[node, label={270:\footnotesize $F_v(G)$}] (v) at (\y + 0*\x, 0*\x) {$v$};
  \node[node, label={270:\footnotesize $0$}] (u) at (\y + 4*\x, 0*\x) {$u$};
  \node[node, label={180:\footnotesize $0$}] (w) at (\y + 0*\x, 3*\x) {$w$};
  
  \node[blank] (G) at (\y + 2*\x-0.02cm, -1.3cm) {$G'$};

  \path[->,draw,thick]
  (v) edge[edge] node[above] {$y$}  (v_)
  (v) edge[edge] node[left] {$z-y$}  (w)
  (v_) edge[edge] node[above] {$y$}  (u)
  (v_) edge[edge] node[right] {$z-y$}  (w)
  ;

\end{tikzpicture}
\caption{Graphs considered in the first part of the proof of Lemma~\ref{lemma:pr-katz:loops}.
The weight of each node and edge is shown.}
\label{fig:loop1}
\end{figure}
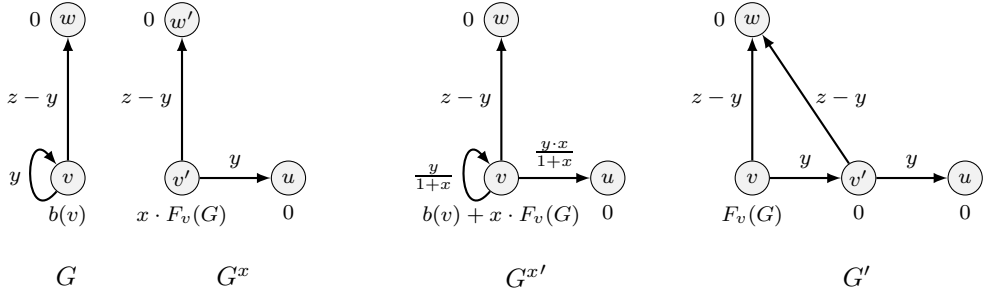

Next, for an arbitrary $x \in \mathbb{R}_{\ge 0}$ let us consider graph $G^x = (\{v',u,w'\},\{(v',u),(v',w)\},b^x,c^x)$, where $b^x(v')=x \cdot F_v(G)$, $b^x(u)=b^x(w')=0$ while $c^x(v',u)=y$ and $c^x(v',w')=z-y$.
See Figure~\ref{fig:loop1} for an illustration.
Let us add both graphs together to obtain $G + G'$.
From Lemma~\ref{lemma:pr-katz:source-node} we have $F_{v'}(G + G')= x \cdot F_v(G)$ and from LOC $F_v(G + G')=F_v(G)$.
Thus, when we combine $v'$ into $v$ and then $w'$ into $w$ we obtain graph 
${G^x}'=(\{v,u,w\},\{(v,v),(v,u), (v,w)\},b^{x*},{c^x}')$ in which
$b^{x*}(v)=b(v) + x \cdot F_v(G)$ and $b^{x*}(u)=b^{x*}(w)=0$ while
${c^x}'(v,v)= y/(1+x)$, ${c^x}'(v,u)=y \cdot x/(1+x)$, and ${c^x}'(v,w)=z-y$.
Moreover, graph $G+G^x$ is semi-out-regular, thus from Lemma~\ref{lemma:pr-katz:nc-in-sor}b
\begin{equation}
    \label{eq:lemma:pr-katz:loops:1}
    F_v({G^x}')=F_v({G^x})+F_v(G)=(1 + x) \cdot F_v(G).
\end{equation}
Now, we will show that graph ${G^x}'$ for specific $x$ can be obtained also in another way.

To this end, consider graph $G'=(\{v,v',u,w\},\{(v,v'),(v,w)(v',u),(v',w)\},b',c')$ in which $b'(v)=F_v(G)$, $b'(v')=b'(w)=b'(u)=0$, $c'(v,v')=c'(v',u)=y$, and $c'(v,w)=c'(v',w)=z-y$
(see Figure~\ref{fig:loop1}).
Now, observe that $v'$ is not a successor of itself in $G'$.
Hence, if we remove node $u$ and edges $(v',u),(v',w)$ from graph $G'$, then by ED and LOC the centrality of $v'$ does not change.
What remains is a graph from Definition~\ref{def:profit-function}, thus $F_{v'}(G')=p_F(F_{v}(G),y,z)$.
To get a particular value, let us consider two cases depending on an axiom that centrality $F$ satisfies: EM (I) or EC (II).

(I) If centrality $F$ satisfies EM, then from Lemma~\ref{lemma:pr-katz:two-arrow-graphs:pr} we obtain that $F_{v'}(G') = a_F \cdot y/z \cdot F_v(G).$
Let us denote the constant $a_F \cdot y/z$ as $\bar{x}$ (then simply $F_{v'}(G') = \bar{x} \cdot F_v(G)$).
From Lemma~\ref{lemma:pr-katz:source-node} we have that $F_v(G')=F_v(G)$.
Thus, when in $G'$ we combine node $v'$ into $v$ we obtain graph $G''=(\{v,u,w\},\{(v,v),(v,u),(v,w)\},b'',c'')$ in which
$b''(v) = F_v(G)$, $b''(u)=b''(w)=0$ while $c''(v,v)=y/(1+\bar{x})$, $c''(v,u)=y \cdot \bar{x}/(1+\bar{x})$, and $c''(v,w)=z-y$.
Moreover, from Lemma~\ref{lemma:pr-katz:nc-in-sor}b we get that
$$F_v(G'')=F_v(G') + F_w(G')=(1 +\bar{x})F_v(G).$$
Thus, from~\eqref{eq:lemma:pr-katz:loops:1} if we take $x=\bar{x}$, we get $F_v(G'')=F_v({G^{\bar{x}}}')$.
Observe that also the edge weights in both $G''$ and ${G^{\bar{x}}}'$ are the same.
Hence, we have two graphs with the same nodes, edges and edge weights in which $v$ has the same centrality.
Moreover, in both graphs only node $v$ has a positive node weight.
Thus, from Lemma~\ref{lemma:pr-katz:node-weights}b this means that $G'' = {G^{\bar{x}}}'$.
Comparing the weight of node $v$ in both graphs we get that 
$$F_v(G) = b''(v) = {b^{\bar{x}}}'(v) = b(v) + a_F \cdot F_v(G) \cdot y/z.$$
Which from Lemma~\ref{lemma:pr-katz:two-arrow-graphs:pr} concludes this part of the proof.
Observe that we also obtain that the centrality of $v$ is a linear function of its weight, i.e.,
$F_v(G)=b(v)/(1-a_F \cdot y/z)$.

(II) If $F$ satisfies EC instead of EM, then Lemma~\ref{lemma:pr-katz:two-arrow-graphs:katz} yields $F_{v'}(G') = a_F \cdot y \cdot F_v(G).$
Furthermore, if instead of $\bar{x}=a_F \cdot y/z$ we take $\bar{x}= a_F \cdot y$, then the proof follows analogously.
Furthermore, as before we obtain that the centrality of node $v$ can be seen as a linear function of its weight, i.e.,
$F_v(G)=b(v)/(1-a_F \cdot y)$.

As a result, based on both cases, (I) and (II), we can now conclude that for every graph $G=(\{v,w\},\{(v,v),(v,w)\},b,c)$ such that $b(w)=0$, $c(v,v)=y$, and $c(v,w)=z-y$ there exists a constant $r_{F,y,z}$ such that
\begin{equation}
    \label{eq:lemma:pr-katz:loops:2}
    F_v(G) = b(v) \cdot r_{F,y,z}.
\end{equation}
Observe, that if instead of a graph with two nodes that we have just considered, i.e., $G=(\{v, w\},\{(v,v), (v,w)\},b,c)$, we consider graph without edge $(v,w)$, i.e., $G=(\{v,w\},\{(v,v)\},b,c)$ then, the proof is analogous, hence Equation~\eqref{eq:lemma:pr-katz:loops:2} holds even if $y=z$.


In the remainder of the proof, let us consider arbitrary semi-out-regular graph $G=(V,E,b,c)$ and its arbitrary node $v \in V$.
If $v$ does not have a loop, then the thesis follows from Lemma~\ref{lemma:pr-katz:recursive}.
Hence, let assume otherwise, i.e., $(v,v) \in E$.
Let us denote $y=c(v,v)$ and $z = \deg^+_v(G)$.
Now, let us consider graph $\hat{G}=(\hat{V},\hat{E},\hat{b},\hat{c})$ that is modification of $G$ in which:
node $v$ does not have a loop, but it has a new outgoing edge to a new node, $w \not \in V$, and the weight of node $v$ is increased by $p_F(F_v(G),y,z)$.
See Figure~\ref{fig:loop2} for an illustration.
Formally, $\hat{V} = V \cup \{w\}$, $\hat{E} = E \setminus \{(v,v)\} \cup \{v,w\}$, $\hat{b}_{-v}=b_{-v}$, $\hat{b}(v)= b(v) + p_F(F_v(G),y,z)$, and $\hat{b}(w)=0$, while $\hat{c}_{-(v,w)} = c_{-(v,v)}$ and $\hat{c}(v,w)=c(v,v)$.
In this way, $\deg^+_v(\hat{G})=\deg^+_v(G)$ and $\hat{G}$ is still semi-out-regular.
Observe that from Lemma~\ref{lemma:pr-katz:recursive} we get
\(
    F_v(\hat{G})= \hat{b}(v) + \sum_{u \in P^1_v(G) \setminus \{v\}} p_F(F_u(\hat{G}),\hat{c}(u,v),\deg^+_u(\hat{G})).
\)
Observe that $\hat{c}(u,v)=c(u,v)$ and $\deg^+_u(\hat{G})=\deg^+_u(G)$ for every $u \in P^1_v(G) \setminus \{v\}$.
Also, $\hat{b}(v)=b(v) + p_F(F_v(G),y,z)$.
Hence, we get that
\begin{equation}
    \label{eq:lemma:pr-katz:loops:2.5}
    F_v(\hat{G})= b(v) +  p_F(F_v(G),y,z) + \sum_{u \in P^1_v(G) \setminus \{v\}} p_F(F_u(\hat{G}),c(u,v),\deg^+_u(G)).
\end{equation}
In the remainder of the proof, through a series of graph operations, we will show that $F_u(G)=F_u(\hat{G})$ for every $u \in V$.
Combined with the above equation, this will yield the thesis.

\begin{figure}[t]
\centering
\begin{tikzpicture}
  \def\x{0.7cm} 
  \def\y{0cm} 

  \tikzset{
    node_blank/.style={circle,draw,minimum size=0.5cm,inner sep=0, color=white}, 
    node/.style={circle,draw,minimum size=0.45cm,inner sep=0, fill = black!05, font=\footnotesize}, 
    edge/.style={-latex,above,font=\footnotesize},
    arrow/.style={draw, single arrow, minimum width = 0.9cm, minimum height=\y-6*\x+\s, fill=black!10},
    blank/.style={}
  }

  \node[node] (z) at (\y + 0*\x, 4*\x) {};
  \node[node] (v) at (\y + 0*\x, 2*\x) {$v$};
  \node[node] (u) at (\y + 0*\x, 0*\x) {};
  \node[node] (v_) at (\y + 2*\x, 2*\x) {$v'$};
  \node[node] (w) at (\y + 2*\x, 4*\x) {$w$};
  
  \node[blank] (G) at (\y + 1*\x-0.02cm, -1.3cm) {$G^r+G^{\uparrow r}$};

  \path[->,draw,thick]
  (v) edge[edge] node[left] {$z-y$}  (z)
  (v) edge[edge, out=230, in=130, looseness = 5] node[left] {$y$} (v)
  (u) edge[edge] (v)
  (v_) edge[edge] node[right] {$z$} (w) 
  ;
  
  \def\y{4.5cm} 
  
  \node[node] (z) at (\y + 0*\x, 4*\x) {};
  \node[node] (v) at (\y + 0*\x, 2*\x) {$v$};
  \node[node] (u) at (\y + 0*\x, 0*\x) {};
  \node[node] (w) at (\y + 2*\x, 4*\x) {$w$};
  
  \node[blank] (G) at (\y + 1*\x-0.02cm, -1.3cm) {$\tilde{G}$};

  \path[->,draw,thick]
  (v) edge[edge] node[left] {$\frac{z-y}{2}$}  (z)
  (v) edge[edge, out=230, in=130, looseness = 5] node[left] {$\frac{y}{2}$} (v)
  (u) edge[edge] (v)
  (v) edge[edge] node[below] {$\frac{z}{2}$} (w) 
  ;

  \def\y{9cm} 
 
  \node[node] (z) at (\y + 0*\x, 4*\x) {};
  \node[node] (v) at (\y + 0*\x, 2*\x) {$v$};
  \node[node] (u) at (\y + 0*\x, 0*\x) {};
  \node[node] (v_) at (\y + 2*\x, 2*\x) {$v'$};
  \node[node] (w) at (\y + 2*\x, 4*\x) {$w$};
  \node[node] (w_) at (\y + 4*\x, 4*\x) {$w'$};
  
  \node[blank] (G) at (\y + 1.5*\x-0.02cm, -1.3cm) {$\hat{G}^s+\hat{G}^{\circ s}$};

  \path[->,draw,thick]
  (v) edge[edge] node[left] {$z-y$}  (z)
  (v) edge[edge] node[left] {$y$} (w)
  (u) edge[edge] (v)
  (v_) edge[edge, out=230, in=130, looseness = 5] node[left] {$y$} (v_) 
  (v_) edge[edge] node[right] {$z-y$} (w_) 
  ;

\end{tikzpicture}
\caption{An illustration to the second part of the proof of Lemma~\ref{lemma:pr-katz:loops} for an example graph, $G$.
The weights of outgoing edges of nodes $v$ and $v'$ are shown.}
\label{fig:loop2}
\end{figure}
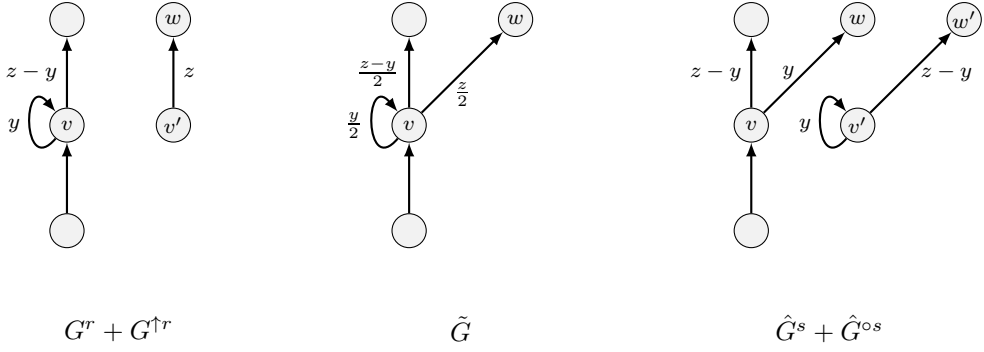

To this end, we parametrize both graphs $G$ and $\hat{G}$ by adding arbitrary weights $r,s \in \mathbb{R}_{\ge 0}$ to node $v$ in both of them respectively.
Formally, let us denote $G^r = (V,E,b^r,c)$ and $\hat{G}^s = (\hat{V},\hat{E},\hat{b}^s,\hat{c})$, where $b^r_{-v}=b_{-v}$, $b^r(v)=b(v)+r$, $\hat{b}^s_{-v}=\hat{b}_{-v}$, and $\hat{b}^s(v)=\hat{b}(v)+s$.
To both of them we will add a small two-node graph with which we will proportionally combine their nodes in order to obtain the same graph $\tilde{G}$ in both cases.

We start with $G_r$.
Take $v' \not \in \hat{V}$ and consider graph 
\(
    G^{\uparrow r} = (
        \{v',w\},\{(v',w)\},
        b^{\uparrow r}, c^{\uparrow r})
\)
where $b^{\uparrow r}(v') = F_v(G^r)$, $b^{\uparrow r}(w) = 0$, and $c^{\uparrow r}(v',w)=z$.
In the sum, $G^r + G^{\uparrow r}$, let us proportionally combine node $v'$ into $v$, i.e., let
$\tilde{G}^{r} = C^F_{v' \rightarrow v}(G^r + G^{\uparrow r})$.
Let us denote $\tilde{G}^{r} = (\tilde{V}, \tilde{E}, \tilde{b}^r, \tilde{c}^r)$.
Observe that from Locality we obtain $F_v(G^r + G^{\uparrow r})=F_v(G^r)$ and from Lemma~\ref{lemma:pr-katz:source-node} we have that also $F_{v'}(G^r + G^{\uparrow r}) = F_v(G^r)$.
This has two implications:
First, since $G^r + G^{\uparrow r}$ is semi-out-regular, from Lemma~\ref{lemma:pr-katz:nc-in-sor}b we get that
\begin{equation}
    \label{eq:lemma:pr-katz:loops:3}
    F_u(\tilde{G}^r)= 
    \begin{cases}
        2 \cdot F_v(G^r) & \mbox{if } u=v,\\
        F_u(G^r) & \mbox{if } u \in V \setminus \{v\}.
    \end{cases}
\end{equation}
Second, this means that the weights of outgoing edges of $v$ and $v'$ are divided by two in graph $\tilde{G}^r$, i.e, $\tilde{c}^{r}_{\Gamma^+_v(G)} = c_{\Gamma^+_v(G)}/2$ and $\tilde{c}^{r}(v,w) = z /2$, and the weights of other edges are unchanged $\tilde{c}^{r}_{-\Gamma^+_v(G)}=c_{-\Gamma^+_v(G)}$.
See Fig.~\ref{fig:loop2} for an illustration.

Now, let us move to graph $\hat{G}^s$.
Here, take node $w' \not \in V' \cup \{v'\}$ and consider adding graph
\(
    \hat{G}^{\circ s} = (
        \{v',w'\},\{(v',v'),(v',w')\},
        b^{\circ s},c^\circ)
    ),
\)
where
\(
    b^{\circ s}(v') = F_v(\hat{G}^s)/c_{F,y,z},
\)
$b^{\circ s}(w')=0$, $c^\circ(v',v')=y$, and $c^\circ(v',w')=z-y$.
Observe that from Equation~\eqref{eq:lemma:pr-katz:loops:2} this gives us that
\begin{equation}
    \label{eq:lemma:pr-katz:loops:4}
    F_{v'}(\hat{G}^{\circ s})=F_v(\hat{G}^s).
\end{equation}
In the sum of both graphs let us combine $w'$ into $w$ and also $v'$ into $v$.
In result, we obtain graph in which nodes and edges are the same as in $\tilde{G}$.
Thus, let us denote
$\tilde{G}^{*s} = (\tilde{V}, \tilde{E}, \tilde{b}^{*s},\tilde{c}^{*s}) = C^F_{v' \rightarrow v}(C^F_{w' \rightarrow w}(\hat{G}^s + \hat{G}^{\circ s}))$.
From Locality, we have that $F_v(\hat{G}^s + \hat{G}^{\circ s})=F_v(\hat{G}^s)$ and from Locality and Equation~\eqref{eq:lemma:pr-katz:loops:4} we get that also
$F_{v'}(\hat{G}^s + \hat{G}^{\circ s})=F_v(\hat{G}^s)$.
Again, this has two implications:
First, from Lemma~\ref{lemma:pr-katz:nc-in-sor}b we get that
\begin{equation}
    \label{eq:lemma:pr-katz:loops:5}
    F_u(\tilde{G}^{*s})= 
    \begin{cases}
        2 \cdot F_v(\hat{G}^s) & \mbox{if } u=v,\\
        F_u(\hat{G}^s) & \mbox{if } u \in V \setminus \{v\}.
    \end{cases}
\end{equation}
Second, as before, the weights of outgoing edges of $v$ and $v'$ are divided by two in graph $\tilde{G}^{*s}$, i.e, $\tilde{c}^{*s}_{\Gamma^+_v(G)} = c_{\Gamma^+_v(G)}/2$ and $\tilde{c}^{*s}(v,w) = z /2$, and the weights of other edges are unchanged $\tilde{c}^{*s}_{-\Gamma^+_v(G)}=c_{-\Gamma^+_v(G)}$.
Hence, $\tilde{c}^{*s}=\tilde{c}^r$ for every $r,s \in \mathbb{R}_{\ge 0}$.

Therefore, the only possible difference between graphs $\tilde{G}^r$ and $\tilde{G}^{*s}$ for any $r,s \in \mathbb{R}_{\ge 0}$ may lay in node weights.
However, observe that for every $r,s \in \mathbb{R}_{\ge 0}$ we have $\tilde{b}^{r}(u)=b(u)=\tilde{b}^{*s}(u)$ for every $u \in V \setminus \{v\}$ and also $\tilde{b}^{r}(w)=0=\tilde{b}^{*s}(w)$.
Thus, the only difference can lay in the weight of node $v$.
In what follows, we will show that in fact, for $s=0$ and $r=0$ it holds that $\tilde{b}^{r}(v)=\tilde{b}^{*s}(v)$.
To this end, let us assume otherwise, i.e., that either $\tilde{b}^{0}(v) > \tilde{b}^{*0}(v)$ (I), or $\tilde{b}^{0}(v) < \tilde{b}^{*0}(v)$ (II).

(I) Assume that $\tilde{b}^{0}(v) > \tilde{b}^{*0}(v)$.
The weight $\tilde{b}^{*s}(v)$ is the sum of weights of $v$ in  $\hat{G}^s$ and $v'$ in $\hat{G}^{\circ s}$, thus we get
\(
    \tilde{b}^{*s}(v) = 
    \hat{b}(v) + s + b^{\circ s}(v') =
    \hat{b}(v) + s + F_v(\hat{G}^s)/c_{F,y,z}.
\)
Also, from Lemma~\ref{lemma:pr-katz:node-weights} (a and b) we know that
\(
    F_v(\hat{G}^s)=F_v(\hat{G}) + s \cdot F_v(\hat{G}_{\1}).
\)
where $\hat{G}_{\1}$ is graph $\hat{G}$ with changed node weights so that $v$ has weight 1 and all other nodes weight 0, i.e.,
$\hat{G}_{\1} = (\hat{V},\hat{E},\1_v,\hat{c})$, where $\1_v(v)=1$ and $\1_v(u)=0$ for every $u \in \hat{V} \setminus \{v\}$.
Both facts imply that $\tilde{b}^{*s}(v)$ is a linear function of $s$.
Hence, there exists $s > 0$ such that $\tilde{b}^{*s}(v) = \tilde{b}^{0}(v)$.
Thus, for such $s$ it holds that
\(
    \tilde{G}^0=\tilde{G}^{*s}.
\)
This implies that
\(
    F_v(\tilde{G}^0)=F_v(\tilde{G}^{*s})
\)
and by Equation~\eqref{eq:lemma:pr-katz:loops:3} and Equation~\eqref{eq:lemma:pr-katz:loops:5} also
\(
    F_v(G)=F_v(\hat{G}^s).
\)
Thus, looking again at node weights, we get that
\begin{equation}
    \label{eq:lemma:pr-katz:loops:6}
    b(v) + F_v(G) = \tilde{b}^{0}(v) = \tilde{b}^{*s}(v) = s + b(v) + p_F(F_v(G),y,z) + b^{\circ s}(v').
\end{equation}
From the first part of the proof, we know that
\(
    F_{v'}(\hat{G}^{\circ s}) = p_F(F_{v'}(G^{\circ s}),y,z) + b^{\circ s}(v').
\)
From Equation~\eqref{eq:lemma:pr-katz:loops:4} we have that 
\(
    F_{v'}(\hat{G}^{\circ s}) = F_{v}(\hat{G}^s) = F_v(G).
\)
Therefore, we obtain
\[
    F_{v}(G) = p_F(F_{v}(G),y,z) + b^{\circ s}(v').
\]
Subtracting this from Equation~\eqref{eq:lemma:pr-katz:loops:6} yields $b(v) = b(v) +s$, which means that $s=0$ and this contradicts the assumption of this case.

(II) Now, assume that $\tilde{b}^{0}(v) < \tilde{b}^{*0}(v)$.
Observe that
\(
    \tilde{b}^{r}(v) = 
    b(v) + r + F_v(G^r).
\)
Also, from Lemma~\ref{lemma:pr-katz:node-weights} (a and b), we get that
\(
    F_v(G^r)=F_v(G) + r \cdot F_v(G_{\1}).
\)
where $G_{\1}$ is graph $G$ with changed node weights so that $v$ has weight 1 and all other nodes weight 0, i.e.,
$G_{\1} = (V,E,\1_v,c)$, where $\1_v(v)=1$ a $\1_v(u)=0$ for every $u \in V \setminus \{v\}$.
Combining both facts, we get that $\tilde{b}^{r}(v)$ is a linear function of $r$.
Hence, there exists $r > 0$ such that $\tilde{b}^{*0} = b^{r*}(v)$.
Thus, for such $r$ it holds that
\(
    \tilde{G}^r=\tilde{G}^{*0}.
\)
This implies that
\(
    F_v(\tilde{G}^r)=F_v(\tilde{G}^{*0})
\)
and by Equation~\eqref{eq:lemma:pr-katz:loops:3} and Equation~\eqref{eq:lemma:pr-katz:loops:5} also that
\(
    F_v(G^r)=F_v(\hat{G}^0).
\)
Thus, looking again at node weights, we get
\begin{equation}
    \label{eq:lemma:pr-katz:loops:7}
    b(v) + r + F_v(G^r) = b^{r*}(v) = \tilde{b}^{*0}(v) = b(v) + p_F(F_v(G),y,z) + b^{\circ 0}(v').
\end{equation}
From the first part of the proof, we know that
\(
    F_{v'}(G^{\circ 0}) = p_F(F_{v'}(G^{\circ 0}),y,z) + b^{\circ 0}(v').
\)
From Equation~\eqref{eq:lemma:pr-katz:loops:4} we have that 
\(
  p_F(F_{v'}(G^{\circ 0}),y,z) = F_{v}(\hat{G}^0) = F_v(G^r).
\)
Therefore, we obtain
\(
    F_{v}(G^r) = p_F(F_{v}(G^r),y,z) + b^{\circ 0}(v').
\)
Subtracting this from Equation~\eqref{eq:lemma:pr-katz:loops:7} yields
\begin{equation}
    \label{eq:lemma:pr-katz:loops:8}
    r = p_F(F_v(G),y,z) - p_F(F_v(G^r),y,z).
\end{equation}
Recall that from Lemma~\ref{lemma:pr-katz:node-weights} (a and b) we have 
\(
    F_v(G^r) = F_v(G) + r \cdot F_v(G_\1).
\)
Since $r > 0$, from Lemma~\ref{lemma:pr-katz:positive-weight} this means that
\(
    F_v(G^r) > F_v(G).
\)
Hence, also
$p_F(F_v(G^r),y,z) > p_F(F_v(G),y,z)$
from Lemma~\ref{lemma:pr-katz:two-arrow-graphs:pr} or Lemma~\ref{lemma:pr-katz:two-arrow-graphs:katz} (depending on satisfying axiom).
Thus, 
\[
  p_F(F_v(G),y,z) - p_F(F_v(G^r),y,z) < 0,
\]
which contradicts Equation~\eqref{eq:lemma:pr-katz:loops:8}

Therefore, it indeed holds that $\tilde{b}^{0}(v) = \tilde{b}^{*0}(v)$.
As a result, it also holds that
\(
    \tilde{G}^0=\tilde{G}^{*0}.
\)
This implies that
\(
    F_u(\tilde{G}^0)=F_u(\tilde{G}^{*0})
\)
for every $u \in V$.
Thus, by Equation~\eqref{eq:lemma:pr-katz:loops:3} and Equation~\eqref{eq:lemma:pr-katz:loops:5} we obtain that $F_u(G)=F_u(\hat{G})$ for every $u \in V$, which yields the thesis by Equation~\eqref{eq:lemma:pr-katz:loops:2.5}
\end{proof}

\begin{lemma}
\label{lemma:pr:semi-out-regular}
If a centrality measure $F$ defined on $\mathcal{G}$ satisfies LOC, ED, NC, EM, and BL then for every semi-out-regular graph $G=(V,E,b,c)$ we have
$$F_v(G) = PR^{a_F}_v(G) \quad \mbox{for every } v \in V.$$
\end{lemma}
\begin{proof}
From Lemmas~\ref{lemma:pr-katz:two-arrow-graphs:pr} and~\ref{lemma:pr-katz:recursive} we get that for every semi-out-regular graph $G=(V,E,b,c)\in \mathcal{G}$ and node $v \in V$ we have
$$F_v(G) = b(v) + \sum_{(u,v) \in E} a_F \cdot \frac{c(u,v)}{\deg^+_u(G)} \cdot F_u(G).$$
Hence, centrality $F$ satisfies PageRank recursive equation (Equation~\eqref{eq:rec:pr}) with decay parameter $a_F$.
The system of PageRank recursive equations has a unique solution, therefore $F_v(G)=PR^{a_F}_v(G)$ for every semi-out-regular graph $G$ and node $v \in V$.
\end{proof}

\begin{lemma}
\label{lemma:pr:final}
If a centrality measure $F$ defined on $\mathcal{G}$ satisfies LOC, ED, NC, EM, and BL then for every graph $G=(V,E,b,c)$ we have
$$F_v(G) = PR^{a_F}_v(G) \quad \mbox{for every } v \in V.$$
\end{lemma}
\begin{proof}
Take arbitrary $G=(V,E,b,c)$ and divide the weight of each edge by the out-degree of its start, i.e.,
let $G'=(V,E,b,c')$, where $c'(u,v) = c(u,v)/\deg^+_u(G)$ for every $(u,v) \in E$.
From EM we have that $F_v(G')=F_v(G)$ for every $v \in V$.
Observe that $G'$ is semi-out-regular, thus from Lemma~\ref{lemma:pr:semi-out-regular} we know that $F_v(G')=PR^{a_F}_v(G')$ for every $v \in V$.
Since PageRank also satisfies EM (Lemma~\ref{lemma:axioms:pr}) we get that $F_v(G)=PR^{a_F}_v(G)$.
\end{proof}

\begin{lemma}
\label{lemma:katz:semi-out-regular}
If a centrality measure $F$ defined on $\mathcal{G}^{K(a)}$ satisfies LOC, ED, NC, EC, and BL then for every semi-out-regular graph $G=(V,E,b,c) \in \mathcal{G}^{K(a)}$  we have
$$F_v(G) = K^{a_F}_v(G) \quad \mbox{for every } v \in V.$$
\end{lemma}
\begin{proof}
From Lemma~\ref{lemma:pr-katz:two-arrow-graphs:katz} and Lemma~\ref{lemma:pr-katz:recursive} we obtain that for every semi-out-regular graph $G=(V,E,b,c)\in \mathcal{G}^{K(a)}$ and every $v \in V$ we have
$$F_v(G) = b(v) + \sum_{(u,v) \in E} a_F \cdot c(u,v) \cdot F_u(G).$$
This means that centrality $F$ satisfies Katz recursive equation (Equation~\eqref{eq:rec:katz}) with decay parameter $a_F$.
Since the system of Katz recursive equations has a unique solution, we obtain that $F_v(G)=K^{a_F}_v(G)$ for every semi-out-regular graph $G$ and node $v \in V$.
\end{proof}

\begin{lemma}
\label{lemma:katz:final}
If a centrality measure $F$ defined on $\mathcal{G}^{K(a)}$ satisfies LOC, ED, NC, EC, and BL then for every graph $G=(V,E,b,c) \in \mathcal{G}^{K(a)}$  we have
$$F_v(G) = K^{a_F}_v(G) \quad \mbox{for every } v \in V.$$
\end{lemma}
\begin{proof}
We will say that node $v \in V$ is a \emph{leaf}, if it does not have any outgoing edges and exactly one incoming edge.
Let $V^L = \{ v \in V : \Gamma^+_v(G)=\emptyset \land |\Gamma^-_v(G)|=1\}$ be the set of all leafs in a graph.
The \emph{parent} of leaf $v$ is a node, $p(v)$, that has an outgoing edge to $v$, i.e., $(p(v),v) \in E$.
We will denote the set of all parents of leafs by $V^P = \{p(v) : v \in V^L\}$.
Intuitively, the out-degree of nodes that are parents of leafs can be arbitrarily increased by multiplying the weight of incoming edge of the leaf.
Thus, based on EC we can level the out-degree of all parents of leafs.
Hence, if all nodes were either leaf or parents of leafs, we would be able to transform the graph into semi-out-regular one.
Therefore, the main obstacle are nodes that are neither leaf nor a parent of a leaf.
We will call such nodes \emph{ordinary} and denote the set of all ordinary nodes by $V^O = V \setminus (V^L \cup V^P)$.
We will prove the thesis by the induction on their number, i.e., $|V^O|$.

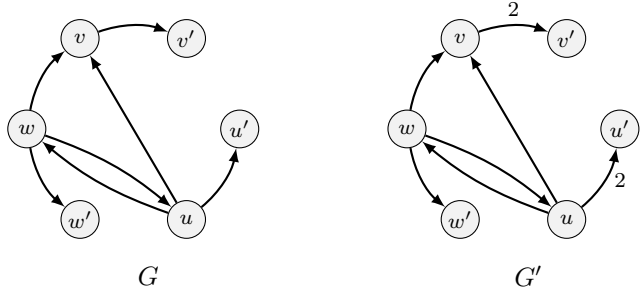
\begin{figure}[t]
\centering
\begin{tikzpicture}
  \def\sx{0.7cm} 
  \def\sy{0.6cm} 
  \def\x{0cm} 
  \def\y{0cm} 
  \def\arrdist{0.35cm}

  \tikzset{
    node_blank/.style={circle,draw,minimum size=0.5cm,inner sep=0, color=white}, 
    node/.style={circle,draw,minimum size=0.5cm,inner sep=0, fill = black!05, font=\footnotesize}, 
    node_emph/.style={circle, minimum size=0.65cm, inner sep=0, fill = black!15, font=\footnotesize}, 
    edge/.style={-latex,above,font=\footnotesize}, 
    el/.style={below,font=\footnotesize}, 
    operation/.style={sloped,>=stealth,above,font=\footnotesize},
    arrow/.style={draw, single arrow, minimum width = 0.9cm, minimum height=\x-6*\s+\s, fill=black!10},
    blank/.style={}
  } 
    
  \node[node] (p) at (\x+1*\sx, 4*\sy + \y) {$v$};
  \node[node] (l) at (\x+3*\sx, 4*\sy + \y) {$v'$};
  \node[node] (s) at (\x+4*\sx, 2*\sy + \y) {$u'$};
  \node[node] (u) at (\x+3*\sx, 0*\sy + \y) {$u$};
  \node[node] (v) at (\x+1*\sx, 0*\sy + \y) {$w'$};
  \node[node] (w) at (\x+0*\sx, 2*\sy + \y) {$w$};
  \node[blank] (blank) at (\x+2*\sx+0.2cm, -0.75cm + \y) {$G$};
  
  \path[->,draw,thick]
  (u) edge[edge, bend left=-20, looseness = 0.9] (s)
  (u) edge[edge] (p)
  (u) edge[edge, bend left=10] (w)
  (w) edge[edge, bend left=-20, looseness = 0.9] (v)
  (w) edge[edge, bend left=20, looseness = 0.9] (p)
  (w) edge[edge, bend left=10] (u)
  (p) edge[edge, bend left=20, looseness = 0.9] (l)
  ;

  \def\x{5cm} 
  
  \node[node] (p) at (\x+1*\sx, 4*\sy + \y) {$v$};
  \node[node] (l) at (\x+3*\sx, 4*\sy + \y) {$v'$};
  \node[node] (s) at (\x+4*\sx, 2*\sy + \y) {$u'$};
  \node[node] (u) at (\x+3*\sx, 0*\sy + \y) {$u$};
  \node[node] (v) at (\x+1*\sx, 0*\sy + \y) {$w'$};
  \node[node] (w) at (\x+0*\sx, 2*\sy + \y) {$w$};
  \node[blank] (blank) at (\x+2*\sx+0.2cm, -0.75cm + \y) {$G'$};
  
  \path[->,draw,thick]
  (u) edge[edge, bend left=-20, looseness = 0.9] node[right] {$2$} (s)
  (u) edge[edge] (p)
  (u) edge[edge, bend left=10] (w)
  (w) edge[edge, bend left=-20, looseness = 0.9] (v)
  (w) edge[edge, bend left=20, looseness = 0.9] (p)
  (w) edge[edge, bend left=10] (u)
  (p) edge[edge, bend left=20, looseness = 0.9] node[above] {$2$}  (l)
  ;

\end{tikzpicture}
\caption{An illustration to the basis of the induction in the proof of Lemma~\ref{lemma:katz:final}.
The weights of edges that do not have weight equal to $1$ are shown.
In an example graph, $G$, all nodes are either leafs, like $u',v'$, and $w'$, or parents of leafs, like $u,v$, and $w$.
Multiplying the weights of edges between leafs and their parents we can obtain a semi-out-regular graph, $G'$.}
\label{fig:lemma:katz:final1}
\end{figure}

If $|V^O|=0$, then each node is either a leaf or a parent of one, i.e., $V^L \cup V^P = V$.
Let us denote the maximal out-degree of all nodes by $x = \max_{v \in V} \deg^+_v(G)$.
In order to transform graph $G$ into semi-out-regular graph, we will scale the weights of edges from parents to leafs in such a way, that all parents have out-degree $x$.
More in detail, for every $v \in V^P$ such that $\deg^+_v(G) < x$ let us consider its leaf, $u$, and multiply the weights of $u$ and edge $(v,u)$ by a constant such that the new weight of $(u,v)$ is equal to $x - (\deg^+_v(G) - c(v,u))$.
In this way, in the new graph, $G'$, node $v$ will have out-degree equal to $x$
(see Figure~\ref{fig:lemma:katz:final1}).
Formally, let $G'=(V,E,b',c')$ where $b'(v) = b(v) \cdot (x - (\deg^+_{p(v)}(G)-c(p(v),v)))/c(p(v),v)$ for every $v \in V^L$ and $b'(v)=b(v)$ for every $v \in V^P$ while $c'(u,v) = x - (\deg^+_u(G) - c(u,v))$ if $v \in V^L$ and $c'(u,v)=c(u,v)$ otherwise.
From EC we know that
\begin{equation}
    \label{eq:lemma:katz:final:1}
    F_v(G') =
    \begin{cases}
        F_v(G) \cdot \left( 1+ \frac{x - \deg^+_{p(v)}(G)}{c(p(v),v)}\right) & \mbox{if } v \in V^L, \\
        F_v(G) & \mbox{otherwise.}
    \end{cases}
\end{equation}
Observe that in graph $G'$ every node is either a leaf, thus does not have any outgoing edges, or its out-degree is equal to $x$.
Thus, graph $G'$ is semi-out-regular.
Hence, from Lemma~\ref{lemma:katz:semi-out-regular} we know that $F_v(G')=K^{a_F}_v(G')$.
Since Katz centrality satisfies EC (Lemma~\ref{lemma:axioms:katz}), we get the thesis from Equation~\eqref{eq:lemma:katz:final:1}.

Let us move to the case in which $|V^O| > 0$.
Then, let us take node $v \in V^O$ such that the number of its successors that are ordinary and are not $v$, i.e., $|S(v) \cap V^O \setminus \{v\}|$, is minimal.
Observe that for every $u \in S(v)$ we have that $S(u) \subseteq S(v)$.
Hence, if $u$ is ordinary, then $S(u) = S(v)$.
Otherwise, it would mean that $|S(u) \cap V^O \setminus \{u\}| < |S(v) \cap V^O \setminus \{v\}|$ since on the right hand side we count each node that we count on the left hand side and also node $u$.
In result, we obtain two cases:
first, if $v \not \in S(v)$, then there are no ordinary successors of $v$, i.e., $S(v) \cap V^O =\emptyset$ (I);
and second, if $v \in S(v)$, then all ordinary successors of $v$ belong to the same strongly connected component as $v$, i.e.,
for every $u,w \in S(v) \cap V^O$ we have $u \in S(w) \cap P(w)$ (II).

(I) Let us begin with the case in which $v \not \in S(v)$ and all successors of $v$ are either leaf or parents of some leafs, i.e., $S(v) \subseteq V^P \cup V^L$.
Then, let us denote the maximal out-degree of successor of $v$ by $x = \max_{u \in S(v)} \deg^+_v(G)$.
In the same way as in case of $k_G=0$, we increase the weight of each edge from a parent to a leaf in $S(v)$ so that all successors of $v$ have out-degree $x$.
Formally, let $G'=(V,E,b',c')$ where $b'(u) = b(u) \cdot (1+((x - \deg^+_{p(u)}(G))/c(p(u),u))$ if $v \in V^L \cap S(v)$ and $b'(u)=b(u)$ otherwise,
while $c'(u,w) = x - \deg^+_u(G) + c(u,w)$ if $w \in V^L \cap S(v)$ and $c'(u,v)=c(u,v)$ otherwise.
From EC we know that
\begin{equation}
    \label{eq:lemma:katz:final:2}
    F_u(G') =
    \begin{cases}
        \! F_u(G) \!\cdot \!\left(\! 1+ \frac{x - \deg^+_{p(u)}(G)}{c(p(u),u)} \! \right) & \!\!\mbox{if } u \in V^L \cap S(v), \\
        \! F_u(G) & \!\!\mbox{otherwise.}
    \end{cases}
\end{equation}
Next, let us multiply the weight of the outgoing edges of $v$ by $x / \deg^+_v(G)$ and divide its weight and weight of its incoming edges by $x / \deg^+_v(G)$.
Formally, let $G''=(V,E,b'',c'')$ where $b''(v)=b'(v) \cdot \deg^+_v(G) / x$ and $b''(u)=b'(u)$ for every $u \in V \setminus \{v\}$ while $c''(e)=c'(e)\cdot \deg^+_v(G) / x$ if $e \in \Gamma^-_v(G)$, $c''(e)=c'(e)\cdot x/ \deg^+_v(G)$ if $e \in \Gamma^+_v(G)$, and $c''(e)=c'(e)$ otherwise.
See Figure~\ref{fig:lemma:katz:final2} for an illustration.
Again from EC we get that
\begin{equation}
    \label{eq:lemma:katz:final:3}
    F_u(G'') =
    \begin{cases}
        F_v(G') \cdot \frac{\deg^+_{v}(G)}{x} & \mbox{if } u = v, \\
        F_u(G') & \mbox{otherwise.}
    \end{cases}
\end{equation}
Observe that incoming edges of $v$ does not come from successors of $v$, because $v \not \in S(v)$.
Thus, the successors of $v$ have the same out-degree in $G''$ as in $G'$, i.e., $\deg^+_u(G'')=x$ for every $u \in S(v) \cap V^P$.
Moreover, $\deg^+_v(G'')=x$ as well.
We will use this fact to add a leaf to node $v$ using Lemma~\ref{lemma:pr-katz:nc-in-sor}a.

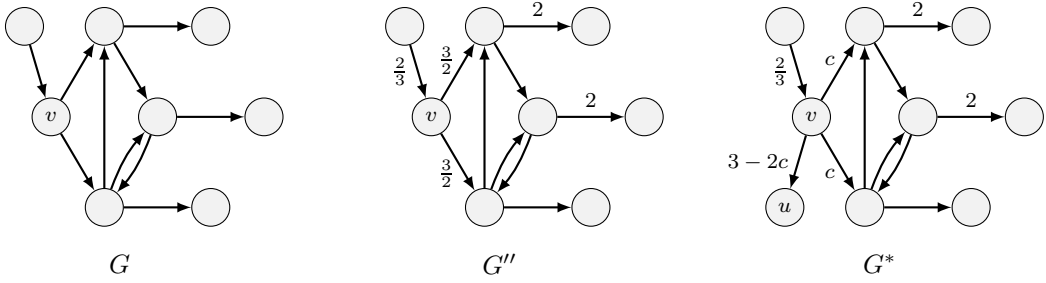
\begin{figure}[t]
\centering
\begin{tikzpicture}
  \def\sx{0.7cm} 
  \def\sy{0.6cm} 
  \def\x{0cm} 
  \def\y{0cm} 
  \def\arrdist{0.35cm}

  \tikzset{
    node_blank/.style={circle,draw,minimum size=0.5cm,inner sep=0, color=white}, 
    node/.style={circle,draw,minimum size=0.5cm,inner sep=0, fill = black!05, font=\footnotesize}, 
    node_emph/.style={circle, minimum size=0.65cm, inner sep=0, fill = black!15, font=\footnotesize}, 
    edge/.style={-latex,above,font=\footnotesize}, 
    el/.style={below,font=\footnotesize}, 
    operation/.style={sloped,>=stealth,above,font=\footnotesize},
    arrow/.style={draw, single arrow, minimum width = 0.9cm, minimum height=\x-6*\s+\s, fill=black!10},
    blank/.style={}
  } 
    
  \node[node] (a) at (\x+1*\sx, 4*\sy + \y) {};
  \node[node] (b) at (\x+2*\sx, 2*\sy + \y) {};
  \node[node] (c) at (\x+1*\sx, 0*\sy + \y) {};
  \node[node] (a_) at (\x+3*\sx, 4*\sy + \y) {};
  \node[node] (b_) at (\x+4*\sx, 2*\sy + \y) {};
  \node[node] (c_) at (\x+3*\sx, 0*\sy + \y) {};
  \node[node] (v) at (\x+0*\sx, 2*\sy + \y) {$v$};
  \node[node] (_v) at (\x-0.5*\sx, 4*\sy + \y) {};
  \node[blank] (blank) at (\x+1*\sx+0.2cm, -0.75cm + \y) {$G$};
  
  \path[->,draw,thick]
  (a) edge[edge] (b)
  (b) edge[edge, bend left=10] (c)
  (c) edge[edge, bend left=10] (b)
  (c) edge[edge] (a)
  (a) edge[edge] (a_)
  (b) edge[edge] (b_)
  (c) edge[edge] (c_)
  (v) edge[edge] (c)
  (v) edge[edge] (a)
  (_v) edge[edge] (v)
  ;

  \def\x{5cm} 
  
  \node[node] (a) at (\x+1*\sx, 4*\sy + \y) {};
  \node[node] (b) at (\x+2*\sx, 2*\sy + \y) {};
  \node[node] (c) at (\x+1*\sx, 0*\sy + \y) {};
  \node[node] (a_) at (\x+3*\sx, 4*\sy + \y) {};
  \node[node] (b_) at (\x+4*\sx, 2*\sy + \y) {};
  \node[node] (c_) at (\x+3*\sx, 0*\sy + \y) {};
  \node[node] (v) at (\x+0*\sx, 2*\sy + \y) {$v$};
  \node[node] (_v) at (\x-0.5*\sx, 4*\sy + \y) {};
  \node[blank] (blank) at (\x+1*\sx+0.2cm, -0.75cm + \y) {$G''$};
  
  \path[->,draw,thick]
  (a) edge[edge] (b)
  (b) edge[edge, bend left=10] (c)
  (c) edge[edge, bend left=10] (b)
  (c) edge[edge] (a)
  (a) edge[edge] node[above] {$2$} (a_)
  (b) edge[edge] node[above] {$2$} (b_)
  (c) edge[edge] (c_)
  (v) edge[edge] node[left, pos = 0.7] {$\frac{3}{2}$} (c)
  (v) edge[edge] node[left, pos = 0.7] {$\frac{3}{2}$} (a)
  (_v) edge[edge] node[left] {$\frac{2}{3}$} (v)
  ;

   \def\x{10cm} 

  \node[node] (a) at (\x+1*\sx, 4*\sy + \y) {};
  \node[node] (b) at (\x+2*\sx, 2*\sy + \y) {};
  \node[node] (c) at (\x+1*\sx, 0*\sy + \y) {};
  \node[node] (a_) at (\x+3*\sx, 4*\sy + \y) {};
  \node[node] (b_) at (\x+4*\sx, 2*\sy + \y) {};
  \node[node] (c_) at (\x+3*\sx, 0*\sy + \y) {};
  \node[node] (v) at (\x+0*\sx, 2*\sy + \y) {$v$};
  \node[node] (_v) at (\x-0.5*\sx, 4*\sy + \y) {};
  \node[node] (u) at (\x-0.5*\sx, 0*\sy + \y) {$u$};
  \node[blank] (blank) at (\x+1*\sx+0.2cm, -0.75cm + \y) {$G^*$};
  
  \path[->,draw,thick]
  (a) edge[edge] (b)
  (b) edge[edge, bend left=10] (c)
  (c) edge[edge, bend left=10] (b)
  (c) edge[edge] (a)
  (a) edge[edge] node[above] {$2$} (a_)
  (b) edge[edge] node[above] {$2$} (b_)
  (c) edge[edge] (c_)
  (v) edge[edge] node[left, pos = 0.7] {$c$} (c)
  (v) edge[edge] node[left, pos = 0.7] {$c$} (a)
  (_v) edge[edge] node[left] {$\frac{2}{3}$} (v)
  (v) edge[edge] node[left] {$3-2c$} (u)
  ;

\end{tikzpicture}
\caption{An illustration to the case (I) in the proof of Lemma~\ref{lemma:katz:final}.
The weights of edges that do not have weight equal to $1$ are shown.
In an example graph, $G$, all of the successors of node $v$ are either leafs or parents of leafs.
Graph $G''$ is obtained from $G$ by multiplying the weights of edges from parents to leafs in such a way that out-degree of each parent is equal to $x=3$.
Also, the weights of the incoming and outgoing edges of $v$ are scaled.
In graph $G^*$ a leaf is added to node $v$.}
\label{fig:lemma:katz:final2}
\end{figure}

To this end, let us consider nodes $u',v' \not \in V$ and add to $G''$ a simple graph that consists of nodes $u'$ and $v'$ connected by an edge  (see Figure~\ref{fig:lemma:katz:final2}).
Formally, let 
\[
  G^\dagger = G'' + (\{v',u'\},\{(v',u')\},\textbf{1}, c^\dagger),
\]
where
$\textbf{1}(u')=\textbf{1}(v')=1$ and $c^\dagger(u',v')=x$.
Now, let us combine node $v'$ into $v$ in graph $G^\dagger$, i.e., let $G^\ddagger = C^F_{v' \rightarrow v}(G^\dagger)$.
From Lemma~\ref{lemma:pr-katz:source-node} we get that $F_{v'}(G^\dagger)=1$.
Moreover, we have that $\deg^+_v(G^\dagger)=\deg^+_{v'}(G^\dagger)=\deg^+_s(G^\dagger)=x$ for any $s \in S(v) \setminus V^L$.
Thus, from Lemma~\ref{lemma:pr-katz:nc-in-sor}a and LOC we get that
\begin{equation}
    \label{eq:lemma:katz:final:4}
    F_u(G^\ddagger) =
    \begin{cases}
        F_v(G'') + 1 & \mbox{if } u = v, \\
        F_u(G'') & \mbox{otherwise.}
    \end{cases}
\end{equation}
On the other hand, in graph $G^\dagger$ node $v$ is not ordinary anymore and no ordinary node was added.
Thus, from inductive assumption we have that $F_u(G^\ddagger)=K^{a_F}_u(G^\ddagger)$ for every $u \in V$.
Since Katz centrality satisfies our axioms (Lemma~\ref{lemma:axioms:katz}), from Equations~\eqref{eq:lemma:katz:final:2}-\eqref{eq:lemma:katz:final:4} we obtain that $F_u(G)=K^{a_F}_u(G)$ for every $u \in V$.

(II) Now, let us move to the case in which $v \in S(v)$ and all ordinary successors of $v$ belong to the same strongly connected component, i.e., for every $u,w \in S(v) \cap V^O$ we have $u \in S(w) \cap P(w)$.
Let us denote all nodes in this strongly connected components as $U = S(v) \cap P(v)$.
Also, let us denote their outgoing edges as $E^U = \{(u,w) \in E : u \in U\}$ and the nodes that they go to and that are not in $U$, i.e., let $U^+ = \{w : (u,w) \in E^U \} \setminus U$.
Let us first assume that $U^+ \neq \emptyset$ and relax this assumption at the end of the proof.
Finally, as before, let us denote the maximal out-degree of a successor of $v$ by $x = \max_{u \in S(v)} \deg^+_u(G)$.

Now, consider an auxiliary graph in which all nodes except for $U$ and $U^+$ are removed.
Formally, let $G^\dagger = (U \cup U^+, E^U, b_{U \cup U^+}, c_{E^U})$.
Observe that since all outgoing edges of $U$ remains, nodes $U$ in graph $G^\dagger$ still constitute a strongly connected component.
In order to make whole graph strongly connected let us add outgoing edges of nodes $U^+$.
More in detail, for each $u \in U^+$ let us add edge $(u,u)$ with weight $x$ and edge $(u,v)$ with weight $1$.
Formally, let $E^+ = \{ (u,u),(u,v) : u \in U^+\}$ and let $G^\circ = (V^\circ, E^\circ, b^\circ, c^\circ)$, where
$V^\circ = U \cup U^+$,
$E^\circ = E^U \cup E^+$,
$b^\circ = b_{U \cup U^+}$,
$c^\circ(u,u)= x$ and $c^\circ(u,v) = 1$ for every $u \in U^+$, and
$c^\circ_{E^U} =  c_{E^U}$
(see Figure~\ref{fig:lemma:katz:final3}).
Observe that $G^\circ$ is indeed strongly connected.

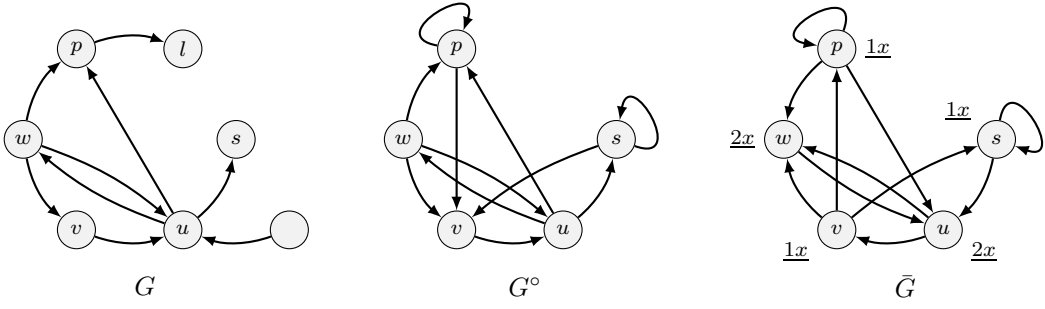
\begin{figure}[t]
\centering
\begin{tikzpicture}
  \def\sx{0.7cm} 
  \def\sy{0.6cm} 
  \def\x{0cm} 
  \def\y{0cm} 
  \def\arrdist{0.35cm}

  \tikzset{
    node_blank/.style={circle,draw,minimum size=0.5cm,inner sep=0, color=white}, 
    node/.style={circle,draw,minimum size=0.5cm,inner sep=0, fill = black!05, font=\footnotesize}, 
    node_emph/.style={circle, minimum size=0.65cm, inner sep=0, fill = black!15, font=\footnotesize}, 
    edge/.style={-latex,above,font=\footnotesize}, 
    el/.style={below,font=\footnotesize}, 
    operation/.style={sloped,>=stealth,above,font=\footnotesize},
    arrow/.style={draw, single arrow, minimum width = 0.9cm, minimum height=\x-6*\s+\s, fill=black!10},
    blank/.style={}
  } 
    
  \node[node] (p) at (\x+1*\sx, 4*\sy + \y) {$p$};
  \node[node] (l) at (\x+3*\sx, 4*\sy + \y) {$l$};
  \node[node] (s) at (\x+4*\sx, 2*\sy + \y) {$s$};
  \node[node] (u) at (\x+3*\sx, 0*\sy + \y) {$u$};
  \node[node] (v) at (\x+1*\sx, 0*\sy + \y) {$v$};
  \node[node] (w) at (\x+0*\sx, 2*\sy + \y) {$w$};
  \node[node] (0) at (\x+5*\sx, 0*\sy + \y) {};
  \node[blank] (blank) at (\x+2*\sx+0.2cm, -0.75cm + \y) {$G$};
  
  \path[->,draw,thick]
  (u) edge[edge, bend left=-20, looseness = 0.9] (s)
  (u) edge[edge] (p)
  (u) edge[edge, bend left=10] (w)
  (v) edge[edge, bend left=-20, looseness = 0.9] (u)
  (w) edge[edge, bend left=-20, looseness = 0.9] (v)
  (w) edge[edge, bend left=20, looseness = 0.9] (p)
  (w) edge[edge, bend left=10] (u)
  (p) edge[edge, bend left=20, looseness = 0.9] (l)
  (0) edge[edge, bend left=20, looseness = 0.9] (u)
  ;

  \def\x{5cm} 
  
  \node[node] (p) at (\x+1*\sx, 4*\sy + \y) {$p$};
  \node[node] (s) at (\x+4*\sx, 2*\sy + \y) {$s$};
  \node[node] (u) at (\x+3*\sx, 0*\sy + \y) {$u$};
  \node[node] (v) at (\x+1*\sx, 0*\sy + \y) {$v$};
  \node[node] (w) at (\x+0*\sx, 2*\sy + \y) {$w$};
  \node[blank] (blank) at (\x+2*\sx+0.2cm, -0.75cm + \y) {$G^\circ$};
  
  \path[->,draw,thick]
  (u) edge[edge, bend left=-20, looseness = 0.9] (s)
  (u) edge[edge] (p)
  (u) edge[edge, bend left=10] (w)
  (v) edge[edge, bend left=-20, looseness = 0.9] (u)
  (w) edge[edge, bend left=-20, looseness = 0.9] (v)
  (w) edge[edge, bend left=20, looseness = 0.9] (p)
  (w) edge[edge, bend left=10] (u)
  (p) edge[edge, out=170, in=70, looseness = 5] (p)
  (p) edge[edge] (v)
  (s) edge[edge, out=-20, in=80, looseness = 5] (s)
  (s) edge[edge, bend left=-10] (v)
  ;

   \def\x{10cm} 

  \node[node, label={0:\footnotesize \underline{$1x$}}] (p) at (\x+1*\sx, 4*\sy + \y) {$p$};
  \node[node, label={150:\footnotesize \underline{$1x$}}] (s) at (\x+4*\sx, 2*\sy + \y) {$s$};
  \node[node, label={350:\footnotesize \underline{$2x$}}] (u) at (\x+3*\sx, 0*\sy + \y) {$u$};
  \node[node, label={190:\footnotesize \underline{$1x$}}] (v) at (\x+1*\sx, 0*\sy + \y) {$v$};
  \node[node, label={180:\footnotesize \underline{$2x$}}] (w) at (\x+0*\sx, 2*\sy + \y) {$w$};
  \node[blank] (blank) at (\x+2*\sx+0.2cm, -0.75cm + \y) {$\bar{G}$};
  
  \path[->,draw,thick]
  (s) edge[edge, bend left=20, looseness = 0.9] (u)
  (p) edge[edge] (u)
  (w) edge[edge, bend left=-10] (u)
  (u) edge[edge, bend left=20, looseness = 0.9] (v)
  (v) edge[edge, bend left=20, looseness = 0.9] (w)
  (p) edge[edge, bend left=-20, looseness = 0.9] (w)
  (u) edge[edge, bend left=-10] (w)
  (p) edge[edge, in=170, out=70, looseness = 5] (p)
  (v) edge[edge] (p)
  (s) edge[edge, in=-20, out=80, looseness = 5] (s)
  (v) edge[edge, bend left=10] (s)
  ;

\end{tikzpicture}
\caption{An illustration to the first part of the case (II) in the proof of Lemma~\ref{lemma:katz:final}.
All edges have weights equal to $1$.
In an example graph, $G$, nodes $v$, $u$ and $w$ constitute a strongly connected component, $U$.
Nodes $l, p, s$ are the successors of $v$ that do not belong to this component and all of them are either leafs or parents of leafs.
Graph $G^\circ$ is a strongly connected graph constructed from the nodes in $U$ and their direct successors.
Maximal out-degree of a successor of $v$ that is not in $U$, i.e., $x$, is equal to $1$, thus loops around nodes $p$ and $s$ have weight $1$.
Graph $\bar{G}$ is an opposite graph to $G^\circ$.
Eigenvector centrality of every node in this graph is shown.
Note that $\lambda = 2$.}
\label{fig:lemma:katz:final3}
\end{figure}

Since $G^\circ$ is strongly connected, we can make it out-regular using the same method that we used in Lemma~\ref{lemma:ev:all-graphs}.
Formally, let us consider the opposite graph to $G^\circ$, i.e., graph $\bar{G}=(V^\circ,\bar{E},b^\circ,\bar{c})$ such that $\bar{E}= \{(u,w) : (w,u) \in E^\circ\}$ and $\bar{c}(u,w) = c(w,u)$ for every $(w,u) \in E^\circ$.
Now, in graph $\bar{G}$ let us multiply the weight of outgoing edges of node $u \in V^\circ$ by $EV_u(\bar{G})$ and divide the weights of its incoming edges as well as its own weight by $EV_u(\bar{G})$.
Because Eigenvector centrality satisfy EC we know that in this way Eigenvector centrality of other nodes in $V^\circ$ does not change and Eigenvector centrality of $u$ becomes 1.
If we proceed with this operation for all nodes in $V^\circ$ we obtain graph $\bar{G}'$ in which all nodes have Eigenvector centrality equal to 1.
Formally, let $\bar{G}' = (V^\circ,\bar{E},b',\bar{c}')$ where $b'(u)=b^\circ(u)/EV_u(\bar{G})$ for every $u \in V^\circ$ and $\bar{c}'(u,w)=\bar{c}(u,w) \cdot EV_u(\bar{G})/EV_w(\bar{G})$.
Observe that if all nodes in a graph $\bar{G}'$ have equal Eigenvector centrality, then from Eigenvector centrality recursive equation (Equation~\eqref{eq:rec:ev}) we get that in-degree of each node is equal, i.e., there exist $\lambda$ such that $\deg^-_u(\bar{G}')=\lambda$ for every $u \in V$.
Moreover, observe that for any $u \in U^+$ we have that $\bar{c}'(u,u)= \bar{c}(u,u) \cdot EV_u(\bar{G})/EV_u(\bar{G}) = c^\circ(u,u)=x$.
Thus, $\lambda > x$.

Now, let us copy this operation on the original graph $G$ to obtain equal out-degrees of nodes in $U$.
At the same time, we want to make sure that we will not increase out-degree of nodes in $U^+$ too much (as we cannot decrease their out-degree using leafs, only increase).
To this end, let us take an arbitrary constant $y \in \mathbb{R}_{>0}$ by which we will multiply the incoming edges and divide outgoing edges of all nodes in $V^\circ$.
Formally, let us define graph
$G'= (C, E, b', c')$ in which
$b'(u) = b(u) \cdot EV_u(\bar{G})\cdot y$ for every $u \in V^\circ$ and $b'(u)=b(u)$ for every $V \setminus V^\circ$
and 
\[
    c'(u,w) =
    \begin{cases}
        c(u,w) \frac{EV_w(\bar{G})}{EV_u(\bar{G})}    & \mbox{if } u,w \in V^\circ, \\
        c(u,w) / EV_u(\bar{G}) / y            & \mbox{if } u \in V^\circ, w \not \in  V^\circ, \\
        c(u,w) \cdot EV_w(\bar{G}) \cdot y                    & \mbox{if } u \not \in V^\circ, w \in  V^\circ, \\
        c(u,w)                                        & \mbox{if } u,w \not \in V^\circ.
    \end{cases}
\]
See Figure~\ref{fig:lemma:katz:final4} for an illustration.
Observe that for every $(u,w) \in E^U$ we have $c'(u,w) = \bar{c}'(w,u)$, thus indeed $\deg^+_u(G')= \deg^-_u(\bar{G}') = \lambda > x$ for every $u \in U$.
Furthermore, from EC we have that
\begin{equation}
    \label{eq:lemma:katz:final:5}
    F_u(G') =
    \begin{cases}
        F_u(G) \cdot EV_u(\bar{G})\cdot y & \mbox{if } u \in V^\circ, \\
        F_u(G) & \mbox{otherwise.}
    \end{cases}
\end{equation}

\begin{figure}[t]
\centering
\begin{tikzpicture}
  \def\sx{0.7cm} 
  \def\sy{0.6cm} 
  \def\x{0cm} 
  \def\y{0cm} 
  \def\arrdist{0.35cm}

  \tikzset{
    node_blank/.style={circle,draw,minimum size=0.5cm,inner sep=0, color=white}, 
    node/.style={circle,draw,minimum size=0.5cm,inner sep=0, fill = black!05, font=\footnotesize}, 
    node_emph/.style={circle, minimum size=0.65cm, inner sep=0, fill = black!15, font=\footnotesize}, 
    edge/.style={-latex,above,font=\footnotesize}, 
    el/.style={below,font=\footnotesize}, 
    operation/.style={sloped,>=stealth,above,font=\footnotesize},
    arrow/.style={draw, single arrow, minimum width = 0.9cm, minimum height=\x-6*\s+\s, fill=black!10},
    blank/.style={}
  } 
    
  \node[node] (p) at (\x+1*\sx, 4*\sy + \y) {$p$};
  \node[node] (l) at (\x+3*\sx, 4*\sy + \y) {$l$};
  \node[node] (s) at (\x+4*\sx, 2*\sy + \y) {$s$};
  \node[node] (u) at (\x+3*\sx, 0*\sy + \y) {$u$};
  \node[node] (v) at (\x+1*\sx, 0*\sy + \y) {$v$};
  \node[node] (w) at (\x+0*\sx, 2*\sy + \y) {$w$};
  \node[node] (0) at (\x+5*\sx, 0*\sy + \y) {};
  \node[blank] (blank) at (\x+2*\sx+0.2cm, -1cm + \y) {$G'$};
  
  \path[->,draw,thick]
  (u) edge[edge, bend left=-20, looseness = 0.9] node[right] {$\frac{1}{2}$} (s)
  (u) edge[edge] node[right] {$\frac{1}{2}$} (p)
  (u) edge[edge, bend left=10] node[below, pos=0.3] {$1$} (w)
  (v) edge[edge, bend left=-20, looseness = 0.9] node[below] {$2$} (u)
  (w) edge[edge, bend left=-20, looseness = 0.9] node[left] {$\frac{1}{2}$} (v)
  (w) edge[edge, bend left=20, looseness = 0.9] node[left] {$\frac{1}{2}$} (p)
  (w) edge[edge, bend left=10] node[above, pos=0.3] {$1$} (u)
  (p) edge[edge, bend left=20, looseness = 0.9] node[above] {$1$} (l)
  (0) edge[edge, bend left=20, looseness = 0.9] node[below] {$2$} (u)
  ;

  \def\x{5cm} 
  
  \node[node] (p) at (\x+1*\sx, 4*\sy + \y) {$p$};
  \node[node] (l) at (\x+3*\sx, 4*\sy + \y) {$l$};
  \node[node] (s) at (\x+4*\sx, 2*\sy + \y) {$s$};
  \node[node] (u) at (\x+3*\sx, 0*\sy + \y) {$u$};
  \node[node] (v) at (\x+1*\sx, 0*\sy + \y) {$v$};
  \node[node] (w) at (\x+0*\sx, 2*\sy + \y) {$w$};
  \node[node] (0) at (\x+5*\sx, 0*\sy + \y) {};
  \node[blank] (blank) at (\x+2*\sx+0.2cm, -1cm + \y) {$G''$};
  
  \path[->,draw,thick]
  (u) edge[edge, bend left=-20, looseness = 0.9] node[right] {$\frac{1}{2}$} (s)
  (u) edge[edge] node[right] {$\frac{1}{2}$} (p)
  (u) edge[edge, bend left=10] node[below, pos=0.3] {$1$} (w)
  (v) edge[edge, bend left=-20, looseness = 0.9] node[below] {$2$} (u)
  (w) edge[edge, bend left=-20, looseness = 0.9] node[left] {$\frac{1}{2}$} (v)
  (w) edge[edge, bend left=20, looseness = 0.9] node[left] {$\frac{1}{2}$} (p)
  (w) edge[edge, bend left=10] node[above, pos=0.3] {$1$} (u)
  (p) edge[edge, bend left=20, looseness = 0.9] node[above] {$2$} (l)
  (0) edge[edge, bend left=20, looseness = 0.9] node[below] {$2$} (u)
  ;

   \def\x{10cm} 

  \node[node] (p) at (\x+1*\sx, 4*\sy + \y) {$p$};
  \node[node] (l) at (\x+3*\sx, 4*\sy + \y) {$l$};
  \node[node] (s) at (\x+4*\sx, 2*\sy + \y) {$s$};
  \node[node] (u) at (\x+3*\sx, 0*\sy + \y) {$u$};
  \node[node] (v) at (\x+1*\sx, 0*\sy + \y) {$v$};
  \node[node] (w) at (\x+0*\sx, 2*\sy + \y) {$w$};
  \node[node] (0) at (\x+5*\sx, 0*\sy + \y) {};
  \node[node] (u_) at (\x-1*\sx, 0*\sy + \y) {$u'$};
  \node[blank] (blank) at (\x+2*\sx+0.2cm, -1cm + \y) {$G^*$};
  
  \path[->,draw,thick]
  (u) edge[edge, bend left=-20, looseness = 0.9] node[right] {$\frac{1}{2}$} (s)
  (u) edge[edge] node[right] {$\frac{1}{2}$} (p)
  (u) edge[edge, bend left=10] node[below, pos=0.3] {$1$} (w)
  (v) edge[edge, bend left=-20, looseness = 0.9] node[below] {$c$} (u)
  (v) edge[edge, bend left=20, looseness = 0.9] node[below] {$2-c$} (u_)
  (w) edge[edge, bend left=-20, looseness = 0.9] node[left] {$\frac{1}{2}$} (v)
  (w) edge[edge, bend left=20, looseness = 0.9] node[left] {$\frac{1}{2}$} (p)
  (w) edge[edge, bend left=10] node[above, pos=0.3] {$1$} (u)
  (p) edge[edge, bend left=20, looseness = 0.9] node[above] {$2$} (l)
  (0) edge[edge, bend left=20, looseness = 0.9] node[below] {$2$} (u)
  ;

\end{tikzpicture}
\caption{An illustration to the second part of the case (II) in the proof of Lemma~\ref{lemma:katz:final}.
The weight of each edge is shown.
Graph $G'$ is obtained from graph $G$ from Fig.~\ref{fig:lemma:katz:final3} by taking eigenvector centrality of each node in graph $\bar{G}$ and dividing by it the weight of its outgoing edges and multiplying by it its weight and the weights of its incoming edges.
Graph $G''$ is obtained from $G'$ by multiplying the weights of edges from parents to leafs so that the out-degree of a parent is equal to $2$
In graph $G^*$ a leaf is added to node $v$.}
\label{fig:lemma:katz:final4}
\end{figure}
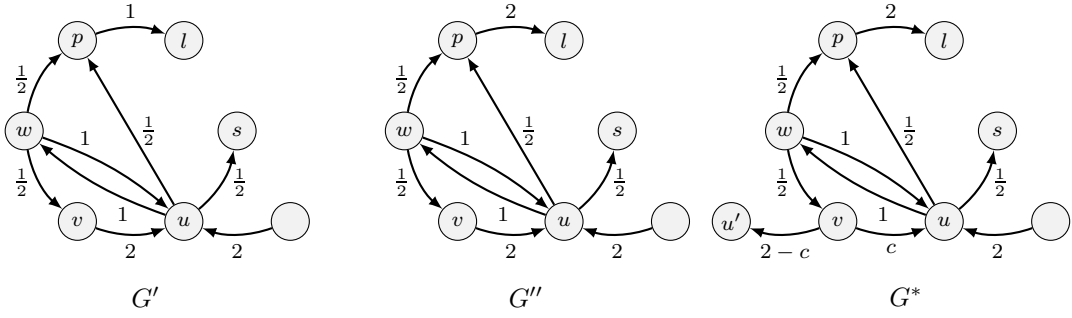

Observe that since for every node $u \in U^+$ outgoing edges of $u$ go to nodes outside of $V^\circ$, we have that $\deg^+_u(G') = \deg^+_u(G) / EV_u(\bar{G}) / y$.
Thus, let us take such $y$ that the maximal out-degree of node in $U^+$ is equal to $x$, i.e., $y = x \cdot \max_{u \in U^+} (\deg^+_u(G) / EV_u(\bar{G}))$.
Then, all successors of $v$ in graph $G'$ have out-degrees equal at most $x$.
Moreover, those that are in a strongly connected component, i.e., nodes in $U$, has out-degree equal to $\lambda > x$, and those that are not, i.e., nodes in $S(v) \setminus U$, are either leafs or parents of leafs and have out-degrees at most $x$.
Thus, let us increase the out-degree of parents of leafs in $S(v) \setminus U$ to $\lambda$ by changing the weight of their edge to a leaf.
Formally, let $G''=(V,E,b'',c'')$ where $b''(u) = b'(u) \cdot (1+((\lambda - \deg^+_{p(u)}(G'))/c'(p(u),u))$ if we have $v \in V^L \cap S(v)$ and $b''(u)=b'(u)$ otherwise, while $c''(u,w) = \lambda - \deg^+_u(G') + c'(u,w)$ if $w \in V^L \cap S(v)$ and $c''(u,v)=c'(u,v)$ otherwise.
From EC we know that
\begin{equation}
    \label{eq:lemma:katz:final:6}
    F_u(G'') \!=\!
    \begin{cases}
        \! F_u(G') \!\cdot \!\left(\! 1+ \frac{\lambda - \deg^+_{p(u)}(G')}{c'(p(u),u)} \!\right)& \!\!\!\mbox{if } u \in\! V^L \!\cap S(v), \\
        \! F_u(G') & \!\!\!\mbox{otherwise.}
    \end{cases}
\end{equation}
Observe that in $G''$ all successors of $v$ are either leafs or have out-degree $\lambda$.
Hence, we will use Lemma~\ref{lemma:pr-katz:nc-in-sor}a to add a leaf to $v$ in a similar way we did it in case (I).
Let us consider nodes $u',v' \not \in V$ and add simple graph of $u'$ and $v'$ connected by an edge to graph $G''$, i.e.,
let $G^\dagger = G'' + (\{v',u'\},\{(v',u')\},\textbf{1}, c^\dagger)$, where $\textbf{1}(u')=\textbf{1}(v')=1$ and $c^\dagger(u',v')=\lambda$.
Now, let us combine node $v'$ into $v$ in graph $G^\dagger$, i.e., let $G^\ddagger = C^F_{v' \rightarrow v}(G^\dagger)$.
From Lemma~\ref{lemma:pr-katz:source-node} we get that $F_{v'}(G^\dagger)=1$.
Moreover, we have that $\deg^+_v(G^\dagger)=\deg^+_{v'}(G^\dagger)=\deg^+_s(G^\dagger)=\lambda$ for any $s \in S(v) \setminus V^L$.
Thus, from Lemma~\ref{lemma:pr-katz:nc-in-sor}a and LOC we get that
\begin{equation}
    \label{eq:lemma:katz:final:7}
    F_u(G^\ddagger) =
    \begin{cases}
        F_v(G'') + 1 & \mbox{if } u = v, \\
        F_u(G'') & \mbox{otherwise.}
    \end{cases}
\end{equation}
On the other hand, in graph $G^\dagger$ node $v$ is not ordinary anymore and no ordinary node was added.
Thus,from inductive assumption we have that $F_u(G^\ddagger)=K^{a_F}_u(G^\ddagger)$ for every $u \in V$.
Since Katz centrality satisfies our axioms (Lemma~\ref{lemma:axioms:katz}), from Equations~\eqref{eq:lemma:katz:final:5}-\eqref{eq:lemma:katz:final:7} we get that $F_u(G)=K^{a_F}_u(G)$ for every $u \in V$.

It remains to consider the case in which $U^+ = \emptyset$, i.e., there are no successors of $v$ that are not predecessors of $v$ as well, i.e., $S(v) \subseteq P(v)$.
In such a case the situation is simpler.
We do not need to make sure that $\lambda$ of graph $\bar{G}$ is greater then $x$, and we do not need to change the out-degrees of successors of $v$ that are not its predecessors (the change from graph $G'$ into $G''$) since there are none.
Apart from that the proof is analogous.
\end{proof}
\end{document}